\documentclass[letterpaper,11pt]{article}
\usepackage{fullpage}
\usepackage{graphicx}
\usepackage{comment}
\usepackage{mathtools}
\usepackage[linesnumbered,ruled,vlined]{algorithm2e}
\usepackage[utf8]{inputenc}
\usepackage{tikz}
\usepackage{amsthm}
\usepackage{subcaption}
\usepackage{amsfonts}
\usepackage{hyperref}
\usepackage[framemethod=TikZ]{mdframed}
\usepackage{amsthm}
\usepackage{thmtools}
\usepackage{thm-restate}
\usepackage[normalem]{ulem}
\usepackage{float}
\usepackage{multirow}

\newtheorem{remark}{Remark}
\newtheorem{theorem}{Theorem}
\newtheorem{lemma}[theorem]{Lemma}

\newtheorem{observation}{Observation}

\newcommand{\llcs}{\ifmmode\mathsf{lcs}\else\textsf{lcs}\fi}

\newcommand{\lisplus}{\mathsf{LIS}^+}

\definecolor{bostonuniversityred}{rgb}{0.8, 0.0, 0.0}
\newcounter{theo}[section] 

\newcommand{\Michael}[1]{}\newcommand{\Saeed}[1]{}
\begin{document}
	
	\title{Dynamic Longest Increasing Subsequence and the Erd\"{o}s-Szekeres Partitioning Problem\footnote{The results of this manuscript appeared in preliminary versions~\cite{our-stoc-paper} (STOC'20) and~\cite{mitzenmacher2020erd}.}}
\author{
	Michael Mitzenmacher\thanks{Toyota Technological Institute at Chicago}
	\and Saeed Seddighin\footnotemark[1]
}
	\date{}
	
	\maketitle
	
	\begin{abstract}
		% Longest increasing subsequence (\textsf{LIS}) is one of the most intriguing problems in combinatorial optimization. On one hand, its simplicity makes it a perfect motivating example for dynamic programming and binary-search. On the other hand, the problem is so hard to tackle in that all we know about it is limited to the $O(n \log n)$ time algorithm via dynamic programming and binary-search. In terms of approximation, the problem has remained relatively untouched after quite a long time. The only known approximate algorithms work only when the gap between the solution size and the input size is small. Distance to monotonicity (\textsf{DTM}) is the complement of \textsf{LIS}.
\setcounter{page}{0}
In this paper, we provide new approximation algorithms for dynamic variations of the longest increasing subsequence (\textsf{LIS}) problem, and the complementary distance to monotonicity (\textsf{DTM}) problem.
In this setting, operations of the following form arrive sequentially:
(i) add an element, (ii) remove an element, or (iii) substitute an element for another.  At every point in time, the algorithm has an approximation to the longest increasing subsequence (or distance to monotonicity).  We present a $(1+\epsilon)$-approximation algorithm for \textsf{DTM} with polylogarithmic worst-case  update time and a constant factor approximation algorithm for \textsf{LIS} with worst-case update time $\tilde O(n^\epsilon)$ for any constant $\epsilon > 0$.% $n$ in the runtime denotes the size of the array at the time the operation arrives.

Our dynamic algorithm for \textsf{LIS} leads to an almost optimal algorithm for the Erd\"{o}s-Szekeres partitioning problem. Erd\"{o}s-Szekeres partitioning problem was introduced by Erd\"{o}s and Szekeres in 1935 and was known to be solvable in time $O(n^{1.5}\log n)$. Subsequent work improve the runtime to $O(n^{1.5})$ only in 1998. Our dynamic \textsf{LIS} algorithm leads to a solution for Erd\"{o}s-Szekeres partitioning problem with runtime $\tilde O_{\epsilon}(n^{1+\epsilon})$ for any constant $\epsilon > 0$.

\thispagestyle{empty}

	\end{abstract}
\section{Introduction}\label{sec:intro}
Longest increasing subsequence (\textsf{LIS})  is one of the oldest problems in computer science. Given an array $a = \langle a_1, a_2, \ldots, a_n\rangle$ of size $n$, \textsf{LIS} is defined as the largest subset of the elements whose values are strictly increasing in the order of their indices. Distance to monotonicity (\textsf{DTM}) is the dual problem. For \textsf{DTM}, we wish to remove the smallest number of elements such that the remaining subsequence is increasing. \textsf{LIS} and \textsf{DTM} are special cases of the celebrated \textit{edit distance} and \textit{longest common subsequence} problems when the input strings are permutations.
 
The classic \textit{patience sorting} solution for \textsf{LIS}  utilizes dynamic programming and binary search to solve \textsf{LIS}
exactly in time $O(n \log n)$.  (In what follows, when we refer to a solution, we typically refer to the size of the \textsf{LIS}, but
also the corresponding increasing subsequence can be found in time proportional to its size.)  
Matching lower bounds ($\Omega(n \log n)$) are known for comparison-based algorithms~\cite{fredman1975computing} and solutions based on algebraic decision trees~\cite{ramanan1997tight}.
For approximation algorithms, for any $\epsilon > 0$, a multiplicative $O(n^\epsilon)$ approximate solution can be determined in truly sublinear time via random sampling\footnote{For an $O(n^\epsilon)$ approximation algorithm, one can sample $O(n^{1-\epsilon})$ many elements from the array and report the \textsf{LIS} of those samples.}. 
Surprisingly, not much is known that improves upon this algorithm generally, although when $n/\mathsf{LIS}(a)$ is subpolynomial  we can obtain better approximation guarantees  for \textsf{LIS}~\cite{saeedfocs19,saks2010estimating}.
% Let $D$ be an array of size $n+1$ whose $i$'th element denotes the smallest element of $a$ such that an increasing subsequence of length $i$ ends at that element. Starting from $D = \infty^{n+1}$, one can iterate over the elements of $a$ and update $D$ accordingly. It is an easy exercise to show that $D$ is increasing and that a binary search suffices to update $D$ in each iteration.

% Despite its simplicity, this is the best algorithm we are aware of for \textsf{LIS} and unlike similar classic problems such as edit distance or longest common subsequence, the problem does not become easier even when the solution size is small. 
From a complexity point of view, unconditional lower bounds apply to \textsf{LIS}. For instance, any algorithm that obtains an $f(n)$ approximate solution for \textsf{LIS} has to make at least $n/(f(n)+1)$ value queries\footnote{A value query provides an $i$ as input and asks for the value of $a_i$.} to the elements of $a$ to distinguish the case that $a$ is decreasing from the case that $a$ has an increasing subsequence of length at least $f(n)+1$. Thus a subpolynomial approximation algorithm for \textsf{LIS} in truly sublinear time is not possible in general. 
% However, this does not justify our embarrassingly limited knowledge of \textsf{LIS}. 
Even if we are guaranteed that the solution size is  $\Theta(n^{1-\epsilon})$ (a setting for which the known complexity lower bounds do not apply), we are not aware of any subpolynomial approximate solution for \textsf{LIS}. Very recently, Rubinstein \textit{et al.}~\cite{saeedfocs19} obtain $O(n^{3\epsilon})$ approximation in time $\tilde O(n^{0.5+7\epsilon})$ in this case. 
Also, these lower  bounds do not hold for stronger computational models such as quantum computation, but we  do not have better general quantum approximation  algorithms.
% Likewise, for stronger computational models such as quantum computation for which the unconditional lower bounds fail, still no improvement is discovered.

% \Michael{Pull up related work on dynamic setting/explain what the dynamic setting is.}\Saeed{The following paragraph is new. Feel free to edit or leave me comments to apply.}

In this work, we focus on approximation algorithms in the dynamic setting, where at each step, the array can be updated by inserting, deleting, or modifying an element.  The goal is to maintain an approximation of the correct value at each step.
Dynamic settings for many problems have been studied, e.g. \cite{DBLP:conf/stoc/HenzingerKNS15,DBLP:conf/stoc/NanongkaiS17,gawrychowski2018optimal,DBLP:conf/stoc/AssadiOSS18,DBLP:conf/soda/AssadiOSS19,chen2013dynamic,DBLP:conf/stoc/LackiOPSZ15,DBLP:journals/corr/abs-1909-03478,DBLP:conf/focs/NanongkaiSW17}. In general, in dynamic settings the goal is to develop an algorithm where the solution can be updated efficiently given incremental changes to the input. In the context of graph algorithms~\cite{DBLP:conf/stoc/NanongkaiS17,DBLP:conf/focs/NanongkaiSW17,DBLP:conf/stoc/LackiOPSZ15,DBLP:conf/stoc/AssadiOSS18,DBLP:conf/soda/AssadiOSS19,DBLP:journals/corr/abs-1909-03478}, such changes are usually modeled by edge addition or edge deletion. For string problems, changes are typically modeled with character insertion and character deletion~\cite{gawrychowski2018optimal,chen2013dynamic}, as we consider here.  

% a query of one of the following types arrives and we update the array accordingly:
% \begin{itemize}
%	\item A new element is added at an arbitrary position of the array.
%	\item An existing element of the array is removed.
%	\item The value of an arbitrary element of the array is updated.
% \end{itemize}
We provide novel approximation algorithms for \textsf{LIS} and \textsf{DTM} in the dynamic setting.
For \textsf{LIS}, for any $0 < \epsilon < 1$, we give a dynamic algorithm with worst-case update time $\tilde O(n^{\epsilon})$ and approximation factor $O_{\epsilon}(1)$;  that is, for constant $\epsilon$, the approximation factor is a constant that depends on $\epsilon$. For \textsf{DTM}, we present an algorithm with approximation factor $1+\epsilon$ for any constant $\epsilon$ and worst-case update time $O(\log^2 n)$, where the order notation hides factors that can depend on $\epsilon$. We primarily treat $\epsilon$ as constant since the exponent of the $\log$ factors suppressed by the $\tilde O$ notation may depend on $1/\epsilon$. Here, $n$ denotes the size of the array at the time the operation arrives. Thus the update time does not depend on the number of operations arrived prior to the new operation.

\definecolor{LightCyan}{rgb}{0.88,1,1}

\begin{table}[!htbp]
	\centering
	\begin{tabular}{|l|c|c|}
		\hline
		Problem & Approximation factor & Update time\\
		\hline
		\textsf{LIS} & $1+\epsilon$ & $\tilde O(\sqrt{n})$\\
		\hline	
        \textsf{LIS} & $O((1/\epsilon)^{O(1/\epsilon)})$ & $\tilde O(n^{\epsilon})$\\
	    \hline
	    $\mathsf{LIS}^+$ & $O(\log n)$ & $O(\log^3 n)$\\
	    \hline
		\textsf{DTM} & $1+\epsilon$ & $O(\log^2 n)$\\
        \hline
	\end{tabular}
\caption{The results of this paper are summarized in this table. $\mathsf{LIS}^+$ is a special case of \textsf{LIS} where only element insertion is allowed.}\label{table:results}
\end{table}

\subsection{Erd\"{o}s-Szekeres partitioning problem}
Our dynamic algorithm has an interesting application to a long-standing mathematical problem, namely \textit{the Erd\"{o}s-Szekeres partitioning problem}. It is well-known that any sequence of size $n$ can be decomposed into $O(\sqrt{n})$ monotone subsequences. The proof follows from a simple fact: Any sequence of length $n$ contains either an increasing subsequence of length $\sqrt{n}$ or a non-increasing subsequence of length $\sqrt{n}$. Thus, one can iteratively find the maximum increasing and the maximum non-increasing subsequences of a sequence and take the larger one as one of the solution partitions. Next, by removing the partition from the original sequence and repeating this procedure with the remainder of the elements we obtain a decomposition into at most $O(\sqrt{n})$ partitions. The computational challenge, also known as the Erd\"{o}s-Szekeres partitioning problem, is to do this in an efficient way. The above algorithm can be implemented in time $O(n^{1.5} \log n)$ if we use patience sorting in every iteration. Bar-Yehuda and Fogel~\cite{yehuda1998partitioning} improve the runtime down to $O(n^{1.5})$ by designing an algorithm that after a preprocessing step, solves \textsf{LIS} in time $O(n + k^2)$ where the solution size is bounded by $k$. Since any comparison-based algorithm takes time at least $\tilde \Omega(n)$, the gap for Erd\"{o}s-Szekeres partitioning problem has been $\tilde \Omega(\sqrt{n})$ for quite a long time and the question was raised in a number of works as an important open problem~\cite{pettie2003shortest,gronlund2014threesomes}.

We prove that via our dynamic \textsf{LIS} algorithm, the Erd\"{o}s-Szekeres partitioning problem can be solved in time $\tilde O_{\epsilon}(n^{1+\epsilon})$ for any constant $\epsilon > 0$. We assume that our algorithm performs as stated in Table~\ref{table:results}.

\begin{theorem}\label{theorem:main}
	For any constant $\epsilon > 0$, one can in time $\tilde O_{\epsilon}(n^{1+\epsilon})$ partition any sequence of length $n$ of distinct integer numbers into $O_{\epsilon}(\sqrt{n})$ monotone (increasing or decreasing) subsequences.
\end{theorem}
\begin{proof}
	The proof follows directly from our algorithm for dynamic \textsf{LIS}. In our dynamic setting, we start with an empty array $a$ and at every point in time we are allowed to (i) add an element, or (ii) remove an element, or (iii) substitute an element for another. The algorithm is able to update the sequence and estimate the size of the \textsf{LIS} in time $\tilde O_{\epsilon}(|a|^{\epsilon})$ where $|a|$ is the size of the array at the time the operation is performed. Moreover, the approximation factor of our algorithm is constant as long as $\epsilon$ is constant. More precisely, our algorithm estimates the size of the longest increasing subsequence within a multiplicative factor of at most $(1/\epsilon)^{O(1/\epsilon)}$. It follows from our algorithm that by spending additional time proportional to the reported estimation, our algorithm is able to also find an increasing subsequence with size equal to the reported length.
	
	Given a sequence of length $n$ with distinct numbers, we use the dynamic algorithm for \textsf{LIS} to decompose it into $O_{\epsilon}(\sqrt{n})$ monotone subsequences in time $\tilde O_{\epsilon}(n^{1+\epsilon})$. To do so, we initialize two instances of our dynamic algorithm that keep an approximation to the longest increasing subsequence and the longest decreasing subsequence of the array. More precisely, in the first instance, we insert all elements of the array exactly the same way they appear in our sequence and in the second instance we insert the elements in the reverse order. Thus the dynamic algorithm for the second instance always maintains an approximation to the longest decreasing subsequence of our array.
	
	In every iteration, we estimate the size of the longest increasing and longest decreasing subsequences of the array via the dynamic \textsf{LIS} algorithm. We then choose the maximum one and ask the algorithm to give us the sequence corresponding to the solution reported. Finally, we remove the elements from both instances of the dynamic algorithm and repeat the same procedure for the remainder of the elements.
	
	The total runtime of our algorithm is $\tilde O_{\epsilon}(n^{1+\epsilon})$ since we insert $n$ elements in each of the instances and then remove $n$ elements which amounts to $2n$ operations for each instance that runs in time $\tilde O_{\epsilon}(n^{1+\epsilon})$. Moreover, because at every point in time the maximum estimate we receive from each of the dynamic algorithms is at least a constant fraction of the actual longest increasing subsequence, we repeat this procedure at most $O_{\epsilon}(\sqrt{n})$ times. Therefore, we decompose the sequence into $O_{\epsilon}(\sqrt{n})$ monotone subsequences.
\end{proof}

\begin{remark}
	The constant factor hidden in the $O$ notation for the number of partitions is optimal in neither the algorithm of Theorem~\ref{theorem:main} nor the previous algorithm of~\cite{yehuda1998partitioning} nor the simple greedy algorithm that runs patience sorting in every step.
\end{remark}

\subsection{Subsequent Work}
Since our dynamic algorithm has constant approximation factor, in order to make sure the number of partitions remains $O(\sqrt{n})$, one needs to set $\epsilon$ to constant and therefore the gap between our runtime of $\tilde O(n^{1+\epsilon})$ and the lower bound of $\Omega(n)$ remains polynomial. Two independent subsequent work further tighten the gap. Kociumaka and Seddighin~\cite{saeednew} improve the gap to subpolynomial by presenting a dynamic algorithm with approximation factor $1-o(1)$ and update time $O(n^{o(1)})$. Gawrychowski and Janczewski~\cite{gawrychowski2020fully} further tighten the gap to polylogarithmic by obtaining a similar algorithm with polylogarithmic update time (with polynomial dependence on $1/\epsilon$). The work of Kociumaka and Seddighin~\cite{saeednew} also gives the first exact algorithm for dynamic \textsf{LIS} with sublinear update time. Their algorithm is able to update the solution in time $\tilde O(n^{2/3})$ after each operation and gives a correct solution with probability $1-n^{-5}$.

In another subsequent work, Mitzenmacher and Seddighin~\cite{our-soda-paper} use the grid-packing technique given here to obtain an improved sublinear time algorithm for approximating \textsf{LIS}. Their algorithm is able to obtain an approximation of \textsf{LIS} in truly sublinear time within a factor of $\Omega(\lambda^{\epsilon})$ where $\epsilon > 0$ is an arbitrarily small constant factor and $\lambda$ is the ratio of the solution size and the input size. 

\subsection{Related Work}
\textsf{LIS} has received significant attention in the areas of property testing~\cite{DBLP:conf/stoc/ErgunKKRV98, DBLP:conf/random/DodisGLRRS99,DBLP:journals/eatcs/Fischer01,DBLP:conf/approx/AilonCCL04}, streaming~\cite{DBLP:conf/soda/GopalanJKK07,DBLP:conf/focs/GalG07},
and massively parallel computation (MPC)~\cite{DBLP:conf/stoc/ImMS17}, as well as in the standard algorithmic  setting~\cite{fredman1975computing,ramanan1997tight,saeedfocs19,saks2010estimating}. Several questions remain open about approximation algorithms for \textsf{LIS}. Although a linear lower bound on the runtime is trivial when the solution size is $O(1)$, neither convincing lower bounds nor upper bounds are known for approximating \textsf{LIS} within subpolynomial multiplicative factors if the solution size is larger ($\omega(1)$) in general. For a special case when $n/\textsf{LIS}(a)$ is subpolynomial, we can approximate the solution size within a subpolynomial factor in sublinear time~\cite{saeedfocs19,saks2010estimating}. In particular, Saks and Seshadhri~\cite{saks2010estimating} present a $(1+\epsilon)$ approximation algorithm  for \textsf{LIS} in sublinear time if the ratio of $n$ over the solution size is sublogarithmic. The only prior non-trivial dynamic algorithm for \textsf{LIS} that we are aware of is the work of Chen \textit{et al.}~\cite{chen2013dynamic}, where the authors present an exact dynamic algorithm for \textsf{LIS} with worst-case update time $O(r+\log n)$ when the solution size is bounded by $r$. The update time for this algorithm can grow up to $\Omega(n)$ if the solution size is $\Omega(n)$.

When the available memory is sublinear (as it is in the streaming and the MPC models), patience sorting can be used to compute a solution for smaller fragments of the input. Previous work show that these local solutions can be cleverly merged to obtain $1+\epsilon$ approximate solutions in the streaming~\cite{DBLP:conf/soda/GopalanJKK07} and the MPC models~\cite{DBLP:conf/stoc/ImMS17}. In contrast, our technique for approximating \textsf{LIS} is not based on patience sorting.  We show it also has an application to a streaming variant of \textsf{LIS}, and we expect it will have additional applications in the future. 

Distance to monotonicity (\textit{a.k.a} Ulam distance) is also a very well-studied problem~\cite{DBLP:conf/soda/NaumovitzSS17,DBLP:conf/soda/AndoniN10,DBLP:conf/spaa/BoroujeniS19,saks2010estimating}. While \textsf{LIS} has resisted a multiplicative approximation algorithm, \textsf{DTM} can be approximated within a multiplicative factor $1+\epsilon$ in time $\tilde O (n/d + \sqrt{n})$ when the solution size is lower bounded by $d$~\cite{DBLP:conf/soda/NaumovitzSS17}. Streaming~\cite{DBLP:conf/soda/GopalanJKK07} and MPC~\cite{DBLP:conf/spaa/BoroujeniS19} algorithms for \textsf{DTM} have also appeared.

\subsection{Preliminaries}
We consider the two problems, \textsf{LIS} and \textsf{DTM}. Input to both problems is an array $a$ with arbitrary length. For \textsf{LIS}, the goal is to find the length of the largest subsequence of elements such that their values increase according to their indices. For \textsf{DTM} the goal is to determine the smallest number of elements such that the remaining subsequence is increasing. Obviously, $\mathsf{DTM}(a) = |a| - \mathsf{LIS}(a)$. However, an approximate solution for one  problem does not imply an approximate solution for the other (much like maximum matching and vertex cover).  We assume for simplicity and without loss of generality that all the numbers are distinct, although one can easily modify our algorithm to handle repeated numbers. 

% \Michael{  check this for repeats} \Saeed{Done. Removed the parts that we said we only return the solution size. We mentioned this earlier in the intro.}
Our results here are for the dynamic setting. Initially, the input array is empty ($|a| = 0$). At each step, an element is either inserted at an arbitrary position of the array or removed from an arbitrary position of the array. (Element substitution can also be implemented with the previous two operations, so we consider only insertions and removals.) We also study a special case of \textsf{LIS} where all operations add elements to the array. We call this problem $\lisplus$. % Our algorithms return the size of the solution (and not the entire solution), however, it is easy to see that they can also provide a corresponding solution for the size reported if extra time proportional to the size of the solution is provided.

% \Michael{I don't like "query".  I prefer operations, or updates, or update operations. Probably  just operations.  }\Saeed{Fixed! (technical prats are not modified yet)}
We more formally define the array operations. Each insertion operation is of the form ``\textsf{insert $(i,x)$}" where $i$ is an integer between $1$ and the length of the current array plus one. $i$ specifies the position of element $x$. After this operation, all the elements whose previous index was at least $i$ will be shifted to the right. Similarly, an operation ``\textsf{delete $(i)$}" removes the $i$'th element of the array and element $i+1$ will replace its position. Likewise, all the elements whose previous index was at least $i$ will be shifted to the left.

For simplicity, in our algorithms we assume that at any point random access to the elements is provided. That is, in every step, one can access the value of the $i$'th element of the array as a value query. This brings an $O(\log n)$ overhead to the runtime since one needs to design a data structure that allows us element addition, element removal, and access to the $i$'th element. Any balanced binary tree (\textit{e.g.} red-black tree) suffices for that purpose~\cite{dasgupta2008algorithms}. We can also recover the position of each element of the array in logarithmic time with a balanced binary tree.

\section{Summary of the Results and Techniques}\label{sec:results}
Our main result is a dynamic algorithm for \textsf{LIS} with worst-case update time $\tilde O(n^{\epsilon})$ and approximation factor $O((1/\epsilon)^{O(1/\epsilon)})$. 
%\Michael{The question always comes up as to whether this just holds for constant $\epsilon  > 0$ of if $\epsilon$ can depend on $n$ so maybe address it early once.  Also, what is $n$ here?  In the dynamic setting we start with an empty array.  Is $n$ the number of operations?  If you haven't done so clarify somewhere.}\Saeed{I added a sentence to the abstract for $n$. We don't need $\epsilon$ to be constant. However, I don't want to emphasize that $\epsilon$ can depend on $n$ since the exponent of the $\log$ factor hidden in $\tilde O$ notation depends on $1/\epsilon$.}
%\Michael{I would take the explanationn of n out of the abstract, but explain at the first time you use n in the introduction.  That's fine as long as we use it consistently that way throughout or are clear when we are not.  For the epsilon, I would early on say that the reader should think of epsilon as constant, and then footnote or explain as you did in the comment -- that epsilon can be non-constant but for non-constant epsilon one has to be careful for the reasons you've described, so we express our results as though epsilon was constant unless stated otherwise.  (And where it might be useful, state so!)}\Saeed{Added a sentence before the table.}
Our algorithm is based on a novel technique which we call grid packing. % This new technique  has applications to other models. As an example, we show that grid packing leads to an improved streaming algorithm for \textsf{LIS} when a helpful adviser is available.
%\Michael{I've  commented where the above is repetitive -- we've now said this earlier.}\Saeed{Looks good to me}
%\Michael{If lazy algorithm is a term of the art then give a reference.  If not, I don't think we should use the term lazy.  Because I don't see how the algorithm is lazy at all.  My version  of the paragraph follows this version.}\Saeed{Replaced lazy with block-based}
%To simplify the explanation, we use the notion of \textit{block-based algorithms} in this paper. Roughly speaking, we call an algorithm block-based if it starts with an array $a$ of length $n$ and is responsible for at most $g(n)$ operations. The benefit of a block-based algorithm is that we allow for a preprocessing time of $f(n)$ for block-based algorithms. It is not hard to see (see Section \ref{sec:amortized}) that a block-based algorithm with preprocessing time $f(n)$ and worst-case update time $h(n)$ can be used to obtain a dynamic algorithm with worst-case update time $\max\{f(n)/g(n),h(n)\}$. This can be seen as a reduction from worst-case update time to amortized update time.
In this section, we give a high-level summary of grid packing and our overall approach; Full proofs are given in Sections~\ref{sec:grid} and ~\ref{sec:constant}.
\subsection{Block-based Algorithms}
In our work, to simplify our proofs, we utilize the notion of what we call a \textit{block-based algorithm}.  Very roughly speaking, a block-based algorithm starts with an array $a$ of length $n$.  It can use $f(n)$ preprocessing time, after which it is responsible for a  block of $g(n)$ operations, where each operation has worst-case update time $h(n)$.  We show in Section \ref{sec:amortized} via a simple reduction that such a block-based algorithm can be used to obtain a dynamic algorithm with worst-case update time $\max\{f(n)/g(n),h(n)\}$.

%\Michael{Editing to use block-based in the below paragraph, edit back if you like.}\Saeed{This looks better to me. I commented the previous paragraph.}

A motivating example shows the notion of block-based algorithms simplifies the analysis. Chen \textit{et al.}~\cite{chen2013dynamic} show that when \textsf{LIS} for an array is upper bounded by $r$, a dynamic algorithm can maintain the exact solution for \textsf{LIS} with worst-case update time $\tilde O(r)$.  We show that this exact algorithm yields a dynamic $(1+\epsilon)$-approximation algorithm with worst-case update time $\tilde O(\sqrt{n})$. 
%\Michael{Again, clarify here if $n$ is array size, number of operations, what exactly.}

We first provide the intuition and explain the complications.  At any point in time, if the solution value is below $2\sqrt{n}/\epsilon$, then the runtime guarantee is met by using the algorithm of Chen \textit{et al.}~\cite{chen2013dynamic}. Otherwise, we can compute an exact solution, and then use the same value for up to $\sqrt{n}$ steps to maintain a valid approximation.
We then spend time $\tilde O(n)$ for $\sqrt{n}$ operations, leading to amortized update time of $\tilde O(\sqrt{n})$. Deamortizing
this approach to bound the worst-case update time seems cumbersome. The issue is that it is not clear when we should switch between the two algorithms. For example, if we define a threshold $\tau$ and switch between the algorithms when the solution size crosses the threshold $\tau$, we may go back and forth across the threshold.  We could consider multiple thresholds, but at this point we appear to be complicating the analysis beyond what should be necessary.  

Working with the framework of a block-based algorithm conveniently remedies the problem. Assuming we start with an array $a$ of length $n$, we allow a preprocessing time of $f(n) = O(n \log n)$ for the algorithm to compute the \textsf{LIS}. We set $g(n) = \sqrt{n}$. If the \textsf{LIS} value $r$ is above $2\sqrt{n}/\epsilon$, for the next $\sqrt{n}$ steps, we report $r-i$ in the $i$'th step and can be sure that our solution is within a small range from the optimal one.
%\Michael{Why do we have to report $r-i$, can't we stick with a single value?}\Saeed{We don't want to report a value which is larger than the actual \textsf{LIS} even if its within a $1+\epsilon$ multiplicative factor. Since the \textsf{LIS} may decrease after each operation, we have to have an additive term $-i$}\Michael{Why don't we want to report a value that is larger?  This seems like a new restriction-- have we ever stated that as a requirement?  It seems like we need to explain that earlier and why if that's the case.}\Saeed{If the solution size is 10 and I report 9, I could say that my approximation factor is $9/10$. But if the solution size is $10$ and I report $11$, my solution is invalid since no such increasing subsequence exists (and I certainly can't report one if somebody asks for it). Because of this, I prefer to only give smaller values for LIS and larger values for DTM.}
Otherwise, we use the algorithm of Chen \textit{et al.}~\cite{chen2013dynamic} with worst-case update time $O(\sqrt{n})$ for the next $\sqrt{n}$ steps. Using the reduction, this turns to an algorithm with worst-case update time $\tilde O(\sqrt{n})$ for \textsf{LIS} with approximation factor $1+\epsilon$. 

\vspace{0.2cm}
{\noindent \textbf{Corollary}~\ref{theorem:trivial}, [restated informally]. \textit{For any constant $\epsilon > 0$, there exists a dynamic algorithm for \textsf{LIS} with worst-case update time $\tilde O(\sqrt{n})$ and approximation factor $1+\epsilon$.\\}}

%\Michael{grid packing is OK, but grid packing does refer to another algorithmic problem (google it), and it's not clear to me that this is a "packing".  I might call it grid path covering, grid segment covering, or just grid covering.  But I'm open to leaving it as is.}

%\Michael{I'm pretty sure you meant "precedes" when you say segment A proceed (should be precedes) segment B.  Please check and change.}\Saeed{Done.}
\subsection{Grid Packing and Applications}
As mentioned earlier, our algorithm for \textsf{LIS} is based on a technique that we call grid packing. Grid packing is defined on a table of $m \times m$ cells; the only parameter of the problem is $m$.
The problem can be thought of as a game between us and an adversary.
We introduce a number of segments on the table. Each segment covers a consecutive set of cells in either a row or in a column. A segment $A$ \textit{precedes} a segment $B$ if \textbf{every} cell of $A$ is strictly higher than every cell of $B$ and also \textbf{every} cell of $A$ is strictly to the right of every cell of $B$. Two segments are \textit{non-conflicting}, if one of them precedes the other one. Otherwise, we call them \textit{conflicting}.  The segments we introduce can overlap and there is no restriction on the number of segments or the length of each segment. However, we would like to minimize the maximum number of segments that cover a cell. 

\begin{figure}[ht]

\centering

\tikzset{every picture/.style={line width=0.75pt}} %set default line width to 0.75pt        

\begin{tikzpicture}[x=0.75pt,y=0.75pt,yscale=-1,xscale=1]

%Shape: Rectangle [id:dp5254311647291401] 
\draw   (181,11) -- (450,11) -- (450,281) -- (181,281) -- cycle ;
%Straight Lines [id:da3305762449852494] 
\draw    (301,10) -- (301,281) ;

%Straight Lines [id:da3206498977891954] 
\draw    (331,10) -- (331,282) ;

%Straight Lines [id:da02838390385871148] 
\draw    (361,10) -- (361,280) ;

%Straight Lines [id:da6891763144045235] 
\draw    (211,10) -- (211,281) ;

%Straight Lines [id:da38683484497784937] 
\draw    (241,10) -- (241,281) ;

%Straight Lines [id:da7490078717616151] 
\draw    (271,10) -- (271,281) ;

%Straight Lines [id:da3531030059687228] 
\draw    (391,10) -- (391,281) ;

%Straight Lines [id:da43894336811181467] 
\draw    (421,10) -- (421,280) ;

%Straight Lines [id:da27078630886751465] 
\draw    (182,41) -- (450,41) ;

%Straight Lines [id:da005150497881268201] 
\draw    (182,71) -- (450,71) ;

%Straight Lines [id:da23147705862104506] 
\draw    (182,101) -- (450,101) ;

%Straight Lines [id:da7425549844424117] 
\draw    (182,131) -- (450,131) ;

%Straight Lines [id:da2351732964828932] 
\draw    (182,161) -- (450,161) ;

%Straight Lines [id:da7401347371733205] 
\draw    (182,191) -- (450,191) ;

%Straight Lines [id:da8968761436786192] 
\draw    (182,221) -- (450,221) ;

%Straight Lines [id:da8562431949289493] 
\draw    (182,251) -- (450,251) ;

%Straight Lines [id:da7212412833117718] 
\draw [color={rgb, 255:red, 208; green, 2; blue, 27 }  ,draw opacity=1 ][line width=3.75]    (256,114) -- (256,267) ;

%Straight Lines [id:da7675592926230685] 
\draw [color={rgb, 255:red, 74; green, 144; blue, 226 }  ,draw opacity=1 ][line width=3.75]    (319,29) -- (441,29) ;

%Straight Lines [id:da4112941771444574] 
\draw [color={rgb, 255:red, 65; green, 117; blue, 5 }  ,draw opacity=1 ][line width=3.75]    (196,203) -- (196,271) ;

%Shape: Square [id:dp4777174644258826] 
\draw  [color={rgb, 255:red, 0; green, 0; blue, 0 }  ,draw opacity=1 ][fill={rgb, 255:red, 248; green, 231; blue, 28 }  ,fill opacity=0.48 ] (211,41) -- (241,41) -- (241,71) -- (211,71) -- cycle ;
%Straight Lines [id:da09165043527230288] 
\draw [line width=3.75]    (226,53) -- (226,181) ;

%Shape: Square [id:dp8680283690742998] 
\draw  [color={rgb, 255:red, 0; green, 0; blue, 0 }  ,draw opacity=1 ][fill={rgb, 255:red, 74; green, 144; blue, 226 }  ,fill opacity=0.35 ] (271,41) -- (301,41) -- (301,71) -- (271,71) -- cycle ;
%Straight Lines [id:da0026998791284389423] 
\draw [color={rgb, 255:red, 245; green, 166; blue, 35 }  ,draw opacity=1 ][line width=3.75]    (287,57) -- (409,57) ;

\end{tikzpicture}
\caption{Segments are shown on the grid. The pair (black, orange) is conflicting since the yellow cell (covered by the black segment) is on the same row as the blue cell (covered by the orange segment). The following pairs are non-conflicting: (green, black), (green, orange), (green, blue), (red, orange), (red, blue), (black, blue).} \label{fig:crossing}
\end{figure}
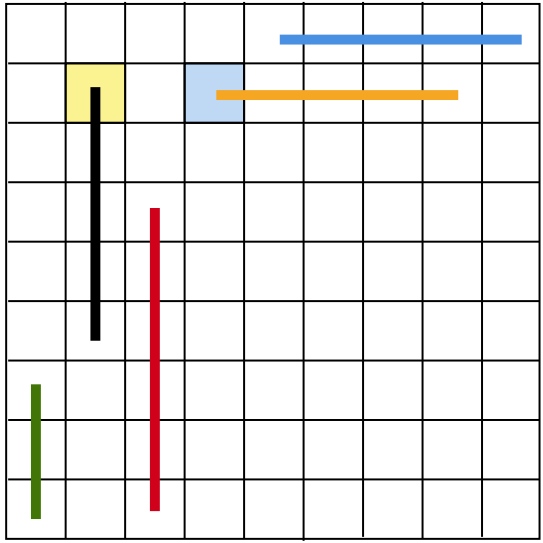

After we introduce the segments, an adversary puts a non-negative number on each cell of the grid. We emphasize that our segments do not depend on these numbers, as the numbers are given after we provide our segments. The score of a subset of cells in the table is the sum of the values in the cells, and the overall score of the table is the maximum score of a path of length $2m-1$ from the bottom-left corner to the top-right corner. (In such a path, each move is either up or to the right.)

The score that we obtain using our segments is the maximum sum of scores of a non-conflicting set of segments.
One can easily verify that the score of the table is a clear upper bound on the score we obtain using any subset of non-conflicting segments. We would like to introduce the segments in such a way that the ratio of the score of the table over our score is always bounded by constant, no matter how the adversary puts the numbers on the table. 
For a fixed $\alpha \geq 1$ and a $\beta \geq 1$, we call a solution $(\alpha,\beta)$-approximate if at most $\alpha$ segments cover each cell and it guarantees a $1/\beta$ fraction of the score of the table for us for any assignment of numbers to the table cells. 
%\Michael{So $\beta$ is bigger than 1 here?}\Saeed{Yes, I changed the previous sentence to make it clear.}
We show in Section \ref{sec:grid} that grid packing admits an $(O(m^\kappa \log m),O(1/\kappa))$-approximate solution for any $0 < \kappa < 1$.
%\Michael{I'm confused, if $\beta$ is bigger than 1, why is this $O(1/\kappa)$?  Is $\kappa$ supposed to be bigger than 1 or smaller than 1 -- it can't be just bigger than 0?  It's odd having O notation in both parts of the expression -- it makes sense for me that the second coordinate, the $\beta$, would be fixed and then we'd get a corresponding $\alpha$.}\Saeed{$\kappa$ is between 0 and 1 so $1/\kappa$ is always greater than $1$. I will think about a clean way to fix one of the parameters and obtain the other one in terms of it. However, it may be OK to leave it as it is since many previous work do this. E.g. https://users.cs.duke.edu/~debmalya/papers/soda13-scheduling.pdf}

Before explaining the idea behind this result, we would like to make a connection between grid packing and \textsf{LIS}. Let us consider an array $a$ of length $n$. We assume for the sake of this example that all the numbers of the array are distinct and are in range $[1,n]$. In other words, $a$ is a permutation of numbers in $[n]$. We map the array to a set of points on the 2D plane by putting a point at $(i,a_i)$ for every position $i$ of the array.

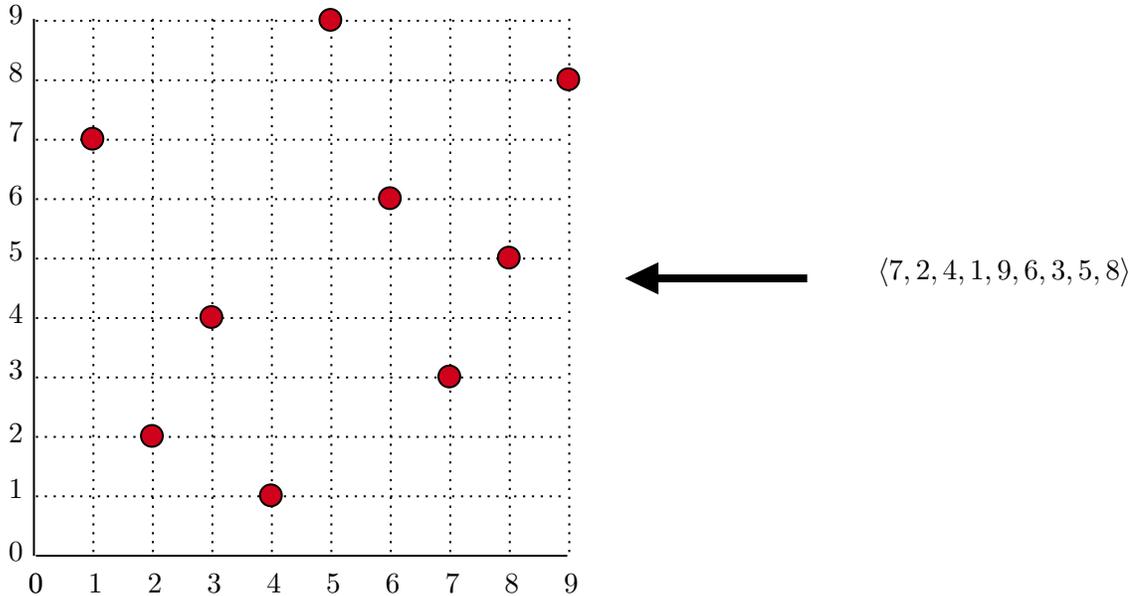
\begin{figure}[ht]

\centering

\tikzset{every picture/.style={line width=0.75pt}} %set default line width to 0.75pt        

\begin{tikzpicture}[x=0.75pt,y=0.75pt,yscale=-1,xscale=1]
%uncomment if require: \path (0,369); %set diagram left start at 0, and has height of 369

%Straight Lines [id:da3191885316685632] 
\draw  [dash pattern={on 0.84pt off 2.51pt}]  (147,49) -- (147,320) ;

%Straight Lines [id:da7559101633943939] 
\draw  [dash pattern={on 0.84pt off 2.51pt}]  (177,49) -- (177,321) ;

%Straight Lines [id:da1291104188279324] 
\draw  [dash pattern={on 0.84pt off 2.51pt}]  (207,49) -- (207,319) ;

%Straight Lines [id:da6247914431079504] 
\draw  [dash pattern={on 0.84pt off 2.51pt}]  (57,49) -- (57,320) ;

%Straight Lines [id:da2443453671621807] 
\draw  [dash pattern={on 0.84pt off 2.51pt}]  (87,49) -- (87,320) ;

%Straight Lines [id:da4676476525067903] 
\draw  [dash pattern={on 0.84pt off 2.51pt}]  (117,49) -- (117,320) ;

%Straight Lines [id:da42635072959794496] 
\draw  [dash pattern={on 0.84pt off 2.51pt}]  (237,49) -- (237,320) ;

%Straight Lines [id:da8277079805087988] 
\draw  [dash pattern={on 0.84pt off 2.51pt}]  (267,49) -- (267,319) ;

%Straight Lines [id:da24985081014096644] 
\draw  [dash pattern={on 0.84pt off 2.51pt}]  (28,80) -- (296,80) ;

%Straight Lines [id:da6339328653912732] 
\draw  [dash pattern={on 0.84pt off 2.51pt}]  (28,110) -- (296,110) ;

%Straight Lines [id:da10731858330496191] 
\draw  [dash pattern={on 0.84pt off 2.51pt}]  (28,140) -- (296,140) ;

%Straight Lines [id:da2712941152341055] 
\draw  [dash pattern={on 0.84pt off 2.51pt}]  (28,170) -- (296,170) ;

%Straight Lines [id:da7317347402429573] 
\draw  [dash pattern={on 0.84pt off 2.51pt}]  (28,200) -- (296,200) ;

%Straight Lines [id:da8994404191047722] 
\draw  [dash pattern={on 0.84pt off 2.51pt}]  (28,230) -- (296,230) ;

%Straight Lines [id:da6044897770856661] 
\draw  [dash pattern={on 0.84pt off 2.51pt}]  (28,260) -- (296,260) ;

%Straight Lines [id:da7815256961262274] 
\draw  [dash pattern={on 0.84pt off 2.51pt}]  (28,290) -- (296,290) ;

%Straight Lines [id:da06443896163595575] 
\draw    (27,49) -- (27,320) ;

%Straight Lines [id:da8160886746264304] 
\draw  [dash pattern={on 0.84pt off 2.51pt}]  (297,49) -- (297,319) ;

%Straight Lines [id:da5343639785815872] 
\draw    (28,320) -- (296,320) ;

%Straight Lines [id:da1659083679036335] 
\draw  [dash pattern={on 0.84pt off 2.51pt}]  (28,50) -- (296,50) ;

%Shape: Ellipse [id:dp9836802741092274] 
\draw  [color={rgb, 255:red, 0; green, 0; blue, 0 }  ,draw opacity=1 ][fill={rgb, 255:red, 208; green, 2; blue, 27 }  ,fill opacity=1 ] (51,109.5) .. controls (51,106.46) and (53.46,104) .. (56.5,104) .. controls (59.54,104) and (62,106.46) .. (62,109.5) .. controls (62,112.54) and (59.54,115) .. (56.5,115) .. controls (53.46,115) and (51,112.54) .. (51,109.5) -- cycle ;
%Shape: Ellipse [id:dp8584062429307988] 
\draw  [color={rgb, 255:red, 0; green, 0; blue, 0 }  ,draw opacity=1 ][fill={rgb, 255:red, 208; green, 2; blue, 27 }  ,fill opacity=1 ] (81,259.5) .. controls (81,256.46) and (83.46,254) .. (86.5,254) .. controls (89.54,254) and (92,256.46) .. (92,259.5) .. controls (92,262.54) and (89.54,265) .. (86.5,265) .. controls (83.46,265) and (81,262.54) .. (81,259.5) -- cycle ;
%Shape: Ellipse [id:dp23277672408707772] 
\draw  [color={rgb, 255:red, 0; green, 0; blue, 0 }  ,draw opacity=1 ][fill={rgb, 255:red, 208; green, 2; blue, 27 }  ,fill opacity=1 ] (111,199.5) .. controls (111,196.46) and (113.46,194) .. (116.5,194) .. controls (119.54,194) and (122,196.46) .. (122,199.5) .. controls (122,202.54) and (119.54,205) .. (116.5,205) .. controls (113.46,205) and (111,202.54) .. (111,199.5) -- cycle ;
%Shape: Ellipse [id:dp25469902245724896] 
\draw  [color={rgb, 255:red, 0; green, 0; blue, 0 }  ,draw opacity=1 ][fill={rgb, 255:red, 208; green, 2; blue, 27 }  ,fill opacity=1 ] (141,289.5) .. controls (141,286.46) and (143.46,284) .. (146.5,284) .. controls (149.54,284) and (152,286.46) .. (152,289.5) .. controls (152,292.54) and (149.54,295) .. (146.5,295) .. controls (143.46,295) and (141,292.54) .. (141,289.5) -- cycle ;
%Shape: Ellipse [id:dp8339059316567126] 
\draw  [color={rgb, 255:red, 0; green, 0; blue, 0 }  ,draw opacity=1 ][fill={rgb, 255:red, 208; green, 2; blue, 27 }  ,fill opacity=1 ] (171,49.5) .. controls (171,46.46) and (173.46,44) .. (176.5,44) .. controls (179.54,44) and (182,46.46) .. (182,49.5) .. controls (182,52.54) and (179.54,55) .. (176.5,55) .. controls (173.46,55) and (171,52.54) .. (171,49.5) -- cycle ;
%Shape: Ellipse [id:dp8836850789925752] 
\draw  [color={rgb, 255:red, 0; green, 0; blue, 0 }  ,draw opacity=1 ][fill={rgb, 255:red, 208; green, 2; blue, 27 }  ,fill opacity=1 ] (201,139.5) .. controls (201,136.46) and (203.46,134) .. (206.5,134) .. controls (209.54,134) and (212,136.46) .. (212,139.5) .. controls (212,142.54) and (209.54,145) .. (206.5,145) .. controls (203.46,145) and (201,142.54) .. (201,139.5) -- cycle ;
%Shape: Ellipse [id:dp25676474134901084] 
\draw  [color={rgb, 255:red, 0; green, 0; blue, 0 }  ,draw opacity=1 ][fill={rgb, 255:red, 208; green, 2; blue, 27 }  ,fill opacity=1 ] (231,229.5) .. controls (231,226.46) and (233.46,224) .. (236.5,224) .. controls (239.54,224) and (242,226.46) .. (242,229.5) .. controls (242,232.54) and (239.54,235) .. (236.5,235) .. controls (233.46,235) and (231,232.54) .. (231,229.5) -- cycle ;
%Shape: Ellipse [id:dp7276624991606142] 
\draw  [color={rgb, 255:red, 0; green, 0; blue, 0 }  ,draw opacity=1 ][fill={rgb, 255:red, 208; green, 2; blue, 27 }  ,fill opacity=1 ] (261,169.5) .. controls (261,166.46) and (263.46,164) .. (266.5,164) .. controls (269.54,164) and (272,166.46) .. (272,169.5) .. controls (272,172.54) and (269.54,175) .. (266.5,175) .. controls (263.46,175) and (261,172.54) .. (261,169.5) -- cycle ;
%Shape: Ellipse [id:dp012143613453408975] 
\draw  [color={rgb, 255:red, 0; green, 0; blue, 0 }  ,draw opacity=1 ][fill={rgb, 255:red, 208; green, 2; blue, 27 }  ,fill opacity=1 ] (291,79.5) .. controls (291,76.46) and (293.46,74) .. (296.5,74) .. controls (299.54,74) and (302,76.46) .. (302,79.5) .. controls (302,82.54) and (299.54,85) .. (296.5,85) .. controls (293.46,85) and (291,82.54) .. (291,79.5) -- cycle ;
%Straight Lines [id:da5514943528211631] 
\draw [line width=3]    (417,180) -- (330,180) ;
\draw [shift={(325,180)}, rotate = 360] [fill={rgb, 255:red, 0; green, 0; blue, 0 }  ][line width=3]  [draw opacity=0] (16.97,-8.15) -- (0,0) -- (16.97,8.15) -- cycle    ;

% Text Node
\draw (-34,147) node   {$ \begin{array}{l}
	\end{array}$};
% Text Node
\draw (517,177) node   {$\langle 7,2,4,1,9,6,3,5,8\rangle$};
% Text Node
\draw (28,334) node   {$0$};
% Text Node
\draw (58,334) node   {$1$};
% Text Node
\draw (88,334) node   {$2$};
% Text Node
\draw (118,334) node   {$3$};
% Text Node
\draw (148,334) node   {$4$};
% Text Node
\draw (178,334) node   {$5$};
% Text Node
\draw (208,334) node   {$6$};
% Text Node
\draw (238,334) node   {$7$};
% Text Node
\draw (268,334) node   {$8$};
% Text Node
\draw (298,334) node   {$9$};
% Text Node
\draw (28,334) node   {$0$};
% Text Node
\draw (18,318) node   {$0$};
% Text Node
\draw (18,286) node   {$1$};
% Text Node
\draw (18,258) node   {$2$};
% Text Node
\draw (18,226) node   {$3$};
% Text Node
\draw (18,198) node   {$4$};
% Text Node
\draw (18,166) node   {$5$};
% Text Node
\draw (18,138) node   {$6$};
% Text Node
\draw (18,106) node   {$7$};
% Text Node
\draw (18,76) node   {$8$};
% Text Node
\draw (18,47) node   {$9$};

\end{tikzpicture}

\caption{An array $\langle 7, 2, 4, 1, 9, 6, 3, 5, 8\rangle$ is mapped to the 2D plane.} \label{fig:lis-grid}
\end{figure}

Now, divide the plane into an $m \times m$ grid, and fix a longest increasing subsequence. The number on each cell of the grid would be equal to the contribution of the elements in that grid cell to the fixed longest increasing subsequence. (We emphasize that the number is {\em not} the longest increasing subsequence inside the cell, but the contribution to the fixed longest increasing subsequence only.)  It follows that the score of the grid is exactly equal to the size of the longest increasing subsequence. Let us assume that the score of each segment is available. To approximate the score of the grid (which equals the size of the \textsf{LIS}) we find the largest score we can obtain using non-conflicting segments by dynamic programming. The last observation which gives us speedup for \textsf{LIS} is the following:
%\Michael{The next sentence isn't following clearly -- why are we approximating  the score of the segement, and not finding it exactly?  And you're saying we're computing the LIS for the "element correspondinng to that segement".  Again, that sounds like an exact, not approximate, computations;  also, I'm not clear on what the  "elements corresponding to that segment" are.  Elements in the original sequence?  But then why are  computing the LIS -- we just said we're NOT computing LIS's within the cell?  Do you mean the grid scores along the segments?  How does this relate to grid packing we just went through in the previous paragraph but is not mentioned here?  Please clarify.}\Saeed{Revised the next paragraph.}
instead of using the score of each segment (which we are not aware of), we use the size of the \textsf{LIS} for each segment as an approximate value for its score. \textsf{LIS} of each segment can be computed in time $\tilde O(n/m)$ since at most $n/m$ elements appear in every row or every column of the grid. This quantity is clearly an upper bound on the score of each segment but can be used to construct a global solution for the entire array (see Section \ref{sec:constant} for more details). In our dynamic algorithm, every time a change is made, we only need to update the approximate score (\textsf{LIS}) of the corresponding segments. %More details about this is given in Section \ref{sec:constant}.
%Therefore, we can dynamically keep the size of \textsf{LIS} for each segment and then each time a change is made, we run the DP to update the approximate solution for the overall grid. More details about this is given in Section .
%\Michael{Again, the above is confusing;  not clear what "elements corresponding to that segment are", what we're computing LIS's on and why....} \Saeed{I cut some of the details and made it clear why we use \textsf{LIS}}

\input{figs/lis-grid2}

%\Michael{So at this point we wander off for next 5 pages, and appear to give full proofs regarding grid packing, etc.  This is not structured right.  At some point the introduction should end, and if we're presenting results, we should divide out the section and give actual results.  Here seems like a good place...}\Saeed{I'll try to cut some of the proof ideas. Since the reviewers usually read the first 10 pages, I like to bring all the ideas (and no proofs) here. If anybody is interested in the proofs, they can read the corresponding section. We can bring Section 1.3 out of the intro and make a separate section if you prefer that.}
Our solution for grid packing is based on a combinatorial construction. The first observation is that any path of length $2m-1$ from the bottom left to the top right of the grid can be decomposed into several disjoint parts such that each part is either completely in a row or completely in a column, and further column-parts or row-parts are non-conflicting using the previous terminology.
%\Michael{Rather than talk about  increasing numbers, can  we just say column-parts are non-conflicting and row-parts are non-conflicting, using our previous terminology?}\Saeed{Done. Does it look better now?}

\input{figs/vh}

Grid packing therefore reduces to the 1-dimensional variant of grid packing, \textit{array packing}, as follows. For the array packing problem, an array of length $m$ is given as input (with no numbers on it). Our goal is to define segments (this time all horizontal), while keeping the maximum number of segments covering each cell small. After we fix our solution, the adversary puts non-negative numbers on the array cells. For any fixed interval $[x,y]$ we would like to have a segment completely in that interval whose score is at least a fraction of the score of that interval. More precisely, a solution for array packing is $(\alpha,\beta)$-approximate if it covers each cell at most $\alpha$ times and the score of any interval over the maximum score of a segment inside it is bounded by $\beta$. Similar to grid packing, we are not aware of the numbers when giving a solution, and the adversary is aware of our solution before deciding which numbers to put and which interval to choose for the comparison.

%\Michael{You tend to switch back and forth between saying $(\alpha,\beta)$-approximate, $(\alpha,\beta)$ approximation, $(\alpha,\beta)$ solution, etc.  Try to pick one and stick with it, or otherwise make clear what your terminology is.}\Saeed{Fixed.}
 An $(\alpha,\beta)$-approximate solution for array packing yields a $(2\alpha,2\beta)$-approximate solution for grid packing as follows. We treat each row and each column of the grid as an array and make a separate solution for the corresponding array packing instance. After the adversary puts the numbers on the grid, any path from bottom left to the top right can be divided into disjoint column intervals or row intervals, one of which provides us a $2$ approximate solution for the score of the grid. Finally, the guarantee of array packing enables us to prove that the above solution is  $(2\alpha,2\beta)$-approximate for grid packing.

Thus, all that remains is to provide a solution for array packing. To begin, for each cell of the array, there should be one segment covering only that cell. Otherwise, there is no way to compete with an adversary that puts $1$ on that cell and 0 on the other cells and uses that cell for the chosen interval. Thus, $m$ segments of length $1$ for the $m$ cells of the array is an inevitable part of any solution. A first idea to extend this construction is to put segments of length $2$ on every other cell of the array, giving $m/2$ segments covering all of the array cells,
and continuing further, for any $1 \leq i \leq \log m$, we use $m/2^i$ segments of length $2^i$ to cover all cells of the array. While with this construction at most $\log m$ segments cover each cell, the best guarantee that we can expect for such a solution in terms of score is a $1/\Omega(\log m)$ fraction of the score of each interval. That is, such a solution is only $(O(\log m), O(\log m))$-approximate.

To improve the approximation factor to a constant, we make $m^{\kappa}$ copies of each set of segments. Roughly speaking, for segments of length $2^i$ we make $m^{\kappa}$ copies by right shifting the segments by $2^i/m^\kappa$ cells each time (see Section \ref{sec:grid} for more details about this construction and edges cases such as when $2^i < m^{\kappa}$). %\Michael{This is not clear.  Are we using $m^\kappa$ segments in total, or replacing each one of the $m/2^i$ segments with $m^\kappa$ segments?  What happens when $2^i < m^\kappa$, which is almost always the case?}\Saeed{This construction is a bit messy. Perhaps we can just say the ideas and leave the details for the technical section? (Section \ref{sec:grid})}
This clearly adds a multiplicative overhead of $m^{\kappa}$ for the number of segments covering each cell. However, as we show in Section \ref{sec:grid} it improves the second parameter of the approximation guarantee down to $O(1/\kappa)$ from $O(\log m)$.
%To see this, assume that an interval $[\alpha,\beta]$ is chosen by the adversary and we would like to prove that the score of one of the segments within this interval is an $\Omega(\kappa)$ fraction of the score of the interval. Due to our construction, we can cover this interval with two segments completely in the interval such that the remaining uncovered part is as small as $O(\frac{\beta-\alpha+1}{m^\kappa})$. Thus, if we recurse on this procedure, we can cover the entire interval with $O(1/\kappa)$ many segments completely in this interval. Therefore, the score of one of them should be at least an $\Omega(\kappa)$ fraction of the score of the interval itself.

%\Michael{If you're coming back to this proof later, this is far too much detail.  If you're not, it's not enough detail.}\Saeed{I cut a lot of the details. We are giving a formal proof for this in Section \ref{sec:grid}}

\vspace{0.2cm}
{\noindent \textbf{Theorem}~\ref{theorem:grid-cover}, [restated informally]. \textit{For any $0 < \kappa < 1$, the grid packing problem on an $m \times m$ grid admits an $(\tilde O(m^\kappa),O(1/\kappa))$-approximate solution.\\}}

Grid packing is a very strong tool for approximating \textsf{LIS}. For example, consider an array $a$ of length $n$ for which we wish to design a block-based dynamic algorithm for \textsf{LIS}. We fix a constant $0 < \kappa < 1$ and set $f(n) = \tilde O(n^{1+\kappa})$.  Let $m = n^{1/3}$ be the size of the grid we construct for this array. The horizontal thresholds are set in a way that separate the elements into $m$ different pieces each containing roughly $n/m$ elements. That is, the first threshold is the value of the $n/m$'th element of the array after sorting the numbers and so on. The vertical thresholds are set to divide the elements into $m$ different parts each containing $n/m$ elements. That is, the first part contains the first $n/m$ elements of the array and so on. This way, each element of the array corresponds to one unique cell of the grid.
%\Michael{The above paragraph is confusing.  You've explained grids. Are you saying something different here?  if not, cut everything from "the vertical thresholds" on.}  \Saeed{Yes, here we are showing how grid packing can be used for a dynamic algorithm. The previous time, we just gave intuition on how grid packing related to LIS in general.}

The most important property of this division is that every row or every column of the grid contains at most $n/m$ elements. This property is asymptotically maintained for the next $g(n) = n^{2/3}$ operations for which the block-based algorithm is responsible. Obviously, this guarantee also holds for the segments. 
%\Michael{So what's not clear at this point is what happens to segments when you insert or delete an element.  When you insert an element, some segments have to change somehow, or you can't cover everything? You need an explanation of what happens to segments on insertion/deletion;  I don't think you've explained it.  Again, this is too much detail for the introduction, and needs to be separated out.  If you want high-level ideas, fine, but all this should come later.} \Saeed{Yes. These are just the high-level ideas. Everything is explained in more details in Section \ref{sec:constant}. In particular, we talk about the issue you raise in the corresponding section.}
We solve the problem in the following way: first we make a solution for grid packing of size $m \times m$. For each segment, we make a separate instance of the dynamic \textsf{LIS} problem that solves the problem for the elements covered by that segment. Initially, we use the naive algorithm that computes \textsf{LIS} from scratch every time an operation arrives. However, since each segment corresponds to at most $O(n/m)$ elements, the worst-case update time for each segment is $\tilde O(n/m)$. Moreover, each cell is covered by at most $\tilde O(m^\kappa)$ segments which means each operation modifies at most $\tilde O(m^\kappa)$ segments. Thus, the total update time is $\tilde O(m^\kappa n/m) = \tilde O(n^{2/3+\kappa})$. In order to approximate the size of the \textsf{LIS}, every time we run a DP on the segments to find a set of non-conflicting segments whose total size of \textsf{LIS} is maximized. Notice that the \textsf{LIS} of each segment is available in time $O(1)$ (since after each update we store the size of the solution for each segment), and DP takes time $\tilde O(m^{2+\kappa})$ which is basically the total number of segments we have. This is obviously a lower bound on the actual solution size since any partial solution for a non-conflicting set of segments can be combined to obtain a global solution for the union of the elements in all segments. Moreover, the size of the \textsf{LIS} for each segment is definitely an upper bound on the contribution of that segment to the optimal solution. Thus, the solution of the DP is at least an $\Omega(\kappa)$ fraction of the size of the \textsf{LIS} for the entire array.

One thing to keep in mind is that updating the grid requires a more careful analysis. Since the column divisions are based on the indices of the elements, when we add or remove some elements, some columns may grow wider or thinner. While the grid illustration may make it seem challenging to manage such update operations, the actual implementation is straightforward. We define $m-1$ thresholds initially set to factors of $n/m$. Every time an element is added or removed, in addition to updating the binary tree data structure tracking the location of each element,  we can also update the thresholds separating the grid into subgrids..

Since the preprocessing time is $\tilde O(n^{1+\kappa})$ and we run the algorithm for $n/m = O(n^{2/3})$ steps and the worst-case update time for each operation is $\tilde O(n^{2/3+\kappa})$, this block-based algorithm can be turned into a dynamic algorithm for \textsf{LIS} with approximation factor $O(1/\kappa)$ and worst-case update time $\tilde O(n^{2/3+\kappa})$. While this is worse than the solution given in Corollary~\ref{theorem:trivial} both in terms of the approximation factor and worst-case update time, this solution can be extended to improve the update time down to $\tilde O(n^\epsilon)$ for any constant $\epsilon > 0$. All it takes to improve the update time is to replace the naive \textsf{LIS} algorithm of each segment by the more clever algorithm we explained above. While this comes at the expense of a larger approximation factor, the worst-case update time improves. We show in Section \ref{sec:constant} that by setting $\kappa = \Omega(\epsilon)$ and recursing on this algorithm $O(1/\epsilon)$ times we obtain a dynamic algorithm for \textsf{LIS} with worst case update time $\tilde O(n^{\epsilon})$ and approximation factor $O((1/\epsilon)^{O(1/\epsilon)})$.

\vspace{0.2cm}
{\noindent \textbf{Theorem}~\ref{theorem:constant}, [restated informally]. \textit{For any constant $\epsilon > 0$, there exists an algorithm for dynamic \textsf{LIS} whose worst-case update time is $\tilde O(n^{\epsilon})$ and whose approximation factor is $O((1/\epsilon)^{O(1/\epsilon)})$.\\}}

It follows from Theorem~\ref{theorem:constant} that after reporting the estimated value of the solution, we can also determine the corresponding sequence in time proportional to its size. More precisely, after using DP to construct a global solution based on partial solutions of the segment, we can find out which segments contribute to such a solution and recursively recover the corresponding increasing subsequences of the relevant segments. To this end, in addition to the DP table which we use for constructing a global solution, we also store which segments contribute to such a solution. This way, the runtime required for determine the corresponding increasing subsequence is proportional to the size of the solution.

\begin{remark}
	After reporting a solution of size $x$ by our dynamic \textsf{LIS} algorithm, our algorithm is able to report an increasing subsequence of length $x$ in time $O_\epsilon(x)$.
\end{remark}

For the special case of $\mathsf{LIS}^+$ where only insertion operations are supported, we improve the approximation factor down to $O(1/\epsilon)$. Moreover, if one favors the update time over the approximation factor, we show that the update time can be reduced to polylogarithmic if we allow the approximation factor to be $O(\log n)$.

\subsubsection{Another Example: Advisory Help}
To illustrate the effectiveness of the grid packing technique, we bring yet another example, this time in the context of streaming algorithms. It has been shown that \textsf{LIS} can be approximated within a factor of $1+\epsilon$ in the streaming model with memory $O(\sqrt{n})$~\cite{DBLP:conf/soda/GopalanJKK07}. Moreover, matching lower bounds are also provided by G{\'{a}}l and Gopalan~\cite{DBLP:conf/focs/GalG07}. They show that it is impossible to beat the $\Omega(\sqrt{n})$ barrier with any deterministic algorithm that runs in a constant number of rounds and obtains a constant factor approximation. We show that this can be improved with a randomized algorithm that reads the input in a particular order. This notion is called advisory help and has been previously studied~\cite{DBLP:conf/esa/CormodeMT10} to provide graph algorithms in the streaming model.

In such a setting, we design a streaming algorithm but we ask the adversary to give us the input in a particular order. To avoid losing information, elements come in the form $(i,a_i)$ which specifies both the position and the value of each element. We show that in three rounds we can obtain an $O(1/\kappa)$ approximation with memory $\tilde O(n^{2/5+\kappa})$. Roughly speaking, in the first round we sample $m=n^{1/5}$ elements from the array and we set horizontal lines of the grid based on their values. Vertical lines just evenly divide the elements based on their indices into portions of size $n/m$.

\input{figs/order}

In the second round, we ask the adversary to give us the elements of the array but in the row order (shown in Figure \ref{fig:order}). In this round, we compute the solution for horizontal segments. Each segment contains at most $\tilde O(n/m) = \tilde O(n^{4/5})$ element with high probability and therefore using the algorithm of~\cite{DBLP:conf/soda/GopalanJKK07} we can approximate its \textsf{LIS} within factor $1+\epsilon$ with memory $\tilde O(n^{2/5})$. Since each cell is covered by at most $\tilde O(m^\kappa)$ segments, at each step we solve the problem for at most $\tilde O(m^{\kappa})$ segments simultaneously which adds an overhead of $\tilde O(m^\kappa) = \tilde O(n^{\kappa})$ to the memory of the algorithm. Thus the overall memory is bounded by $\tilde O(n^{2/5+\kappa})$. The third round solves the problem for vertical segments similar to the horizontal ones. The only difference is that this time we ask the adversary to give us the elements in the column order. Once all the solutions for all segments are available, we run a DP with memory $\tilde O(m^{2+\kappa}) = \tilde O(n^{2/5+\kappa})$ to approximate the final solution size.  Notice that this solution is not refuted by the impossibility result of~\cite{DBLP:conf/soda/GopalanJKK07} since it both uses randomization and extra help from the adversary. In order for the adversary to provide us the array elements in this particular order, she may need to sort the numbers based on their values. However, sorting does not overly simplify the problem. Computing the \textsf{LIS} of an array is equally hard if the elements are given in the sorted order!

\subsection{Distance to Monotonicity}
Additionally, we present a dynamic algorithm for \textsf{DTM}. Distance to monotonicity seems to be more tractable than \textsf{LIS} since previous work obtain much more efficient algorithms for \textsf{DTM} than \textsf{LIS}~\cite{DBLP:conf/soda/AndoniN10,DBLP:conf/soda/GopalanJKK07,DBLP:conf/soda/NaumovitzSS17}. In particular, there are several known techniques for approximating \textsf{DTM} within a constant factor. As an example, one can model the problem with a graph containing $n$ vertices each corresponding to an element. There is an edge between two vertices, if the corresponding elements are not increasing. While in this interpretation, \textsf{LIS} is equivalent to the largest independent set of the graph, \textsf{DTM} translates to vertex cover which can be approximated within a factor of $2$ by maintaining a maximal matching. As part of our algorithm, we show that such a maximal matching can be maintained with worst-case update time $O(\log ^2 n)$ which yields a dynamic algorithm for \textsf{DTM} with approximation factor $2$ and worst-case update time $O(\log ^2 n)$.

However, we further strengthen this result by improving the approximation factor down to $1+\epsilon$ while keeping the update time intact. The heart of our improvement is based on an exact algorithm for computing \textsf{DTM} when an approximate solution is available. We show in Section \ref{sec:dtm} that given random access to the elements of an array $a$ of size $n$ and a constant approximate solution of size $k$ for the array, one can compute an exact solution for \textsf{DTM} in time $O(k \log n)$. Just knowing the size of the approximate solution does not suffice here; our algorithm requires random access to the elements of the approximate solution as well.

\vspace{0.2cm}
{\noindent \textbf{Lemma}~\ref{lemma:dtm}, [restated informally]. \textit{Let $a$ be an array of length $n$ and $S$ be a set of $k$ elements whose removal from $a$ makes $a$ increasing. One can compute the distance to monotonicity of $a$ in time $O(k \log n)$.\\}}

The above algorithm, in addition to the 2-approximate solution, yields a $1+\epsilon$ approximation block-based algorithm for \textsf{DTM}. Starting from an array $a$ and provided access to a $2$-approximate solution, we set $f(a) = \tilde O(k)$ where $k$ is the size of the approximate solution. Here, $f,g,$ and $h$ do not depend on $n$ since the runtimes depend on the solution size. See Section \ref{sec:amortized} for more information. Moreover, $g(a) = k\epsilon/2$ and $h(a) = O(1)$. After computing an exact solution via the algorithm of Lemma~\ref{lemma:dtm} in the preprocessing phase, we keep reporting $d+i$ as an estimate for \textsf{DTM} for the $i$'th operation where $d$ is the solution for the initial array. The block-based algorithm then can be used to obtain a dynamic algorithm with worst-case update time $O(\log ^2 n)$.

\vspace{0.2cm}
{\noindent \textbf{Theorem}~\ref{theorem:dtm}, [restated informally]. \textit{For any constant $\epsilon > 0$, there exists an algorithm for dynamic \textsf{DTM} whose worst-case update time is $O(\log^2 n)$ and whose approximation factor is $1+\epsilon$.\\}}

Although our method is simple, it has a nice implication for classic algorithms. We show that using Lemma~\ref{lemma:dtm}, we can approximate \textsf{DTM} in time $O(n)$ within an approximation factor $1+\epsilon$ (a log is shaved from the runtime by incurring a factor $1+\epsilon$ to the approximation guarantee). This result is tight in two ways: i) Any constant factor approximation algorithm for \textsf{DTM} has to make at least $\Omega(n)$ value queries to the elements of the array to solve the case that the solution is either 0 or 1. ii) Any exact solution which is comparison based or based on algebraic decisions trees has a runtime of at least $\Omega(n \log n)$~\cite{fredman1975computing,ramanan1997tight}.

To achieve this, we first compute a 2-approximate solution for \textsf{DTM} in time $O(n)$ (see Section \ref{sec:dtmclassic} for more details). If the size of the solution is smaller than $\sqrt{n}$, we use the algorithm of Lemma \ref{lemma:dtm} to obtain an exact solution in linear time. Otherwise, we use the algorithm of ~\cite{DBLP:conf/soda/NaumovitzSS17} to obtain a $1+\epsilon$ approximate solution in time $O(n)$\footnote{The runtime of the algorithm given in~\cite{DBLP:conf/soda/NaumovitzSS17} is $\tilde O(n/d+\sqrt{n})$ when the solution size is lower bounded by $d$.}.% Also, if a stronger type of queries for the array is available one can improve the runtime down to $\tilde O(\sqrt{n})$. In such queries, we give two indices $i$ and $j$ of the array and the oracle tells us whether all the elements in this range are sorted or otherwise it reports a conflicting pair.

\newpage
\section{From Amortized Update Time to Worst Case Update Time}\label{sec:amortized}
The goal of this section is to show a reduction that simplifies the problem with respect to worst-case time constraints. Ultimately, in our algorithms, we prove that the update time of each operation is bounded in the worst case. %\Michael{hyphen note: "worst case" is a noun; worst-case is an adjective.  I'll fix where I see...}\Saeed{Got it. Will fix the rest of the sections accordingly.}
However, it is more convenient to allow for larger update times in some cases, while keeping a bounded amortized update time. In this section, we prove that if our algorithms adhere to a certain structure, then worst-case update time reduces to amortized update time. Later in the section, we present a motivating example to show how the reduction enables us to simplify the proofs. In our example, we seek to design a $1+\epsilon$ approximation algorithm for dynamic \textsf{LIS} with worst-case update time $\tilde O(\sqrt{n})$.
%\Michael{If we're proving it here remove all but a sketch earlier.}\Saeed{Yes, we are giving a complete proof here. Will remove the unnecessary details from Section 2.}

In our framework, we start with an array $a$ of size $n$ and our algorithm is allowed to make a preprocessing of time $f(a)$. For the next $g(a)$ steps, the processing time of each operation is bounded by $h(a)$ in the worst case. After $g(a)$ steps, our algorithm is no longer responsible for the operations and terminates. We refer to such an algorithm as \textit{block-based}. Note $f(a)$, $g(a)$, and $h(a)$ are not necessarily determined based only on the size $n$ of the array $a$. For example, in the algorithm of Section \ref{sec:dtm}, the values of these functions are proportional to the solution size for $a$, not  $n$. However, when these functions are only dependent on $n$, we may drop $a$ in the notation and use $n$ instead. Functions $f,g,$ and $h$ should meet one important property: after applying $g(a)$ arbitrary operations to an array $a$ and obtaining a new array $a'$, $f(a'), g(a')$ and $h(a')$ should not change asymptotically.  More specifically, although the reduction holds when the values are within any constant factor, we assume $1/2 \leq f(a)/f(a'), g(a)/g(a'), h(a)/h(a') \leq 2$. We call this property \textit{relativity}. We also assume without loss of generality that $f,g,$ and $h$ are always lower bounded by a constant (say $20$) and when the array size is constant, so are the values for $f(a), g(a),$ and $h(a)$.

 We show in the following that a block-based algorithm $\mathcal{A}$ for \textsf{LIS} or \textsf{DTM} with identifiers $\langle f,g,h \rangle$ can be used as a black box to obtain a dynamic algorithm $\mathcal{A'}$ with worst-case update time $O(\max\{h(a), f(a) / g(a)\})$. The approximation factor of the algorithm is preserved in this reduction. To show a use case of this technique, we provide a simple analysis of a $(1+\epsilon)$-approximate dynamic algorithm for \textsf{LIS} with worst-case update time $\tilde O(\sqrt{n})$. 
 
\begin{lemma}\label{lemma:reduction}
	Let $\mathcal{A}$ be a block-based algorithm with preprocessing time $f(a)$ that approximates dynamic \textsf{LIS} or dynamic \textsf{DTM} for up to $g(a)$ many steps with worst-case update time $h(a)$. If $\langle f,g,h\rangle$ satisfies relativity then there exists a dynamic algorithm $\mathcal{A}'$ for the same problem whose worst-case update time is bounded by $O(\max\{h(a), f(a) / g(a)\})$ and whose approximation factor is the same as $\mathcal{A}$.
\end{lemma}
\begin{proof}
\begin{figure}[ht]

\tikzset{every picture/.style={line width=0.75pt}} %set default line width to 0.75pt        

\begin{tikzpicture}[x=0.75pt,y=0.75pt,yscale=-1,xscale=1]
%uncomment if require: \path (0,197); %set diagram left start at 0, and has height of 197

%Shape: Rectangle [id:dp9378736597768127] 
\draw   (39,29) -- (301.5,29) -- (301.5,52) -- (39,52) -- cycle ;
%Straight Lines [id:da7598911370300216] 
\draw [color={rgb, 255:red, 208; green, 2; blue, 27 }  ,draw opacity=1 ] [dash pattern={on 4.5pt off 4.5pt}]  (239,29) -- (239,79) ;
\draw [shift={(239,81)}, rotate = 270] [fill={rgb, 255:red, 208; green, 2; blue, 27 }  ,fill opacity=1 ][line width=0.75]  [draw opacity=0] (10.72,-5.15) -- (0,0) -- (10.72,5.15) -- (7.12,0) -- cycle    ;

%Shape: Rectangle [id:dp862665075058989] 
\draw   (239,88) -- (501.5,88) -- (501.5,111) -- (239,111) -- cycle ;
%Straight Lines [id:da5148161774301114] 
\draw [color={rgb, 255:red, 208; green, 2; blue, 27 }  ,draw opacity=1 ] [dash pattern={on 4.5pt off 4.5pt}]  (439,32) -- (439,138) ;
\draw [shift={(439,140)}, rotate = 270] [fill={rgb, 255:red, 208; green, 2; blue, 27 }  ,fill opacity=1 ][line width=0.75]  [draw opacity=0] (10.72,-5.15) -- (0,0) -- (10.72,5.15) -- (7.12,0) -- cycle    ;

%Shape: Rectangle [id:dp6160675512833309] 
\draw   (436,143) -- (698.5,143) -- (698.5,166) -- (436,166) -- cycle ;
%Straight Lines [id:da7333810832609753] 
\draw [color={rgb, 255:red, 189; green, 16; blue, 224 }  ,draw opacity=1 ][fill={rgb, 255:red, 189; green, 16; blue, 224 }  ,fill opacity=1 ][line width=1.5]    (246,41) -- (296.5,41) ;
\draw [shift={(299.5,41)}, rotate = 180] [color={rgb, 255:red, 189; green, 16; blue, 224 }  ,draw opacity=1 ][line width=1.5]    (14.21,-6.37) .. controls (9.04,-2.99) and (4.3,-0.87) .. (0,0) .. controls (4.3,0.87) and (9.04,2.99) .. (14.21,6.37)   ;
\draw [shift={(243,41)}, rotate = 0] [color={rgb, 255:red, 189; green, 16; blue, 224 }  ,draw opacity=1 ][line width=1.5]    (14.21,-6.37) .. controls (9.04,-2.99) and (4.3,-0.87) .. (0,0) .. controls (4.3,0.87) and (9.04,2.99) .. (14.21,6.37)   ;
%Straight Lines [id:da4359865948636126] 
\draw [color={rgb, 255:red, 189; green, 16; blue, 224 }  ,draw opacity=1 ][fill={rgb, 255:red, 189; green, 16; blue, 224 }  ,fill opacity=1 ][line width=1.5]    (446,100) -- (496.5,100) ;
\draw [shift={(499.5,100)}, rotate = 180] [color={rgb, 255:red, 189; green, 16; blue, 224 }  ,draw opacity=1 ][line width=1.5]    (14.21,-6.37) .. controls (9.04,-2.99) and (4.3,-0.87) .. (0,0) .. controls (4.3,0.87) and (9.04,2.99) .. (14.21,6.37)   ;
\draw [shift={(443,100)}, rotate = 0] [color={rgb, 255:red, 189; green, 16; blue, 224 }  ,draw opacity=1 ][line width=1.5]    (14.21,-6.37) .. controls (9.04,-2.99) and (4.3,-0.87) .. (0,0) .. controls (4.3,0.87) and (9.04,2.99) .. (14.21,6.37)   ;
%Straight Lines [id:da2682301879425284] 
\draw  [dash pattern={on 0.84pt off 2.51pt}]  (116.45,134.57) -- (272.5,100) ;
\draw [shift={(272.5,100)}, rotate = 347.51] [color={rgb, 255:red, 0; green, 0; blue, 0 }  ][fill={rgb, 255:red, 0; green, 0; blue, 0 }  ][line width=0.75]      (0, 0) circle [x radius= 3.35, y radius= 3.35]   ;
\draw [shift={(114.5,135)}, rotate = 347.51] [fill={rgb, 255:red, 0; green, 0; blue, 0 }  ][line width=0.75]  [draw opacity=0] (8.93,-4.29) -- (0,0) -- (8.93,4.29) -- cycle    ;
%Straight Lines [id:da9562644571915735] 
\draw  [dash pattern={on 0.84pt off 2.51pt}]  (116.5,135.12) -- (469.5,156) ;
\draw [shift={(469.5,156)}, rotate = 3.39] [color={rgb, 255:red, 0; green, 0; blue, 0 }  ][fill={rgb, 255:red, 0; green, 0; blue, 0 }  ][line width=0.75]      (0, 0) circle [x radius= 3.35, y radius= 3.35]   ;
\draw [shift={(114.5,135)}, rotate = 3.39] [fill={rgb, 255:red, 0; green, 0; blue, 0 }  ][line width=0.75]  [draw opacity=0] (8.93,-4.29) -- (0,0) -- (8.93,4.29) -- cycle    ;
%Shape: Rectangle [id:dp3838283745049702] 
\draw  [fill={rgb, 255:red, 155; green, 155; blue, 155 }  ,fill opacity=1 ] (20,124) -- (64.5,124) -- (64.5,147) -- (20,147) -- cycle ;
%Shape: Rectangle [id:dp19287683682351742] 
\draw  [fill={rgb, 255:red, 255; green, 255; blue, 255 }  ,fill opacity=1 ] (65,124) -- (109.5,124) -- (109.5,147) -- (65,147) -- cycle ;
%Straight Lines [id:da501853240981107] 
\draw [line width=0.75]  [dash pattern={on 4.5pt off 4.5pt}]  (301.5,56) -- (301.5,116) ;

%Straight Lines [id:da35773891795656665] 
\draw [line width=0.75]  [dash pattern={on 4.5pt off 4.5pt}]  (501.5,30) -- (501.5,167) ;

%Straight Lines [id:da890848584412387] 
\draw    (94.5,138) -- (94.5,167) ;
\draw [shift={(94.5,169)}, rotate = 270] [fill={rgb, 255:red, 0; green, 0; blue, 0 }  ][line width=0.75]  [draw opacity=0] (8.93,-4.29) -- (0,0) -- (8.93,4.29) -- cycle    ;

%Straight Lines [id:da4780426360771186] 
\draw    (42.5,113) -- (42.5,142) ;

\draw [shift={(42.5,111)}, rotate = 90] [fill={rgb, 255:red, 0; green, 0; blue, 0 }  ][line width=0.75]  [draw opacity=0] (8.93,-4.29) -- (0,0) -- (8.93,4.29) -- cycle    ;

% Text Node
\draw (23,39) node   {$\mathcal{B}_{1}$};
% Text Node
\draw (41,66) node [scale=0.9]  {$a^{(1)}$};
% Text Node
\draw (223,98) node   {$\mathcal{B}_{2}$};
% Text Node
\draw (420,153) node   {$\mathcal{B}_{3}$};
% Text Node
\draw (556,153) node [scale=2.488]  {$\dotsc $};
% Text Node
\draw (67,96) node  [align=left] {preprocessing};
% Text Node
\draw (132,176) node  [align=left] {updating two operations in each step};
% Text Node
\draw (269,62) node [scale=0.7,color={rgb, 255:red, 144; green, 19; blue, 254 }  ,opacity=1 ]  {$g(a^{(1)}) /10$};
% Text Node
\draw (470,121) node [scale=0.7,color={rgb, 255:red, 144; green, 19; blue, 254 }  ,opacity=1 ]  {$g(a^{(2)}) /10$};
% Text Node
\draw (242,126) node [scale=0.9]  {$a^{(2)}$};
% Text Node
\draw (439,181) node [scale=0.9]  {$a^{(3)}$};
% Text Node
\draw (43,15) node [scale=0.7]  {$1$};
% Text Node
\draw (240,15) node [scale=0.7]  {$\frac{9g(a^{(1)})}{10}$};
% Text Node
\draw (301,13) node [scale=0.7]  {$g(a^{(1)})$};
% Text Node
\draw (512,15) node [scale=0.7]  {$\frac{9g(a^{(1)})}{10} +g(a^{(2)})$};
% Text Node
\draw (428,15) node [scale=0.7]  {$\frac{9( g(a^{(1)}) +g(a^{(2)}))}{10}$};

\end{tikzpicture}

\caption{The reduction is shown in this figure.} \label{fig:reduction}
\end{figure}
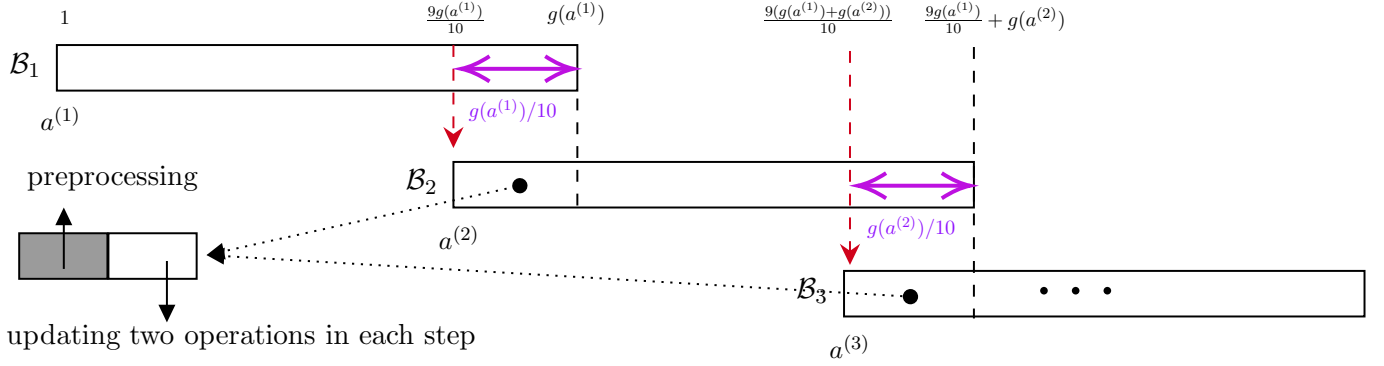

%\Saeed{The figure needs minor modifications.}

%\Saeed{We should distinguish between the step number and the size of the array in the figure.}

Figure \ref{fig:reduction} gives a pictorial depiction of the proof idea. We construct an algorithm $\mathcal{A'}$ in the following way: $\mathcal{A}'$ uses algorithm $\mathcal{A}$ repeatedly. To distinguish between multiple instances of $\mathcal{A}$, we add subscripts; the first time we use algorithm $\mathcal{A}$ we call it $\mathcal{B}_1$. Every $\mathcal{B}_i$ is basically a copy of the block-based algorithm $\mathcal{A}$ which is modified slightly to execute the preprocessing part in multiple steps. We begin with using our block-based algorithm $\mathcal{B}_1$ at step 1. At this point we call the initial array (which is empty) $a^{(1)}$. Since the size of the array is constant, so is the preprocessing time and therefore we can ignore it when bounding the time complexity. For $g(a^{(1)})$ many steps, we use algorithm $\mathcal{B}_1$ to preserve an approximate solution and from then on, we use a separate algorithm for the rest of the operations, namely $\mathcal{B}_2$. The construction of $\mathcal{B}_2$ is given below:
%\Michael{This notation is bad.  I think that $\mathcal{A}_2$  is just another  run of  $\mathcal{A}_1$, but it's not.  Let's call it $\mathcal{B}_2$ or something, and explain that  the $\mathcal{B}_i$ run two copies of  $\mathcal{A}$ in parallel.}\Saeed{Done. Added a sentence above to clarify the similarity and differences between $\mathcal{B}_i$ and $\mathcal{A}$.}

%\Michael{The below is confusing  because it's not clear what $\mathcal{B}_2$ is;  best to  change  notation. }
When $\mathcal{B}_1$ has gone $9/10$ of the way and is only responsible for $g(a^{(1)})/10$ more operations, we initiate algorithm $\mathcal{B}_2$. Let $a^{(2)}$ be the array at this point. $\mathcal{B}_2$ needs to run the preprocessing step which requires $f(a^{(2)})$ many operations. This may not be possible in a single step, therefore, we break the computation into $g(a^{(1)})/20$ pieces and execute each piece in the next $g(a^{(1)})/20$ steps. Moreover, in the next $g(a^{(1)})/20$ steps operations that arrive after the construction of $\mathcal{B}_2$ are processed: two operations in each step. While this is happening, algorithm $\mathcal{B}_1$ processes the operations and updates the solution size. When we reach $g(a^{(1)})/10$ many steps after the construction of $\mathcal{B}_2$, algorithm $\mathcal{B}_2$ has already finished the preprocessing and all the operations that have arrived so far are applied to it.  This is exactly the time that $\mathcal{B}_1$ terminates, and from then on, we use algorithm $\mathcal{B}_2$ to process each operation.

Similarly, $\mathcal{B}_3$ is constructed when $\mathcal{B}_2$ has applied $g(a^{(2)}) 9/10$ operations. This construction goes on as long as operations arrive.

The correctness of $\mathcal{A}'$ follows from that of $\mathcal{A}$. Therefore, any approximation factor that $\mathcal{A}$ guarantees for us also carries over to $\mathcal{A}'$. For the update time, the construction guarantees that at every time, at most two instances of algorithm $\mathcal{A}$ are active. At every step, in each algorithm, we either perform an operation or we are initializing the algorithm in which case the update time %\Michael{Not clear what this means -- runtime  = worst-case step time?  Runtime typically refers to full run, not a step.  I think you mean "update time", time  todo  an update?}\Saeed{That's correct. Replaced runtime by update time. Will make that change globally.}
is bounded by $O(\max\{f(a)/g(a),h(a)\})$. 
%\Michael{Worth adding a sentence that there are no problems if operation "stop" during the middle of one of these blocks-- but why is that the case?}

One thing to keep in mind in the above argument is that because of relativity the value of functions $f,g,$ and $h$ remain asymptotically the same during two consecutive runs of algorithm $\mathcal{A}$. Thus, $g(a)$ is within a constant factor for two consecutive runs of $\mathcal{A}$ and therefore $f(a^{(i)})/g(a^{(i-1)})$ is asymptotically the same as $f(a^{(i)})/g(a^{(i)})$.
\end{proof}
 
We emphasize that there is a constant-factor overhead in the update-time of the reduction which is hidden in the $O$ notation.
 
To illustrate the effectiveness of our reduction, we bring a motivating example to show how it simplifies the design of dynamic algorithms for \textsf{LIS}.
 
 \subsection{Warm Up: Block-based Algorithm for \textsf{LIS}}\label{sec:warmup}
 The reduction of Section \ref{sec:amortized} gives us a very convenient framework to design dynamic algorithms for  \textsf{LIS} and \textsf{DTM}. Here we bring a simple block-based algorithm for dynamic \textsf{LIS} with approximation factor $1+\epsilon$ that results in an algorithm with the same approximation factor and worst-case update time $\tilde O(\sqrt{n})$.
 
 Since in the following, functions $f,g$, and $h$ only depend on the size of the array in this case, we write them in terms of $n$.
 
 \begin{lemma}\label{lemma:trivial}
 	For any $\epsilon > 0$, there exists a block-based algorithm for \textsf{LIS} with approximation factor $1+\epsilon$, preprocessing time $f(n) = \tilde O(n)$ and update time $h(n) = \tilde O(\sqrt{n})$ that maintains an approximate solution to \textsf{LIS} for up to $g(n) = \sqrt{n}$ many operations.
 \end{lemma}
 \begin{proof}
 	In the preprocessing phase, we first compute the longest increasing subsequence in time $\tilde O(n)$. Let the solution size be $r$. Based on the value of $r$, we consider two different strategies: If $r$ is already at least $2\sqrt{n}/\epsilon$, in the next $\sqrt{n}$ steps, we do not make any changes to the array and output $r-i$ after $i$'th operation. Since each operation hurts the solution size by an additive factor of at most $1$, our solution is always valid and has approximation factor $1+\epsilon$ throughout this process.
 	
 	Otherwise, we use the algorithm of Chen \textit{et al.}~\cite{chen2013dynamic} to update the solution in every step. The setup cost for the algorithm of Chen \textit{et al.}~\cite{chen2013dynamic} is $\tilde O(n)$ which can be executed in the preprocessing step (since this is not explicitly mentioned in~\cite{chen2013dynamic}, we bring more details about their algorithm in Appendix~\ref{sec:chen-ap}). Moreover, the solution size is initially upper bounded by $2\sqrt{n}/\epsilon$ and it can grow to at most $2\sqrt{n}/\epsilon+\sqrt{n}$ after $\sqrt{n}$ operations. Therefore, the update time remains $\tilde O(\sqrt{n})$ in the worst case.
 \end{proof}

 Based on Lemma \ref{lemma:reduction}, a dynamic algorithm for \textsf{LIS} can be implemented with worst-case update time $\tilde O(\sqrt{n})$ and approximation factor $1+\epsilon$. 
 \begin{theorem}\label{theorem:trivial}[a corollary of Lemmas \ref{lemma:reduction} and \ref{lemma:trivial}]
 	For any constant $\epsilon > 0$, there exists a dynamic algorithm for \textsf{LIS} with approximation factor $1+\epsilon$ and update time $\tilde O(\sqrt{n})$ in the worst case.
 \end{theorem}

\newpage
\section{Grid Packing}\label{sec:grid}
This section is dedicated to a combinatorial problem which we call \textit{grid packing}. Definitions are given in Section~\ref{sec:results} but for the sake of completeness we restate them here. In this problem, we have a table of size $m \times m$. Our goal is to introduce a number of segments on the table. Each segment either covers a consecutive set of cells in a row or in a column. A segment $A$ \textit{precedes} a segment $B$ if \textbf{every} cell of $A$ is strictly higher than every cell of $B$ and also \textbf{every} cell of $A$ is strictly to the right of every cell of $B$. Two segments are \textit{non-conflicting}, if one of them precedes the other one. Otherwise, we call them \textit{conflicting}.  The segments we introduce can overlap and there is no restriction on the number of segments or the length of each segment. However, we would like to minimize the maximum number of segments that cover each cell. 

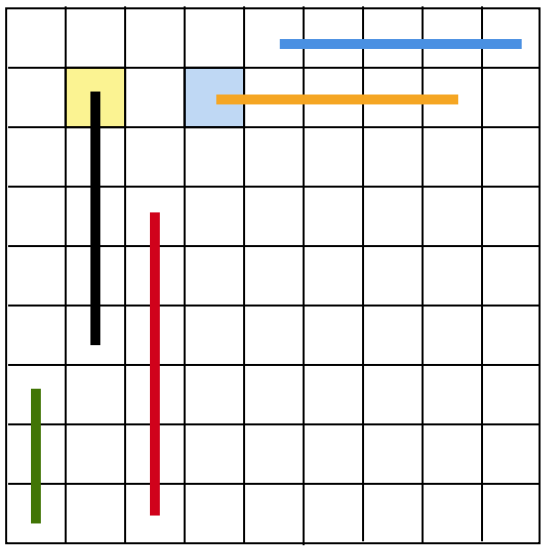
\begin{figure}[ht]

\centering

\tikzset{every picture/.style={line width=0.75pt}} %set default line width to 0.75pt        

\begin{tikzpicture}[x=0.75pt,y=0.75pt,yscale=-1,xscale=1]

%Shape: Rectangle [id:dp5254311647291401] 
\draw   (181,11) -- (450,11) -- (450,281) -- (181,281) -- cycle ;
%Straight Lines [id:da3305762449852494] 
\draw    (301,10) -- (301,281) ;

%Straight Lines [id:da3206498977891954] 
\draw    (331,10) -- (331,282) ;

%Straight Lines [id:da02838390385871148] 
\draw    (361,10) -- (361,280) ;

%Straight Lines [id:da6891763144045235] 
\draw    (211,10) -- (211,281) ;

%Straight Lines [id:da38683484497784937] 
\draw    (241,10) -- (241,281) ;

%Straight Lines [id:da7490078717616151] 
\draw    (271,10) -- (271,281) ;

%Straight Lines [id:da3531030059687228] 
\draw    (391,10) -- (391,281) ;

%Straight Lines [id:da43894336811181467] 
\draw    (421,10) -- (421,280) ;

%Straight Lines [id:da27078630886751465] 
\draw    (182,41) -- (450,41) ;

%Straight Lines [id:da005150497881268201] 
\draw    (182,71) -- (450,71) ;

%Straight Lines [id:da23147705862104506] 
\draw    (182,101) -- (450,101) ;

%Straight Lines [id:da7425549844424117] 
\draw    (182,131) -- (450,131) ;

%Straight Lines [id:da2351732964828932] 
\draw    (182,161) -- (450,161) ;

%Straight Lines [id:da7401347371733205] 
\draw    (182,191) -- (450,191) ;

%Straight Lines [id:da8968761436786192] 
\draw    (182,221) -- (450,221) ;

%Straight Lines [id:da8562431949289493] 
\draw    (182,251) -- (450,251) ;

%Straight Lines [id:da7212412833117718] 
\draw [color={rgb, 255:red, 208; green, 2; blue, 27 }  ,draw opacity=1 ][line width=3.75]    (256,114) -- (256,267) ;

%Straight Lines [id:da7675592926230685] 
\draw [color={rgb, 255:red, 74; green, 144; blue, 226 }  ,draw opacity=1 ][line width=3.75]    (319,29) -- (441,29) ;

%Straight Lines [id:da4112941771444574] 
\draw [color={rgb, 255:red, 65; green, 117; blue, 5 }  ,draw opacity=1 ][line width=3.75]    (196,203) -- (196,271) ;

%Shape: Square [id:dp4777174644258826] 
\draw  [color={rgb, 255:red, 0; green, 0; blue, 0 }  ,draw opacity=1 ][fill={rgb, 255:red, 248; green, 231; blue, 28 }  ,fill opacity=0.48 ] (211,41) -- (241,41) -- (241,71) -- (211,71) -- cycle ;
%Straight Lines [id:da09165043527230288] 
\draw [line width=3.75]    (226,53) -- (226,181) ;

%Shape: Square [id:dp8680283690742998] 
\draw  [color={rgb, 255:red, 0; green, 0; blue, 0 }  ,draw opacity=1 ][fill={rgb, 255:red, 74; green, 144; blue, 226 }  ,fill opacity=0.35 ] (271,41) -- (301,41) -- (301,71) -- (271,71) -- cycle ;
%Straight Lines [id:da0026998791284389423] 
\draw [color={rgb, 255:red, 245; green, 166; blue, 35 }  ,draw opacity=1 ][line width=3.75]    (287,57) -- (409,57) ;

\end{tikzpicture}
\caption{Segments are shown on the grid. The pair (black, orange) is conflicting since the yellow cell (covered by the black segment) is on the same row as the blue cell (covered by the orange segment). The following pairs are non-conflicting: (green, black), (green, orange), (green, blue), (red, orange), (red, blue), (black, blue).} \label{fig:crossing}
\end{figure}

After we choose the segments, an adversary puts a non-negative number on each cell of the grid. The score of a subset of cells of the table would be the sum of their values and the overall score of the table is the maximum score of a path of length $2m-1$ from the bottom left corner to the top right corner. In such a path, we always either move up or to the right.

The score of a segment is the sum of the numbers on the cells it covers. We obtain the maximum sum of the scores of a non-conflicting set of segments.  The score of the table is an upper bound of on the score of any set of non-conflicting segments. We would like to choose segments so that the ratio of the score of the table and our score is bounded by a constant, no matter how the adversary puts the numbers on the table. More precisely, we call a solution $(\alpha,\beta)$-approximate, if at most $\alpha$ segments cover each cell and it guarantees a $1/\beta$ fraction of the score of the table for us for any assignment of numbers to the table cells.

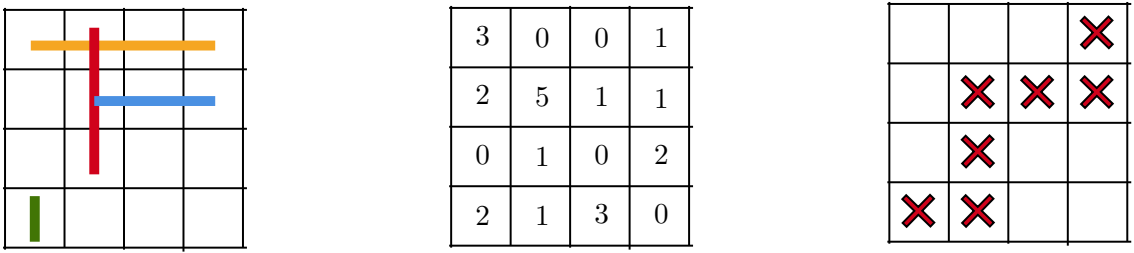
\begin{figure}[ht]

\centering

\tikzset{every picture/.style={line width=0.75pt}} %set default line width to 0.75pt        

\begin{tikzpicture}[x=0.75pt,y=0.75pt,yscale=-1,xscale=1]
%uncomment if require: \path (0,172); %set diagram left start at 0, and has height of 172

%Straight Lines [id:da6633043819092925] 
\draw    (270,30) -- (270,152) ;

%Straight Lines [id:da6069877637104277] 
\draw    (300,31) -- (300,151) ;

%Straight Lines [id:da5789537013543653] 
\draw    (330,30) -- (330,152) ;

%Straight Lines [id:da6178882799434375] 
\draw    (360,30) -- (360,153) ;

%Straight Lines [id:da33725599772480996] 
\draw    (390,30) -- (390,151) ;

%Straight Lines [id:da736940795940934] 
\draw    (269,31) -- (392,31) ;

%Straight Lines [id:da9741470082112216] 
\draw    (269,61) -- (392,61) ;

%Straight Lines [id:da2779143347704205] 
\draw    (269,91) -- (392,91) ;

%Straight Lines [id:da3530720908879399] 
\draw    (269,121) -- (392,121) ;

%Straight Lines [id:da038384016048926606] 
\draw    (269,151) -- (392,151) ;

%Straight Lines [id:da5998474774869911] 
\draw    (45,31) -- (45,153) ;

%Straight Lines [id:da5597571026286936] 
\draw    (75,32) -- (75,152) ;

%Straight Lines [id:da16479280447754507] 
\draw    (105,31) -- (105,153) ;

%Straight Lines [id:da3085712099513582] 
\draw    (135,31) -- (135,154) ;

%Straight Lines [id:da12613347387221086] 
\draw    (165,31) -- (165,152) ;

%Straight Lines [id:da6266120443422065] 
\draw    (44,32) -- (167,32) ;

%Straight Lines [id:da5072506675360497] 
\draw    (44,62) -- (167,62) ;

%Straight Lines [id:da7640683512061104] 
\draw    (44,92) -- (167,92) ;

%Straight Lines [id:da5807436513303619] 
\draw    (44,122) -- (167,122) ;

%Straight Lines [id:da2277681145302426] 
\draw    (44,152) -- (167,152) ;

%Straight Lines [id:da8599665215928698] 
\draw    (491,28) -- (491,150) ;

%Straight Lines [id:da9296198582234105] 
\draw    (521,29) -- (521,149) ;

%Straight Lines [id:da5188209661505319] 
\draw    (551,28) -- (551,150) ;

%Straight Lines [id:da7563699076629549] 
\draw    (581,28) -- (581,151) ;

%Straight Lines [id:da7100636887892002] 
\draw    (611,28) -- (611,149) ;

%Straight Lines [id:da22647026815320848] 
\draw    (490,29) -- (613,29) ;

%Straight Lines [id:da7140685594390213] 
\draw    (490,59) -- (613,59) ;

%Straight Lines [id:da7042967357998233] 
\draw    (490,89) -- (613,89) ;

%Straight Lines [id:da8328482917067017] 
\draw    (490,119) -- (613,119) ;

%Straight Lines [id:da535803664626858] 
\draw    (490,149) -- (613,149) ;

%Straight Lines [id:da2001147987658074] 
\draw [color={rgb, 255:red, 245; green, 166; blue, 35 }  ,draw opacity=1 ][line width=3.75]    (58,50) -- (151,50) ;

%Straight Lines [id:da20733782106809806] 
\draw [color={rgb, 255:red, 208; green, 2; blue, 27 }  ,draw opacity=1 ][line width=3.75]    (90,41) -- (90,115) ;

%Straight Lines [id:da3188339231354045] 
\draw [color={rgb, 255:red, 74; green, 144; blue, 226 }  ,draw opacity=1 ][line width=3.75]    (90,78) -- (151,78) ;

%Straight Lines [id:da6087723399785627] 
\draw [color={rgb, 255:red, 65; green, 117; blue, 5 }  ,draw opacity=1 ][line width=3.75]    (60,126) -- (60,149) ;

%Shape: Cross [id:dp18788217453772793] 
\draw  [color={rgb, 255:red, 0; green, 0; blue, 0 }  ,draw opacity=1 ][fill={rgb, 255:red, 208; green, 2; blue, 27 }  ,fill opacity=1 ] (497.54,127.53) -- (499.49,125.52) -- (505.58,131.45) -- (511.51,125.36) -- (513.51,127.31) -- (507.58,133.4) -- (513.67,139.33) -- (511.72,141.33) -- (505.63,135.41) -- (499.7,141.49) -- (497.7,139.54) -- (503.63,133.45) -- cycle ;
%Shape: Cross [id:dp5876950991142253] 
\draw  [color={rgb, 255:red, 0; green, 0; blue, 0 }  ,draw opacity=1 ][fill={rgb, 255:red, 208; green, 2; blue, 27 }  ,fill opacity=1 ] (527.54,127.53) -- (529.49,125.52) -- (535.58,131.45) -- (541.51,125.36) -- (543.51,127.31) -- (537.58,133.4) -- (543.67,139.33) -- (541.72,141.33) -- (535.63,135.41) -- (529.7,141.49) -- (527.7,139.54) -- (533.63,133.45) -- cycle ;
%Shape: Cross [id:dp6329349377511551] 
\draw  [color={rgb, 255:red, 0; green, 0; blue, 0 }  ,draw opacity=1 ][fill={rgb, 255:red, 208; green, 2; blue, 27 }  ,fill opacity=1 ] (527.54,97.53) -- (529.49,95.52) -- (535.58,101.45) -- (541.51,95.36) -- (543.51,97.31) -- (537.58,103.4) -- (543.67,109.33) -- (541.72,111.33) -- (535.63,105.41) -- (529.7,111.49) -- (527.7,109.54) -- (533.63,103.45) -- cycle ;
%Shape: Cross [id:dp5032748735899488] 
\draw  [color={rgb, 255:red, 0; green, 0; blue, 0 }  ,draw opacity=1 ][fill={rgb, 255:red, 208; green, 2; blue, 27 }  ,fill opacity=1 ] (527.54,67.53) -- (529.49,65.52) -- (535.58,71.45) -- (541.51,65.36) -- (543.51,67.31) -- (537.58,73.4) -- (543.67,79.33) -- (541.72,81.33) -- (535.63,75.41) -- (529.7,81.49) -- (527.7,79.54) -- (533.63,73.45) -- cycle ;
%Shape: Cross [id:dp4781774334975044] 
\draw  [color={rgb, 255:red, 0; green, 0; blue, 0 }  ,draw opacity=1 ][fill={rgb, 255:red, 208; green, 2; blue, 27 }  ,fill opacity=1 ] (557.54,67.53) -- (559.49,65.52) -- (565.58,71.45) -- (571.51,65.36) -- (573.51,67.31) -- (567.58,73.4) -- (573.67,79.33) -- (571.72,81.33) -- (565.63,75.41) -- (559.7,81.49) -- (557.7,79.54) -- (563.63,73.45) -- cycle ;
%Shape: Cross [id:dp6534331657458179] 
\draw  [color={rgb, 255:red, 0; green, 0; blue, 0 }  ,draw opacity=1 ][fill={rgb, 255:red, 208; green, 2; blue, 27 }  ,fill opacity=1 ] (587.54,67.53) -- (589.49,65.52) -- (595.58,71.45) -- (601.51,65.36) -- (603.51,67.31) -- (597.58,73.4) -- (603.67,79.33) -- (601.72,81.33) -- (595.63,75.41) -- (589.7,81.49) -- (587.7,79.54) -- (593.63,73.45) -- cycle ;
%Shape: Cross [id:dp5630155449984502] 
\draw  [color={rgb, 255:red, 0; green, 0; blue, 0 }  ,draw opacity=1 ][fill={rgb, 255:red, 208; green, 2; blue, 27 }  ,fill opacity=1 ] (587.54,37.53) -- (589.49,35.52) -- (595.58,41.45) -- (601.51,35.36) -- (603.51,37.31) -- (597.58,43.4) -- (603.67,49.33) -- (601.72,51.33) -- (595.63,45.41) -- (589.7,51.49) -- (587.7,49.54) -- (593.63,43.45) -- cycle ;

% Text Node
\draw (286,136) node   {$2$};
% Text Node
\draw (286,105) node   {$0$};
% Text Node
\draw (286,75) node   {$2$};
% Text Node
\draw (286,45) node   {$3$};
% Text Node
\draw (316,136) node   {$1$};
% Text Node
\draw (316,106) node   {$1$};
% Text Node
\draw (316,76) node   {$5$};
% Text Node
\draw (316,46) node   {$0$};
% Text Node
\draw (346,76) node   {$1$};
% Text Node
\draw (346,105) node   {$0$};
% Text Node
\draw (346,135) node   {$3$};
% Text Node
\draw (376,135) node   {$0$};
% Text Node
\draw (376,105) node   {$2$};
% Text Node
\draw (376,77) node   {$1$};
% Text Node
\draw (346,46) node   {$0$};
% Text Node
\draw (376,46) node   {$1$};

\end{tikzpicture}

\caption{After we introduce the segments (left figure), the adversary puts the numbers on the table (middle figure). In this case, the score of the table is equal to $12$ (via the path depicted on the right figure), and our score is equal to $9$ obtained from two non-conflicting segments green and blue.} \label{fig:crossing}
\end{figure}

In this section, we prove the following theorem: For any grid $m \times m$ and any $0 < \kappa < 1$, there exists a grid packing with  guarantee $(O(m^\kappa \log m),O(1/\kappa))$. That is, each cell is covered by at most $O(m^\kappa \log m)$ segments and the ratio of the table's score over our score is bounded by $O(1/\kappa)$ in the worst case. This solution is constructive and our proof also gives us the segments.

\subsection{Array Packing}
We first consider a useful sub-problem, array packing.
Array packing is a one-dimensional variant of grid packing, where we have an array of size $m$ and we choose segments of consecutive cells in this array. Again segments can overlap and there is no constraint on the number or size of the segments;  after we fix the segments an adversary puts non-negative numbers on the array cells;  and the score of a subset of the array cells would be the sum of their values. Here we call a solution $(\alpha,\beta)$-approximate if no more than $\alpha$ segments cover each cell and for any interval $[x,y]$ of the array, there exists a segment which completely lies in this interval with a score is at least a $1/\beta$ fraction of the overall score of the interval.

\begin{figure}[ht]

\centering

\tikzset{every picture/.style={line width=0.75pt}} %set default line width to 0.75pt        

\begin{tikzpicture}[x=0.75pt,y=0.75pt,yscale=-1,xscale=1]
%uncomment if require: \path (0,127); %set diagram left start at 0, and has height of 127

%Straight Lines [id:da000861350101278191] 
\draw    (306,38) -- (306,69) ;

%Straight Lines [id:da4515491798042337] 
\draw    (184,39) -- (486,39) ;

%Straight Lines [id:da5711586529634816] 
\draw    (184,69) -- (486,69) ;

%Straight Lines [id:da8492635772371222] 
\draw    (276,38) -- (276,69) ;

%Straight Lines [id:da7439963969228394] 
\draw    (246,38) -- (246,69) ;

%Straight Lines [id:da87771159737836] 
\draw    (216,38) -- (216,69) ;

%Straight Lines [id:da5639157080082706] 
\draw    (186,38) -- (186,69) ;

%Straight Lines [id:da1925307213491647] 
\draw    (366,38) -- (366,69) ;

%Straight Lines [id:da9014920419090529] 
\draw    (336,38) -- (336,69) ;

%Straight Lines [id:da5064291836062897] 
\draw    (426,38) -- (426,69) ;

%Straight Lines [id:da018132113507773218] 
\draw    (396,38) -- (396,69) ;

%Straight Lines [id:da8655584270971388] 
\draw    (486,38) -- (486,69) ;

%Straight Lines [id:da7067561661186135] 
\draw    (456,38) -- (456,69) ;

%Straight Lines [id:da7282534052566849] 
\draw [color={rgb, 255:red, 74; green, 144; blue, 226 }  ,draw opacity=1 ][line width=3.75]    (198,63) -- (259,63) ;

%Straight Lines [id:da6673738300554781] 
\draw [color={rgb, 255:red, 74; green, 144; blue, 226 }  ,draw opacity=1 ][line width=3.75]    (291,63) -- (352,63) ;

%Straight Lines [id:da7171514864092876] 
\draw [color={rgb, 255:red, 74; green, 144; blue, 226 }  ,draw opacity=1 ][line width=3.75]    (378,63) -- (439,63) ;

%Straight Lines [id:da4324133809248172] 
\draw [color={rgb, 255:red, 208; green, 2; blue, 27 }  ,draw opacity=1 ][line width=3.75]    (198,53.5) -- (232,53.5) ;

%Straight Lines [id:da014089775369803048] 
\draw [color={rgb, 255:red, 208; green, 2; blue, 27 }  ,draw opacity=1 ][line width=3.75]    (258,53.5) -- (292,53.5) ;

%Straight Lines [id:da5002605763083692] 
\draw [color={rgb, 255:red, 208; green, 2; blue, 27 }  ,draw opacity=1 ][line width=3.75]    (318,53.5) -- (352,53.5) ;

%Straight Lines [id:da5117820153311623] 
\draw [color={rgb, 255:red, 208; green, 2; blue, 27 }  ,draw opacity=1 ][line width=3.75]    (378,53.5) -- (412,53.5) ;

%Straight Lines [id:da11685276261039368] 
\draw [color={rgb, 255:red, 208; green, 2; blue, 27 }  ,draw opacity=1 ][line width=3.75]    (438,53.5) -- (472,53.5) ;

%Straight Lines [id:da6031246093370757] 
\draw [color={rgb, 255:red, 0; green, 0; blue, 0 }  ,draw opacity=1 ][line width=3.75]    (190,44.5) -- (212,44.5) ;

%Straight Lines [id:da6352855656108114] 
\draw [color={rgb, 255:red, 0; green, 0; blue, 0 }  ,draw opacity=1 ][line width=3.75]    (220,44.5) -- (242,44.5) ;

%Straight Lines [id:da4989459827916627] 
\draw [color={rgb, 255:red, 0; green, 0; blue, 0 }  ,draw opacity=1 ][line width=3.75]    (250,44.5) -- (272,44.5) ;

%Straight Lines [id:da9651012848435458] 
\draw [color={rgb, 255:red, 0; green, 0; blue, 0 }  ,draw opacity=1 ][line width=3.75]    (280,44.5) -- (302,44.5) ;

%Straight Lines [id:da0787010865142661] 
\draw [color={rgb, 255:red, 0; green, 0; blue, 0 }  ,draw opacity=1 ][line width=3.75]    (310,44.5) -- (332,44.5) ;

%Straight Lines [id:da8442033560848261] 
\draw [color={rgb, 255:red, 0; green, 0; blue, 0 }  ,draw opacity=1 ][line width=3.75]    (340,44.5) -- (362,44.5) ;

%Straight Lines [id:da12718137256496487] 
\draw [color={rgb, 255:red, 0; green, 0; blue, 0 }  ,draw opacity=1 ][line width=3.75]    (370,44.5) -- (392,44.5) ;

%Straight Lines [id:da9446134807052542] 
\draw [color={rgb, 255:red, 0; green, 0; blue, 0 }  ,draw opacity=1 ][line width=3.75]    (400,44.5) -- (422,44.5) ;

%Straight Lines [id:da12385101999421555] 
\draw [color={rgb, 255:red, 0; green, 0; blue, 0 }  ,draw opacity=1 ][line width=3.75]    (430,44.5) -- (452,44.5) ;

%Straight Lines [id:da38266859660472297] 
\draw [color={rgb, 255:red, 0; green, 0; blue, 0 }  ,draw opacity=1 ][line width=3.75]    (460,44.5) -- (482,44.5) ;

\end{tikzpicture}

\caption{A $(3,5)$-approximate solution is shown for an array packing problem of size $10$. Blue segments cover three consecutive cells, red segments cover two consecutive cells and each black segment covers a single cell. Each cell is covered by at most $3$ segments. Also, one can prove that no matter how the adversary puts the numbers on the array, for any interval, there is a segment completely in that interval whose score is at least $1/5$ of the score of that interval.} \label{fig:array}
\end{figure}
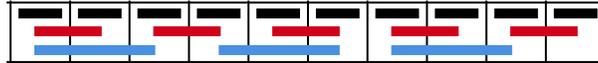

Here we  show that when the array size is $m$, for any $0 < \kappa < 1$, there exists a solution for array packing with approximation guarantee $(O(m^{\kappa}\log m),O(1/\kappa))$.  We then use this in a solution for grid packing.  

We construct two sets of segments in the following way: let $d$ be the largest power of 2 such that $d^2 \leq m^{\kappa}$. In the first set, for any $1 \leq i \leq d$ and any $1 \leq j \leq m-i+1$ we introduce a segment starting from cell $j$ and ending at cell $i+j-1$. 

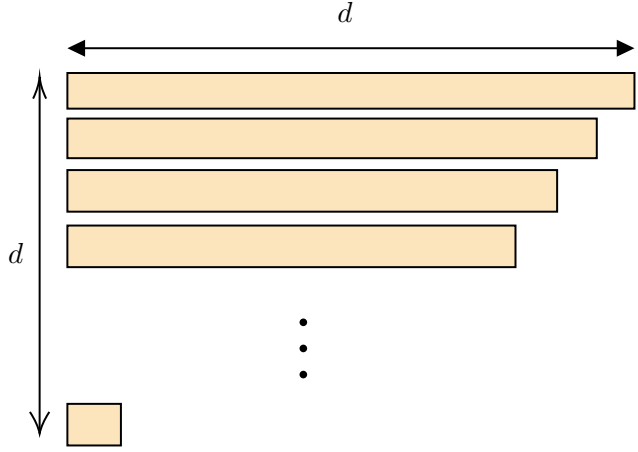
\begin{figure}[ht]

\centering

\tikzset{every picture/.style={line width=0.75pt}} %set default line width to 0.75pt        

\begin{tikzpicture}[x=0.75pt,y=0.75pt,yscale=-1,xscale=1]
%uncomment if require: \path (0,243); %set diagram left start at 0, and has height of 243

%Shape: Rectangle [id:dp7198015778122937] 
\draw  [color={rgb, 255:red, 0; green, 0; blue, 0 }  ,draw opacity=1 ][fill={rgb, 255:red, 245; green, 166; blue, 35 }  ,fill opacity=0.3 ] (199,39) -- (485,39) -- (485,57) -- (199,57) -- cycle ;
%Shape: Rectangle [id:dp326449100428279] 
\draw  [color={rgb, 255:red, 0; green, 0; blue, 0 }  ,draw opacity=1 ][fill={rgb, 255:red, 245; green, 166; blue, 35 }  ,fill opacity=0.3 ] (199,62) -- (466,62) -- (466,82) -- (199,82) -- cycle ;
%Shape: Rectangle [id:dp5227205005194291] 
\draw  [color={rgb, 255:red, 0; green, 0; blue, 0 }  ,draw opacity=1 ][fill={rgb, 255:red, 245; green, 166; blue, 35 }  ,fill opacity=0.3 ] (199,88) -- (446,88) -- (446,109) -- (199,109) -- cycle ;
%Shape: Rectangle [id:dp3199835946686478] 
\draw  [color={rgb, 255:red, 0; green, 0; blue, 0 }  ,draw opacity=1 ][fill={rgb, 255:red, 245; green, 166; blue, 35 }  ,fill opacity=0.3 ] (199,116) -- (425,116) -- (425,137) -- (199,137) -- cycle ;
%Shape: Rectangle [id:dp29431208806742304] 
\draw  [color={rgb, 255:red, 0; green, 0; blue, 0 }  ,draw opacity=1 ][fill={rgb, 255:red, 245; green, 166; blue, 35 }  ,fill opacity=0.3 ] (199,206) -- (226,206) -- (226,227) -- (199,227) -- cycle ;
%Straight Lines [id:da4486140676617101] 
\draw    (201,26.5) -- (483,26.5) ;
\draw [shift={(485,26.5)}, rotate = 180] [fill={rgb, 255:red, 0; green, 0; blue, 0 }  ][line width=0.75]  [draw opacity=0] (8.93,-4.29) -- (0,0) -- (8.93,4.29) -- cycle    ;
\draw [shift={(199,26.5)}, rotate = 0] [fill={rgb, 255:red, 0; green, 0; blue, 0 }  ][line width=0.75]  [draw opacity=0] (8.93,-4.29) -- (0,0) -- (8.93,4.29) -- cycle    ;
%Straight Lines [id:da605367449207038] 
\draw    (185,43) -- (185,219) ;
\draw [shift={(185,221)}, rotate = 270] [color={rgb, 255:red, 0; green, 0; blue, 0 }  ][line width=0.75]    (10.93,-4.9) .. controls (6.95,-2.3) and (3.31,-0.67) .. (0,0) .. controls (3.31,0.67) and (6.95,2.3) .. (10.93,4.9)   ;
\draw [shift={(185,41)}, rotate = 90] [color={rgb, 255:red, 0; green, 0; blue, 0 }  ][line width=0.75]    (10.93,-3.29) .. controls (6.95,-1.4) and (3.31,-0.3) .. (0,0) .. controls (3.31,0.3) and (6.95,1.4) .. (10.93,3.29)   ;

% Text Node
\draw (339,8) node   {$d$};
% Text Node
\draw (318,168) node [scale=2.488]  {$\vdots $};
% Text Node
\draw (173,130) node   {$d$};

\end{tikzpicture}

\caption{In the first set, from each cell we introduce $d$ segments with lengths $1,2,\ldots,d$.} \label{fig:set1}
\end{figure}

In the second set, for any integer $0 \leq i \leq \log m$ and any $j$ divisible by $2^i$ such that $j+d2^i-1 \leq m$, we introduce a segment that spans the interval $[j,j+d2^{i}-1]$.
\begin{figure}[ht]

\centering

\tikzset{every picture/.style={line width=0.75pt}} %set default line width to 0.75pt        

\begin{tikzpicture}[x=0.75pt,y=0.75pt,yscale=-1,xscale=1]
%uncomment if require: \path (0,260); %set diagram left start at 0, and has height of 260

%Shape: Rectangle [id:dp3780865282338346] 
\draw  [color={rgb, 255:red, 0; green, 0; blue, 0 }  ,draw opacity=1 ][fill={rgb, 255:red, 245; green, 166; blue, 35 }  ,fill opacity=0.3 ] (122,72) -- (230,72) -- (230,89) -- (122,89) -- cycle ;
%Shape: Rectangle [id:dp7096124539638322] 
\draw  [color={rgb, 255:red, 0; green, 0; blue, 0 }  ,draw opacity=1 ][fill={rgb, 255:red, 245; green, 166; blue, 35 }  ,fill opacity=0.3 ] (232,72) -- (340,72) -- (340,89) -- (232,89) -- cycle ;
%Shape: Rectangle [id:dp009399095832083537] 
\draw  [color={rgb, 255:red, 0; green, 0; blue, 0 }  ,draw opacity=1 ][fill={rgb, 255:red, 245; green, 166; blue, 35 }  ,fill opacity=0.3 ] (342,72) -- (450,72) -- (450,89) -- (342,89) -- cycle ;
%Shape: Rectangle [id:dp7513724579509333] 
\draw  [color={rgb, 255:red, 0; green, 0; blue, 0 }  ,draw opacity=1 ][fill={rgb, 255:red, 245; green, 166; blue, 35 }  ,fill opacity=0.3 ] (452,72) -- (560,72) -- (560,89) -- (452,89) -- cycle ;
%Straight Lines [id:da6179638507145953] 
\draw  [dash pattern={on 4.5pt off 4.5pt}]  (122,43) -- (122,72) ;

%Straight Lines [id:da9430208529028419] 
\draw  [dash pattern={on 4.5pt off 4.5pt}]  (232,41) -- (232,70) ;

%Straight Lines [id:da7195852022125899] 
\draw  [dash pattern={on 4.5pt off 4.5pt}]  (342,41) -- (342,70) ;

%Straight Lines [id:da33497106168716906] 
\draw  [dash pattern={on 4.5pt off 4.5pt}]  (452,41) -- (452,70) ;

%Straight Lines [id:da8157068284750819] 
\draw  [dash pattern={on 4.5pt off 4.5pt}]  (562,41) -- (562,70) ;

%Shape: Rectangle [id:dp4023652981308239] 
\draw  [color={rgb, 255:red, 0; green, 0; blue, 0 }  ,draw opacity=1 ][fill={rgb, 255:red, 245; green, 166; blue, 35 }  ,fill opacity=0.3 ] (149,95) -- (257,95) -- (257,112) -- (149,112) -- cycle ;
%Shape: Rectangle [id:dp26301835733287526] 
\draw  [color={rgb, 255:red, 0; green, 0; blue, 0 }  ,draw opacity=1 ][fill={rgb, 255:red, 245; green, 166; blue, 35 }  ,fill opacity=0.3 ] (259,95) -- (367,95) -- (367,112) -- (259,112) -- cycle ;
%Shape: Rectangle [id:dp7641849405878227] 
\draw  [color={rgb, 255:red, 0; green, 0; blue, 0 }  ,draw opacity=1 ][fill={rgb, 255:red, 245; green, 166; blue, 35 }  ,fill opacity=0.3 ] (369,95) -- (477,95) -- (477,112) -- (369,112) -- cycle ;
%Shape: Rectangle [id:dp9379290868108303] 
\draw  [color={rgb, 255:red, 0; green, 0; blue, 0 }  ,draw opacity=1 ][fill={rgb, 255:red, 245; green, 166; blue, 35 }  ,fill opacity=0.3 ] (479,95) -- (587,95) -- (587,112) -- (479,112) -- cycle ;
%Straight Lines [id:da9257227625652331] 
\draw    (124,103.5) -- (147,103.5) ;
\draw [shift={(149,103.5)}, rotate = 180] [fill={rgb, 255:red, 0; green, 0; blue, 0 }  ][line width=0.75]  [draw opacity=0] (8.93,-4.29) -- (0,0) -- (8.93,4.29) -- cycle    ;
\draw [shift={(122,103.5)}, rotate = 0] [fill={rgb, 255:red, 0; green, 0; blue, 0 }  ][line width=0.75]  [draw opacity=0] (8.93,-4.29) -- (0,0) -- (8.93,4.29) -- cycle    ;
%Shape: Rectangle [id:dp9242228786576219] 
\draw  [color={rgb, 255:red, 0; green, 0; blue, 0 }  ,draw opacity=1 ][fill={rgb, 255:red, 245; green, 166; blue, 35 }  ,fill opacity=0.3 ] (174,120) -- (282,120) -- (282,137) -- (174,137) -- cycle ;
%Shape: Rectangle [id:dp005030147703118715] 
\draw  [color={rgb, 255:red, 0; green, 0; blue, 0 }  ,draw opacity=1 ][fill={rgb, 255:red, 245; green, 166; blue, 35 }  ,fill opacity=0.3 ] (284,120) -- (392,120) -- (392,137) -- (284,137) -- cycle ;
%Shape: Rectangle [id:dp5859158598734318] 
\draw  [color={rgb, 255:red, 0; green, 0; blue, 0 }  ,draw opacity=1 ][fill={rgb, 255:red, 245; green, 166; blue, 35 }  ,fill opacity=0.3 ] (394,120) -- (502,120) -- (502,137) -- (394,137) -- cycle ;
%Shape: Rectangle [id:dp4457114817002761] 
\draw  [color={rgb, 255:red, 0; green, 0; blue, 0 }  ,draw opacity=1 ][fill={rgb, 255:red, 245; green, 166; blue, 35 }  ,fill opacity=0.3 ] (504,120) -- (612,120) -- (612,137) -- (504,137) -- cycle ;
%Straight Lines [id:da6551105761833025] 
\draw    (148,128.5) -- (171,128.5) ;
\draw [shift={(173,128.5)}, rotate = 180] [fill={rgb, 255:red, 0; green, 0; blue, 0 }  ][line width=0.75]  [draw opacity=0] (8.93,-4.29) -- (0,0) -- (8.93,4.29) -- cycle    ;
\draw [shift={(146,128.5)}, rotate = 0] [fill={rgb, 255:red, 0; green, 0; blue, 0 }  ][line width=0.75]  [draw opacity=0] (8.93,-4.29) -- (0,0) -- (8.93,4.29) -- cycle    ;
%Shape: Rectangle [id:dp49286915722617497] 
\draw  [color={rgb, 255:red, 0; green, 0; blue, 0 }  ,draw opacity=1 ][fill={rgb, 255:red, 245; green, 166; blue, 35 }  ,fill opacity=0.3 ] (207,210) -- (315,210) -- (315,227) -- (207,227) -- cycle ;
%Shape: Rectangle [id:dp48330418500477546] 
\draw  [color={rgb, 255:red, 0; green, 0; blue, 0 }  ,draw opacity=1 ][fill={rgb, 255:red, 245; green, 166; blue, 35 }  ,fill opacity=0.3 ] (317,210) -- (425,210) -- (425,227) -- (317,227) -- cycle ;
%Shape: Rectangle [id:dp18249322992525996] 
\draw  [color={rgb, 255:red, 0; green, 0; blue, 0 }  ,draw opacity=1 ][fill={rgb, 255:red, 245; green, 166; blue, 35 }  ,fill opacity=0.3 ] (427,210) -- (535,210) -- (535,227) -- (427,227) -- cycle ;
%Shape: Rectangle [id:dp893797241969857] 
\draw  [color={rgb, 255:red, 0; green, 0; blue, 0 }  ,draw opacity=1 ][fill={rgb, 255:red, 245; green, 166; blue, 35 }  ,fill opacity=0.3 ] (537,210) -- (645,210) -- (645,227) -- (537,227) -- cycle ;
%Straight Lines [id:da7831472743519785] 
\draw    (208,198.5) -- (231,198.5) ;
\draw [shift={(233,198.5)}, rotate = 180] [fill={rgb, 255:red, 0; green, 0; blue, 0 }  ][line width=0.75]  [draw opacity=0] (8.93,-4.29) -- (0,0) -- (8.93,4.29) -- cycle    ;
\draw [shift={(206,198.5)}, rotate = 0] [fill={rgb, 255:red, 0; green, 0; blue, 0 }  ][line width=0.75]  [draw opacity=0] (8.93,-4.29) -- (0,0) -- (8.93,4.29) -- cycle    ;
%Straight Lines [id:da8186397121702664] 
\draw  [dash pattern={on 4.5pt off 4.5pt}]  (232,141) -- (232,198.5) ;

%Straight Lines [id:da10510335495109491] 
\draw    (25,69) -- (25,245) ;
\draw [shift={(25,247)}, rotate = 270] [color={rgb, 255:red, 0; green, 0; blue, 0 }  ][line width=0.75]    (10.93,-4.9) .. controls (6.95,-2.3) and (3.31,-0.67) .. (0,0) .. controls (3.31,0.67) and (6.95,2.3) .. (10.93,4.9)   ;
\draw [shift={(25,67)}, rotate = 90] [color={rgb, 255:red, 0; green, 0; blue, 0 }  ][line width=0.75]    (10.93,-3.29) .. controls (6.95,-1.4) and (3.31,-0.3) .. (0,0) .. controls (3.31,0.3) and (6.95,1.4) .. (10.93,3.29)   ;
%Shape: Rectangle [id:dp516670719142682] 
\draw  [color={rgb, 255:red, 0; green, 0; blue, 0 }  ,draw opacity=1 ][fill={rgb, 255:red, 245; green, 166; blue, 35 }  ,fill opacity=0.3 ] (39,95) -- (147,95) -- (147,112) -- (39,112) -- cycle ;
%Shape: Rectangle [id:dp2577325239029944] 
\draw  [color={rgb, 255:red, 0; green, 0; blue, 0 }  ,draw opacity=1 ][fill={rgb, 255:red, 245; green, 166; blue, 35 }  ,fill opacity=0.3 ] (63,120) -- (171,120) -- (171,137) -- (63,137) -- cycle ;
%Shape: Rectangle [id:dp5318440496116796] 
\draw  [color={rgb, 255:red, 0; green, 0; blue, 0 }  ,draw opacity=1 ][fill={rgb, 255:red, 245; green, 166; blue, 35 }  ,fill opacity=0.3 ] (97,210) -- (205,210) -- (205,227) -- (97,227) -- cycle ;
%Straight Lines [id:da2827715340750767] 
\draw  [dash pattern={on 0.84pt off 2.51pt}]  (60.07,222.31) -- (135.5,103.5) ;

\draw [shift={(59,224)}, rotate = 302.41] [fill={rgb, 255:red, 0; green, 0; blue, 0 }  ][line width=0.75]  [draw opacity=0] (8.93,-4.29) -- (0,0) -- (8.93,4.29) -- cycle    ;
%Straight Lines [id:da059330727268365635] 
\draw  [dash pattern={on 0.84pt off 2.51pt}]  (60.45,222.62) -- (159.5,128.5) ;

\draw [shift={(59,224)}, rotate = 316.46] [fill={rgb, 255:red, 0; green, 0; blue, 0 }  ][line width=0.75]  [draw opacity=0] (8.93,-4.29) -- (0,0) -- (8.93,4.29) -- cycle    ;
%Straight Lines [id:da19451834729452222] 
\draw  [dash pattern={on 0.84pt off 2.51pt}]  (60.98,223.69) -- (219.5,198.5) ;

\draw [shift={(59,224)}, rotate = 350.97] [fill={rgb, 255:red, 0; green, 0; blue, 0 }  ][line width=0.75]  [draw opacity=0] (8.93,-4.29) -- (0,0) -- (8.93,4.29) -- cycle    ;

% Text Node
\draw (576,75) node   {$\dotsc $};
% Text Node
\draw (124,24) node   {$d2^{i}$};
% Text Node
\draw (231,24) node   {$2d2^{i}$};
% Text Node
\draw (341,24) node   {$3d2^{i}$};
% Text Node
\draw (452,24) node   {$4d2^{i}$};
% Text Node
\draw (562,24) node   {$5d2^{i}$};
% Text Node
\draw (60,235) node   {$2^{i}$};
% Text Node
\draw (422,165) node [scale=2.488]  {$\vdots $};
% Text Node
\draw (607,97) node   {$\dotsc $};
% Text Node
\draw (634,125) node   {$\dotsc $};
% Text Node
\draw (13,156) node   {$d$};

\end{tikzpicture}

\caption{In the second set, for each $i$, there is a segment of length $d2^i$ starting from any cell whose starting point is divisible by $2^i$. The distance between consecutive segments is $2^i$.} \label{fig:set1}
\end{figure}
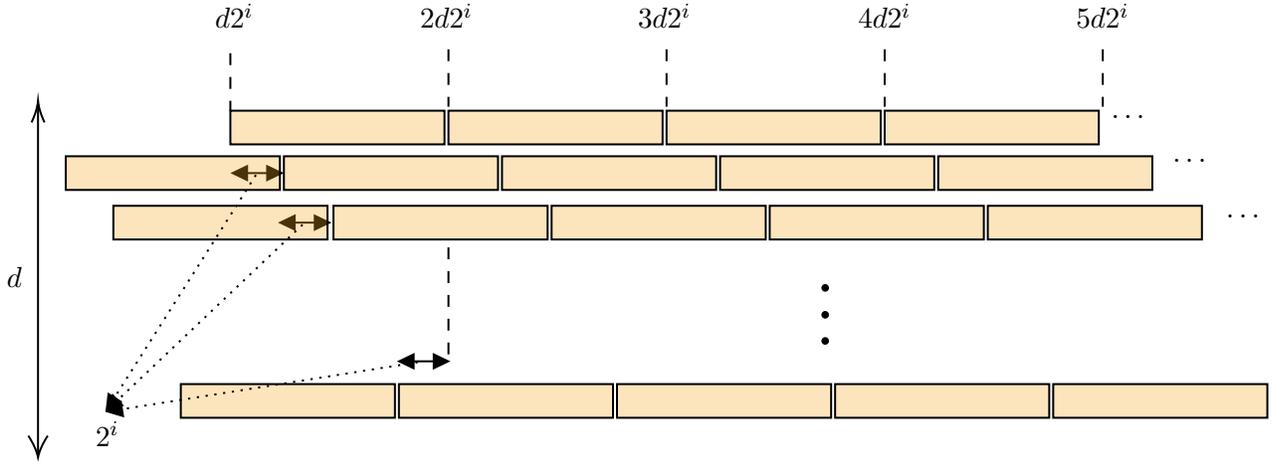

% We prove in the following that such a construction gives us the desired approximation guarantee.
%\Michael{I assume our construction is the UNION of these two sets. This should be explicitly stated;  I put it  in the start of the proof below.}\Saeed{Sounds good.}

\begin{lemma}\label{lemma:array}
	For any $m \geq 1$ and any $0 < \kappa < 1$ there exists a solution for the array packing problem of size $m$ with approximation guarantee $$(O(m^\kappa\log m), O(1/\kappa)).$$
\end{lemma}
\begin{proof}
	The solution is the union of the two sets of segments described above. The key to showing it gives the required approximation is the following: For each interval $[x,y]$ we can cover the entire interval with $O(1/\kappa)$ segments that completely lie in the interval $[x,y]$. Therefore, no matter how the adversary puts the numbers on the cells, the summation of the scores for such segments is at least as much as the score of interval $[x,y]$ and thus one of those segments has an $\Omega(\kappa)$ fraction of the score of interval $[x,y]$. To show this, we cover the cells of the intervals in the following way.
	
	%\Michael{So elsewhere in the section we talk about an $(\alpha,\beta)$ approximation, and here we're using an interval $[x,y]$.  While not the end of the world it's to be avoided.  Any reason we  can't just use $[x,y]$ as the interval?  Try to choose some other letters?}\Saeed{We use $i$ and $j$ for the construction of the segments. But I agree that $\alpha$ and $\beta$ are not the right choices. I will change it to $x$ and $y$ for now.}
	
	If $y - x < d$, then we can cover the entire interval with a single segment (of the first type) and the proof is trivial. Otherwise, we only use the segments of the second type in our covering. Let $q$ be the size of the largest segment that completely lies in the interval $[x,y]$.
	
	We start with an empty set $S_\ell$ and a pointer $p_\ell$ initially equal to $x$. Each time we find a segment that starts in the range $[x,p_\ell]$ and ends in the range $[x,y]$. If many such segments exist, we choose the one with the right-most ending point and in case of a tie, we break the tie arbitrarily. We add the new segment to $S_\ell$ and continue the process by increasing $p_\ell$ to the cell right after the new segment ends. We continue this process so long as $p_\ell - x < q/d$
	
	We repeat the same process but this time starting from $p_r = y$ and proceeding backwards. Initially we set $S_r$ to be an empty set. Every time, we find a segment that ends in range $[p_r,y]$ and starts in range $[x,y]$. Similarly, when multiple options are available, we choose the one whose starting point is the smallest. We add the new segment to $S_r$ and update $p_r$ to be the right-most cell not covered by the new segment. We terminate this process when $y-p_r \geq q/d$.
	
	Obviously segments in $S_\ell \cup S_r$ cover both intervals $[x,p_\ell-1]$ and $[p_r+1,y]$. If $p_r < p_\ell$ then the entire interval $[x,y]$ is covered. Otherwise, we can add two more segments of size $q$ to fill the gap. More precisely, we add the left-most and right-most segments of length $q$ that lie completely in the interval $[x,y]$. The analysis is based on the following property of our construction: the distance between consecutive segments of length $q$ is bounded by $q/d$. Thus, after putting such segments on the interval $[x,y]$, each cell is covered except the ones whose distance to one of the end points is smaller than $q/d$. However, segments of sets $S_\ell$ and $S_r$ cover those cells.
	
	We show that $|S_\ell|,|S_r| \leq O(1/\kappa)$. To this end, we prove that every time we add a new segment to $S_\ell$, the size of the new segment is at least $d$ times larger than the size of the previous segment we added to $S_\ell$. The reason is the following: let $l$ be the length of the segment we add to $S_\ell$ at some point. This means that after adding this segment to $S_\ell$ we have $p_\ell \geq x+l$. $l$ is a power of $2$ since we only use the segments of the second type. Moreover, one cell in this range is divisible by $l$ which means that it is the starting point of another segment of length $dl$. Thus the next segment that we add to $S_\ell$ has a length at least $dl$ which is $d$ times larger than the current one. Therefore, the size of $|S_\ell|$ is bounded by $\log_d m = O(1/\kappa)$. The same also holds for $S_r$.
	
	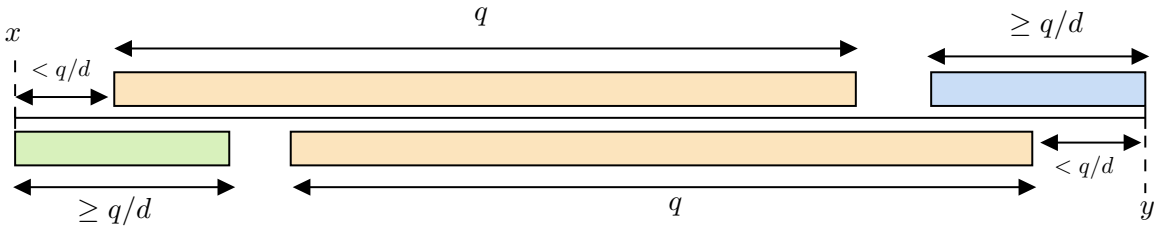
\begin{figure}[ht]

\centering

\tikzset{every picture/.style={line width=0.75pt}} %set default line width to 0.75pt        

\begin{tikzpicture}[x=0.75pt,y=0.75pt,yscale=-1,xscale=1]
%uncomment if require: \path (0,123); %set diagram left start at 0, and has height of 123

%Straight Lines [id:da6826662328602326] 
\draw    (40,60) -- (610,60) ;
\draw [shift={(610,60)}, rotate = 180] [color={rgb, 255:red, 0; green, 0; blue, 0 }  ][line width=0.75]    (0,5.59) -- (0,-5.59)   ;
\draw [shift={(40,60)}, rotate = 180] [color={rgb, 255:red, 0; green, 0; blue, 0 }  ][line width=0.75]    (0,5.59) -- (0,-5.59)   ;
%Shape: Rectangle [id:dp23545492012890912] 
\draw  [color={rgb, 255:red, 0; green, 0; blue, 0 }  ,draw opacity=1 ][fill={rgb, 255:red, 245; green, 166; blue, 35 }  ,fill opacity=0.3 ] (90,37) -- (464,37) -- (464,54) -- (90,54) -- cycle ;
%Shape: Rectangle [id:dp14149673151289388] 
\draw  [color={rgb, 255:red, 0; green, 0; blue, 0 }  ,draw opacity=1 ][fill={rgb, 255:red, 245; green, 166; blue, 35 }  ,fill opacity=0.3 ] (179,67) -- (553,67) -- (553,84) -- (179,84) -- cycle ;
%Straight Lines [id:da35804818707319486] 
\draw    (182,95) -- (553,95) ;
\draw [shift={(555,95)}, rotate = 180] [fill={rgb, 255:red, 0; green, 0; blue, 0 }  ][line width=0.75]  [draw opacity=0] (8.93,-4.29) -- (0,0) -- (8.93,4.29) -- cycle    ;
\draw [shift={(180,95)}, rotate = 0] [fill={rgb, 255:red, 0; green, 0; blue, 0 }  ][line width=0.75]  [draw opacity=0] (8.93,-4.29) -- (0,0) -- (8.93,4.29) -- cycle    ;
%Straight Lines [id:da6497706704680086] 
\draw    (92,25) -- (463,25) ;
\draw [shift={(465,25)}, rotate = 180] [fill={rgb, 255:red, 0; green, 0; blue, 0 }  ][line width=0.75]  [draw opacity=0] (8.93,-4.29) -- (0,0) -- (8.93,4.29) -- cycle    ;
\draw [shift={(90,25)}, rotate = 0] [fill={rgb, 255:red, 0; green, 0; blue, 0 }  ][line width=0.75]  [draw opacity=0] (8.93,-4.29) -- (0,0) -- (8.93,4.29) -- cycle    ;
%Shape: Rectangle [id:dp4417477082264687] 
\draw  [color={rgb, 255:red, 0; green, 0; blue, 0 }  ,draw opacity=1 ][fill={rgb, 255:red, 74; green, 144; blue, 226 }  ,fill opacity=0.3 ] (502,37) -- (610,37) -- (610,54) -- (502,54) -- cycle ;
%Shape: Rectangle [id:dp49180493926898716] 
\draw  [color={rgb, 255:red, 0; green, 0; blue, 0 }  ,draw opacity=1 ][fill={rgb, 255:red, 126; green, 211; blue, 33 }  ,fill opacity=0.3 ] (40,67) -- (148,67) -- (148,84) -- (40,84) -- cycle ;
%Straight Lines [id:da4592175028028467] 
\draw  [dash pattern={on 4.5pt off 4.5pt}]  (40,31) -- (40,60) ;

%Straight Lines [id:da37907827815881534] 
\draw  [dash pattern={on 4.5pt off 4.5pt}]  (610,69) -- (610,98) ;

%Straight Lines [id:da42140038104766475] 
\draw    (41,95) -- (149,95) ;
\draw [shift={(151,95)}, rotate = 180] [fill={rgb, 255:red, 0; green, 0; blue, 0 }  ][line width=0.75]  [draw opacity=0] (8.93,-4.29) -- (0,0) -- (8.93,4.29) -- cycle    ;
\draw [shift={(39,95)}, rotate = 0] [fill={rgb, 255:red, 0; green, 0; blue, 0 }  ][line width=0.75]  [draw opacity=0] (8.93,-4.29) -- (0,0) -- (8.93,4.29) -- cycle    ;
%Straight Lines [id:da46499990188669416] 
\draw    (502,29) -- (610,29) ;
\draw [shift={(612,29)}, rotate = 180] [fill={rgb, 255:red, 0; green, 0; blue, 0 }  ][line width=0.75]  [draw opacity=0] (8.93,-4.29) -- (0,0) -- (8.93,4.29) -- cycle    ;
\draw [shift={(500,29)}, rotate = 0] [fill={rgb, 255:red, 0; green, 0; blue, 0 }  ][line width=0.75]  [draw opacity=0] (8.93,-4.29) -- (0,0) -- (8.93,4.29) -- cycle    ;
%Straight Lines [id:da15800952593965167] 
\draw    (42,49.5) -- (85,49.5) ;
\draw [shift={(87,49.5)}, rotate = 180] [fill={rgb, 255:red, 0; green, 0; blue, 0 }  ][line width=0.75]  [draw opacity=0] (8.93,-4.29) -- (0,0) -- (8.93,4.29) -- cycle    ;
\draw [shift={(40,49.5)}, rotate = 0] [fill={rgb, 255:red, 0; green, 0; blue, 0 }  ][line width=0.75]  [draw opacity=0] (8.93,-4.29) -- (0,0) -- (8.93,4.29) -- cycle    ;
%Straight Lines [id:da4271946711517298] 
\draw    (559,72.5) -- (606,72.5) ;
\draw [shift={(608,72.5)}, rotate = 180] [fill={rgb, 255:red, 0; green, 0; blue, 0 }  ][line width=0.75]  [draw opacity=0] (8.93,-4.29) -- (0,0) -- (8.93,4.29) -- cycle    ;
\draw [shift={(557,72.5)}, rotate = 0] [fill={rgb, 255:red, 0; green, 0; blue, 0 }  ][line width=0.75]  [draw opacity=0] (8.93,-4.29) -- (0,0) -- (8.93,4.29) -- cycle    ;

% Text Node
\draw (39,18) node   {$x $};
% Text Node
\draw (611,107) node   {$y $};
% Text Node
\draw (373,105) node   {$q$};
% Text Node
\draw (275,9) node   {$q$};
% Text Node
\draw (90,107) node   {$\geq q/d$};
% Text Node
\draw (560,13) node   {$\geq q/d$};
% Text Node
\draw (63,36) node [scale=0.8]  {$< q/d$};
% Text Node
\draw (579,85) node [scale=0.8]  {$< q/d$};

\end{tikzpicture}

\caption{Green parts are covered with segments in $S_\ell$ and blue parts are covered by segments in $S_r$. Two segments of length $q$ cover the rest of the elements.} \label{fig:proof}
\end{figure}
	
	All that remains is to show that each cell of the array is covered by at most $O(m^{\kappa} \log m)$ different segments. In the first set, the length of each segment is bounded by $d$ and therefore there are at most $d^2$ different combinations for the starting and ending cells of the intervals that cover a particular cell. Since $d^2 \leq m^\kappa$, this guarantee is met by the first set of segments. For the second set, notice that there are at most $O(\log m)$ distinct segment sizes. Moreover, each cell is covered by at most $d$ segments of a particular size. Thus, each cell is covered by at most $O(d \log m)$ different segments of the second set which is bounded by $O(m^\kappa \log m)$.
	
\end{proof}
\subsection{An $(O(m^\kappa \log m),O(1/\kappa))$-approximate Solution for Grid Packing}
We provide a reduction from grid packing to array packing. The intuition behind the reduction is given below: Let us fix a path from the bottom-left to the top-right of the array. We can divide the cells of the path into some disjoint vertical and horizontal intervals as show in Figure \ref{fig:vh}. In this decomposition, all the row-intervals are non-conflicting and all the column-intervals are also non-conflicting (row-intervals and column-intervals may be conflicting).

\input{figs/vh}

For any path and any combination of numbers on the cells of the path, either the sum of the numbers on the row-intervals or the sum of numbers on the column intervals is at least a $1/2$ fraction of the total sum of the numbers on the path. Based on this, we can reduce grid packing to array packing in the following way.
% : We treat each row of the grid as an array and construct a solution of array packing for it.
% \Michael{Make clear we construct an array packing solution using the lemma of the previous section?}
% Next, we treat each column of the grid as an array and construct a solution of array packing for that column. We show in Theorem \ref{theorem:grid-cover} that this provides an $(O(m^\kappa \log m), O(1/\kappa))$-approximate solution.

\begin{theorem}\label{theorem:grid-cover}
	For any $0 < \kappa < 1$, the grid packing problem on an $m \times m$ grid admits an $(O(m^\kappa \log m), O(1/\kappa))$-approximate solution.
\end{theorem}
\begin{proof}
	We treat every row and every column of the grid as an array of length $m$ and construct an $(O(m^\kappa \log m), O(1/\kappa))$-approximate solution of the array packing problem for that row or column (Lemma \ref{lemma:array}). 
	%\Michael{Make clear we construct an array packing solution using the lemma of the previous section?}\Saeed{Sounds good.}
	With this construction, every cell is covered by at most $O(m^\kappa \log m)$ segments, because every cell is covered by at most $O(m^\kappa \log m)$ horizontal segments and every cell is covered by at most $O(m^\kappa \log m)$ vertical segments, so the total number of segments covering each cell is bounded by $O(m^\kappa \log m)$.
	
	After the adversary puts the numbers on the cells of the grid, the score of the grid equals the largest score of a path of the length $2m-1$ that starts from the bottom left corner and moves to the top right corner. As mentioned previously, we can divide  the cells of such a path into non-conflicting row intervals and non-conflicting column intervals. The score of either the row intervals or column  intervals is at least half the score  of this path. Let it be the  row intervals without loss of generality. By Lemma~\ref{lemma:array}, 
	%\Michael{put in specific reference to array result}\Saeed{Done}
	there is one segment in each row that approximates the score of the corresponding interval within a factor $ O(1/\kappa)$ and fits completely within that interval. Since all the row-intervals are non-conflicting, these segments are non-conflicting, and the score of those segments is at least an $O(1/\kappa)$ fraction of the score of the row intervals, which is at least $1/2$ of the score of the grid. This completes the proof.
\end{proof}

\newpage
\section{Constant Factor Approximation for \textsf{LIS} with Update Time $\tilde O(n^\epsilon)$}\label{sec:constant}
In this section, we show that for any constant $\epsilon > 0$, there is a dynamic algorithm for \textsf{LIS} that loses only a constant factor in the approximation (where the constant depends on $\epsilon$) and guarantees an update time bounded by $\tilde O(n^\epsilon)$. Our algorithm is based on the grid packing technique explained in Section \ref{sec:grid}.

We begin by explaining a simple algorithm for dynamic \textsf{LIS} where the approximation factor is constant and the update time is close to $\tilde O(n^{2/3})$. This result is weaker than what we give in Section \ref{sec:amortized} both in terms of update time and approximation factor, but we show this new technique can be extended to obtain the result for any $\epsilon > 0$ above.

We remind the reader that an exact algorithm with worst-case update time $\tilde O(n)$ is trivial; we compute the \textsf{LIS} from the scratch after each operation arrives. We call this algorithm $\mathcal{A'}_0$\footnote{It will be clear to the reader later in the section why we use such an unconventional notation for this algorithm.}. From $\mathcal{A'}_0$, we make a block-based algorithm $\mathcal{A}_1$ and then turn it into an algorithm $\mathcal{A'}_1$ with worst-case update time $\tilde O(n^{2/3+\kappa})$ for a small enough $\kappa > 0$.

In our block-based algorithm, we begin by an array $a$ of length $n$. Map the elements of $a$ onto the 2D plane by creating a point $(i,a_i)$ for every element of the array. Recall that we assume the elements are distinct but there is no bound on their values.
%\Michael{Are we assuming the element are distinct?  Assuming the are specifically $[1,n]$ for convenience?  Make clear up front.}\Saeed{Done.}
Define $m = n^{1/3}$ and construct a grid in the following way: draw $m-1$ horizontal lines that divide the plane into parts of as equal size as possible with respect to the number of points included in each part. If $n$ is not divisible by $m$, some parts may have one more point than other parts. Similarly, we draw $m-1$ vertical lines that separate the plane into $m$ parts each having either $\lfloor n/m \rfloor$ or $\lceil n/m \rceil$ points.  This gives an $m$ by $m$ grid, with each grid cell corresponding to a ``rectangle'' in the 2D plane.
\input{figs/example1}

Our block-based algorithm works in the following way: fix a $0 < \kappa < 1$ and construct a grid packing solution with approximation $(O(m^\kappa \log m), O(1/\kappa))$ for the $m$ by $m$ grid. Each element of the array lies in exactly one cell of the grid and corresponds to all segments that cover that cell. Next for each segment in the solution of grid packing, we construct a separate instance of the \textsf{LIS} problem that maintains a solution for the \textsf{LIS} of the corresponding elements. %\Michael{This is the part that is vague and confusing--  you should describe this as the score that will be associated with each grid cell, and explain what it is --- "LIS that is only concenred with the points lying in the corresponding cells" is pretty ambiguous.  "Concerned with" is not a mathematial term.}\Saeed{We have not talked about the scores at this point. This comes later in the proof of the correctness. I will add more details later when we define the score of the cells and the ones for segments.}% Keep in mind that the solution of \textsf{LIS} for a segment is an upper bound on the score of the segment as we described earlier.

Each time an operation arrives, we update the solution for the corresponding segments. Let us be more specific about this. Initially, $m-1$ vertical lines divide the array into chunks of size roughly $n/m$. As operations arrive, the elements are shifted to the left or to the right (their indices are updated). Each vertical line can be thought of as a separator between two consecutive elements that is also shifted to the left or to the right when elements are added or removed. Thus, although the vertical lines move, each element which is inserted or deleted lies between two thresholds and corresponds to a unique column of the grid. The corresponding row is uniquely determined by the horizontal lines (those lines remain unchanged). Thus, every element insertion or deletion affects only one cell of the grid which is covered by a bounded number of segments. 
%\Michael{It is ambiguous what the "corresponding segments" are, particularly since you haven't explained what an  update -- e.g. insertion -- does to the grid.  As we've discussed, inserting a point potentially does a lot of things to the grid -- chaning line boundaries, etc. -- and the reader doesn't know what you mean at this point.  Spell things out precisely, then get to the result.}\Saeed{Added a sentence earlier to define the corresponding segments and a few senences to resolve the issue.}
When the size of the \textsf{LIS} is desired, we run a dynamic program and find a subset of non-conflicting set of segments whose total solution size is maximized. We show that this gives us approximation factor $O(1/\kappa)$ for the \textsf{LIS} of the array. Our algorithm is responsible for up to $g(n) = n^{2/3}$ many operations.

\input{figs/example2}

\begin{lemma}\label{lemma:chert}
	Let $0 < \kappa < 1$ be an arbitrarily small constant used for the solution of grid packing. $\mathcal{A}_1$ is a block-based algorithm for dynamic \textsf{LIS} with approximation factor $O(1/\kappa)$ whose preprocessing time is $\tilde O(n^{1+\kappa})$ and whose update time is $\tilde O(n^{2/3+\kappa})$. Moreover, $\mathcal{A}_1$ runs for $n^{2/3}$ many steps.
\end{lemma}
\begin{proof}
	We first prove that the preprocessing time of $\mathcal{A}_1$ is $\tilde O(n^{1+\kappa})$. It takes time $\tilde O(n)$ to sort the numbers and draw the grid lines. Also, the runtime for constructing the solution for grid packing is $\tilde O(m^{2+\kappa})$, which is smaller than $\tilde O(n)$. Every cell of the grid is covered by at most $\tilde O(m^{\kappa})$ different segments. We construct a separate \textsf{LIS} instance for every segment. Each grid cell appears in at most $\tilde O(m^{\kappa})$ segments. Therefore, every element of the array in included in in at most $\tilde O(m^\kappa)$ \textsf{LIS} instances. Thus, the total size of all the instances combined is bounded by $\tilde O(nm^{\kappa})$ 
	%\Michael{Explain  this? Where does the $n$ come from? }\Saeed{Added more details.}
	and thus finding an \textsf{LIS} for each segment takes time $\tilde O(nm^{\kappa})$ in total. Since $m = n^{1/3}$ the preprocessing time is bounded by $\tilde O(n^{1+\kappa/3}) = \tilde O(n^{1+\kappa})$.

	The total number of points in every column, or every row of the grid is bounded by $\lceil n/m \rceil = O(n^{2/3})$. Thus, the size of the \textsf{LIS} instance for each segment is also bounded by $O(n^{2/3})$. Since we run the algorithm for at most $O(n^{2/3})$ many steps, the sizes of the \textsf{LIS} instances remain bounded by $O(n^{2/3})$ even if we add more numbers to them in the next $O(n^{2/3})$ many operations.  %\Michael{Again, up to this point, there has been no explanation of how the grid changes/updates when an insert operation occurs for the LIS problem, this is unclear.}\Saeed{Some new details come earlier.}
	
	When a new operation arrives, this only affects one cell of the grid which we can find using its position and its value. We update all the corresponding segments that cover that cell. % \Michael{How?}\Saeed{Added more details.}
	Their count is bounded by $\tilde O(m^{\kappa}) \leq \tilde O(n^{\kappa})$. Moreover, each one we can update in time $\tilde O(n^{2/3})$ since the problem size for each segment is bounded by $O(n^{2/3})$.  Every time the size of the \textsf{LIS} is desired, we run a DP in time $\tilde O(m^{2+\kappa})$ and find a solution that can be constructed using non-conflicting segments. The size of the \textsf{LIS} for each segment is available in time $O(1)$. Thus, the runtime for approximating the \textsf{LIS} is bounded by $\tilde O(m^{2+\kappa}) \leq \tilde O(n^{2/3+\kappa})$. Therefore, the update time $h(n)$ is bounded by $\tilde O(n^{2/3+\kappa})$ for every operation.
	
	For the dynamic program, we define an $m \times m$ table $D$ such that $D[i][j]$ denotes the solution for any subset of non-conflicting segments in the first $i$ rows and the first $j$ columns of the grid. Obviously $D[i][j] \geq D[i-1][j],D[i][j-1]$. Thus, when computing the value for $D[i][j]$, we start by assigning $\max\{D[i-1][j],D[i][j-1]\}$ to it. Next, for any segment that ends at cell $(i,j)$, we update $D[i][j]$ as
	 $$D[i][j] = \max\{D[i][j],D[i'-1][j'-1]+w\}$$
	 where $(i',j')$ are coordinates of the bottom-left corner of the segment and $w$ is the length of its \textsf{LIS}. The total runtime of this algorithm is asymptotically equal to the number of segments on the grid.

	%\Michael{I  keep seeing operations that are $\tilde O(m^{2/3+\kappa})$ that you bound by $\tilde O(n^{2/3+\kappa})$, which is  much  biggers.  Is there anything that  is  really $\tilde O(n^{2/3+\kappa})$ or  could we use  $\tilde O(n^{2/3})$?}\Saeed{Good point. The actual runtime is $O(n^{2/3+\kappa/3}$ but I thought it may be cleaner to just say $O(n^{2/3+\kappa}$.}
	
	Finally, we show that this gives us an $O(1/\kappa)$-approximate solution for the \textsf{LIS} of the entire array. For any point in time, fix an arbitrary solution of the longest increasing subsequence for the array at that time. Assume that the numbers the adversary puts on the cells of the grid are the contributions of the cells to the fixed longest common subsequence. This way, the score of the grid is exactly equal to the size of the longest increasing subsequence.
	
	In addition to the above, the score of every segment is a lower bound on the \textsf{LIS} of the elements for that segment. Thus, by the guarantee of the grid packing solution, the solution we obtain by appending the solutions of non-conflicting segments is definitely an $O(1/\kappa)$ approximate solution for the score of the grid which is the size of the solution. Finally, the validity of our solution follows from the fact that since all the segments are non-conflicting, then we can combine their partial solutions and that gives us a valid increasing subsequence.

	Keep in mind that when we add elements to or remove elements from the array, the indices of the numbers change. Thus, we need to also update the coordinates of the vertical lines. This can be done similarly to the way we update the indices of the array elements. Therefore, this does not add computational difficulty as all we need to know for each new operation is which column of the grid this operation applies to and which row of the grid is affected by the new operation. Also, for edge cases (when the left most or right most element of a column is removed or added), we have a choice of which column we add the new points to. In any case, the solution we find preserves the approximation.
	%\Michael{This sort of explanation has to  come much earlier!}\Saeed{done.}
	
\end{proof}

One thing to note about the algorithm of Lemma \ref{lemma:chert} is that by setting $\kappa$ arbitrarily close to $0$, we can obtain a update time close to $\tilde O(n^{2/3})$. By Lemma \ref{lemma:reduction}, algorithm $\mathcal{A}_1$ can be turned into an algorithm $\mathcal{A'}_1$ with the same approximation factor but worst-case update time $\tilde O(n^{2/3+\kappa})$. Although $\mathcal{A'}_1$ has an approximation factor of $O(1/\kappa)$, its update time is $\tilde O(n^{1/3-\kappa})$ times smaller than that of $\mathcal{A'}_0$. Thus, our approach to improve the overall update time is to replace $\mathcal{A'}_0$ by $\mathcal{A'}_1$ to obtain a faster (but worse in terms of approximation factor) algorithm.

\begin{lemma}[as a corollary of Lemmas \ref{lemma:chert} and \ref{lemma:reduction}]
	For any $0 < \kappa < 1$, there exists a dynamic algorithm for \textsf{LIS} with worst case update time $\tilde O(n^{2/3+\kappa})$ and approximation factor $\tilde O(1/\kappa)$.
\end{lemma}

It is not hard to see that one can recurse on the above idea to improve the update time. This comes however, at the expense of a larger approximation factor. We prove in Theorem \ref{theorem:constant} that, similar to what we did for Lemma \ref{lemma:chert}, one can devise an algorithm with worst-case update time $\tilde O(n^{\epsilon})$ with approximation factor $O((1/\epsilon)^{O(1/\epsilon)})$ for any $\epsilon > 0$.

\begin{theorem}\label{theorem:constant}
	For any constant $\epsilon > 0$, there exists an algorithm for dynamic \textsf{LIS} whose worst-case update time is $\tilde O(n^{\epsilon})$ and whose approximation factor is $O((1/\epsilon)^{O(1/\epsilon)})$.
\end{theorem}
\begin{proof}
	\begin{table}[!htbp]
\centering
\begin{tabular}{|c|c|c|}
	\hline
	Dynamic & worst-case  & approximation \\
	algorithm &  update time &  factor\\
	\hline 
	$\mathcal{A'}_0$ & $\tilde O(n)$ & 1\\
	\hline 
	$\mathcal{A'}_1$ & $\tilde O(n^{2/3+\kappa})$ & $O(1/\kappa)$\\
	\hline
	$\mathcal{A'}_2$ & $\tilde O(n^{1/2+\kappa})$ & $O(1/\kappa^2)$\\
	\hline
	$\mathcal{A'}_3$ & $\tilde O(n^{2/5+\kappa})$ & $O(1/\kappa^3)$\\
	\hline
	\vdots  & \vdots & \vdots \\
	\hline
	$\mathcal{A'}_k$ & $\tilde O(n^{\frac{2}{k+2}+\kappa})$ & $O(1/\kappa)^k$\\
	\hline
\end{tabular}
\caption{Guarantees of the dynamic solutions are shown in this figure.}
\end{table}

\begin{table}[!htbp]
\centering
\begin{tabular}{|c|c|c|c|c|c|}
	\hline
	block-based & $m$ & $f(n)$ & $g(n)$ & $h(n)$ & approximation \\
	 algorithm &  & &  &  & factor \\
%	\hline 
%	- & - & -  &- & - & -\\
	\hline
	$\mathcal{A}_1$ & $n^{1/3}$ & $\tilde O(n^{1+\kappa})$ & $n^{2/3}$ & $\tilde O(n^{2/3+\kappa})$ & $O(1/\kappa)$\\
	\hline
	$\mathcal{A}_2$ & $n^{1/4}$ & $\tilde O(n^{1+\kappa})$ & $n^{3/4}$ & $\tilde O(n^{1/2+\kappa})$ & $O(1/\kappa^2)$\\
	\hline
	$\mathcal{A}_3$ & $n^{1/5}$ & $\tilde O(n^{1+\kappa})$ & $n^{4/5}$ & $\tilde O(n^{2/5+\kappa})$ & $O(1/\kappa^3)$\\
	\hline
	\vdots  & \vdots & \vdots & \vdots & \vdots & \vdots \\
	\hline
	$\mathcal{A}_k$ & $n^{\frac{1}{k+2}}$ &$\tilde O(n^{1+\kappa})$ & $n^{\frac{k+1}{k+2}}$ & $\tilde O(n^{\frac{2}{k+2}+\kappa})$ & $O(1/\kappa)^k$\\
	\hline
\end{tabular}
\caption{Guarantees of the block-based solutions are shown in this figure. We assume that the length of the initial array $a$ for the block-based algorithm is $n$.}
\end{table}

	The proof is by induction. Let us fix a constant $0 < \kappa < 1$ that is used for all recursions. For any $k \geq 1$, our aim is to design a dynamic algorithm for \textsf{LIS} with approximation factor $O((1/\kappa)^k)$ and worst-case update time $\tilde O(n^{\frac{2}{k+2}+\kappa})$. We call such an algorithm $\mathcal{A}'_k$. The base case is for $k = 0$ for which we already know a solution $\mathcal{A}'_0$. We also strengthen our hypothesis: If instead of starting from an empty array, we start with an array of length $n$, our algorithm needs a preprocessing time of $\tilde O(n^{1+\kappa})$.
	
	Let us assume that for $k \geq 1$, $\mathcal{A}'_{k-1}$ with desirable guarantees is available and the goal is to design $\mathcal{A}'_{k}$. To this end, we first make a block-based algorithm $\mathcal{A}_k$ with the following properties: For an initial array $a$ of length $n$, we set $f(n) = n^{1+\kappa}$, $g(n) = n^{\frac{k+1}{k+2}}$, and $h(n) = \tilde O(n^{\frac{k+1}{k+2}+\kappa})$. Similar to what we did in Lemma \ref{lemma:chert}, we map every element of the initial array $a$ onto the 2D plane by putting a point $(i,a_i)$ for every element. We set $m = n^{\frac{1}{k+2}}$ and divide the plane into an $m \times m$ grid, such that in each row and in each column there are at most $\lceil n/m \rceil$ points. Moreover, we construct a grid packing solution for the $m \times m$ grid with approximation guarantee $(\tilde O(m^{\kappa}),O(1/\kappa))$. Next, for each segment we initiate an \textsf{LIS} instance that will be solved with algorithm $\mathcal{A}'_{k-1}$.
	
	The preprocessing time consists of two parts: (i) constructing the solution to grid packing in time $\tilde O(m^{2+\kappa})$ and (ii) constructing an initial solution for each segment. Since the problem size for each segment is bounded by $n/m$ and every element appears in at most $\tilde O(m^\kappa)$ segments and the initialization time for $\mathcal{A}'_{k-1}$ is $\tilde O(n^{1+\kappa})$ this can be bounded by $\tilde O(m^{\kappa} n (n/m)^\kappa) = \tilde O(n^{1+\kappa})$. Thus, the total preprocessing time is bounded by
	 $$\tilde O(m^{2+\kappa} + n^{1+\kappa}) = \tilde O(n^{1+\kappa}).$$
	 
	 By the construction of the grid, the size of the problem for each segment is bounded by $O(n/m) = O(n^{\frac{k+1}{k+2}})$. Also, since $g(n) = O(n^{\frac{k+1}{k+2}})$, the size of the problem instance corresponding to each segment remains in this range throughout the $g(n)$ many steps. Each time an operation arrives, we update the solution for $\tilde O(m^{\kappa})$ many segments each in time $\tilde O((n/m)^{\frac{2}{k+1}+\kappa})$ (recall that we use $\mathcal{A}'_{k-1}$ for the solution of the segments). Thus, the update time is bounded by 
	 \begin{align*}
	 \tilde O(m^\kappa (n/m)^{\frac{2}{k+1}+\kappa}) &= \tilde O(n^\kappa (n/m)^{\frac{2}{k+1}})\\
	 & = \tilde O(n^\kappa (n^{\frac{k+1}{k+2}})^{\frac{2}{k+1}})\\
	 & = \tilde O(n^\kappa n^{\frac{2}{k+2}})\\
	 & = \tilde O( n^{\frac{2}{k+2}+\kappa}).
	 \end{align*}
	 Moreover, in order to find an approximate solution for \textsf{LIS} using non-conflicting segments, we need to run a DP in time $\tilde O(m^{2+\kappa}) = \tilde O(n^{\frac{2+\kappa}{k+2}}) = \tilde O( n^{\frac{2}{k+2}+\kappa})$. Thus, the overall worst-case update time is equal to $h(n) = \tilde O( n^{\frac{2}{k+2}+\kappa})$.

	 Finally, the approximation factor increases by a multiplicative factor of $O(1/\kappa)$ in each level of recursion. Thus, the approximation factor of $\mathcal{A}_k$ is bounded by $O((1/\kappa)^k)$.  Using Lemma \ref{lemma:reduction}, we can the turn block-based algorithm $\mathcal{A}_k$ into an algorithm for \textsf{LIS} with worst-case update time $\tilde O( n^{\frac{2}{k+2}+\kappa})$. Finally, since we are constructing $\mathcal{A}'_k$ from a block-based algorithm with preprocessing time $\tilde O(n^{1+\kappa})$,  if we start from an array $a$ of length $n$, we need a preprocessing time of $\tilde O(n^{1+\kappa})$ which is another condition of the hypothesis.
	 
	 Now, for a fixed $\epsilon > 0$, we set $\kappa = \epsilon/2$ and  $k = \lceil 4/\epsilon \rceil$. Thus, the worst-case update time of the algorithm would be bounded by 
	 \begin{align*}
	 \tilde O(n^{\frac{2}{k+2}+\kappa}) &\leq \tilde O(n^{\frac{2}{k}+\kappa})\\
	 & = \tilde O(n^{\frac{2}{k}+\epsilon/2})\\
	 &\leq \tilde O(n^{\epsilon/2+\epsilon/2})\\
	 & = \tilde O(n^{\epsilon})
	 \end{align*}
	 which is desired. Also, the approximation factor would be  bounded by $O((1/\epsilon)^{O(1/\epsilon)})$.
\end{proof}

\newpage
\section{$1+\epsilon$ Approximation for \textsf{DTM}}\label{sec:dtm}
In this section, we give a dynamic algorithm for \textsf{DTM} with approximation factor $1+\epsilon$ and update time $O(\log^2 n)$. If access to the elements of the array is provided in time $O(1)$, the update time improves to $O(\log n)$. However, since we use a balanced binary tree for the elements, there is an additional $O(\log n)$ overhead for the update time. We explain the algorithm in four steps. In the first step, we present a simple algorithm that obtains an approximation factor of $2$ with the same update time. Next, in Step 2, we show how a $2$-approximate solution can be used to obtain an exact solution in time $O(k^2 \log n)$ when the solution size is bounded by $k$. In the third step, we improve the runtime of  the same algorithm to $O(k \log n)$. Finally, we show in the last step that such an algorithm along with the $2$-approximate solution in time $O(\log^2 n)$ gives us a $1+\epsilon$ approximate solution with update time $O(\log^2 n)$.

%For simplicity, let us first consider the more relaxed setting in which an amortized update time of $O(\log n)$ is desired. However, we show later in the section that the framework explained in Section \ref{sec:amortized} can be used to make the update time $O(\log n)$ in the worst case.

\subsection{A $2$-approximate Solution for \textsf{DTM}}\label{sec:2-approx}
Similar to previous work~\cite{DBLP:conf/soda/GopalanJKK07}, we call a pair $(a_i,a_j)$ an \textit{inversion}, if $i < j$ but $a_i > a_j$. The heart of the analysis is that a maximal set of disjoint inversion pairs is a 2-approximate solution for \textsf{DTM}. We first formally give a proof to this claim and next show how such a maximal set can be maintained with $O(\log^2 n)$ update time.

\begin{observation}\label{observation:obvious}
	Let $a = \langle a_1, a_2, \ldots, a_n\rangle$ be an array of length $n$ and $S = \{(a_{\alpha_1},a_{\beta_1}),(a_{\alpha_2},a_{\beta_2}),\ldots\}$ be a maximal set of disjoint inversions of $a$. Then we have
	$$|S| \leq \mathsf{DTM}(a) \leq 2|S|.$$
\end{observation}
\begin{proof}
	Since every pair $(a_{\alpha_i},a_{\beta_i})$ is an inversion, then any solution for \textsf{DTM} should remove one element from each pair which implies $\mathsf{DTM}(a) \geq |S|$. On the other hand, if we remove all the $2|S|$ elements of $S$ from $a$, the remaining subsequence is increasing since $S$ is maximal and thereby there is no inversion in the remaining elements. This implies that $\mathsf{DTM}(a) \leq 2|S|$.
\end{proof}

Based on Observation \ref{observation:obvious}, our 2-approximate algorithm maintains a maximal set of disjoint inversions with update time $O(\log^2 n)$.

\begin{lemma}
	There exists a $2$-approximate solution for $\mathsf{DTM}$ with update time $O(\log^2 n)$.
\end{lemma}
\begin{proof}
	Our algorithm maintains a maximal collection of disjoint inversion pairs, namely $S$. In addition to this, both the elements used in this collection and the elements not used in this collection are stored in a balanced tree\footnote{Red black tree could be one implementation.} that allows for search, insertion, and deletion in logarithmic time. We refer to the tree containing the elements of $S$ with $T_S$ and the tree containing other elements by $T_N$.
	
	Whenever a new element $a_i$ is inserted into the array, we first check if it makes an inversion with the elements of $T_N$. Notice that these elements are increasing in the order of their indices since there is no inversion between them. Thus, in order to verify whether $a_i$ makes an inversion with any element of $T_N$, we just need to compare that to the largest $a_j \in T_N$ which is smaller than  $a_i$ or the smallest $a_j \in T_N$ which is larger than $a_i$. Both of these two operations can be done in time $O(\log^2 n)$ (the exponent of $\log$ is 2 since it takes time $O(\log n)$ for us to get the value of the $i$'th element of the array). If an inversion is detected, we add it to $S$ and update $T_N$ and $T_S$ accordingly. Otherwise, we add $a_i$ into $T_N$.
	
	Removing an element is also strait-forward. If the element belongs to $T_N$, then no action is required other than updating $T_N$. Otherwise, after removing $a_i$, we have to be careful about the element which made an inversion with $a_i$ previously. That can be handled in time $O(\log^2 n)$ similar to adding new elements. We check if it makes an inversion with the elements of $T_N$ and update both $T_N$ and $T_S$ accordingly. All these operations can be done in time $O(\log^2 n)$.
\end{proof}

\subsection{From $2$-approximate Solution to an Exact Solution}\label{sec:exact_slow}
We show that a $2$-approximate solution for \textsf{DTM} can be used to obtain an exact solution. In fact, this idea carries over to any constant approximate solution but for simplicity we state it only for $2$-approximate solutions. Let us denote the size of the $2$-approximate solution by $k$. This way, we know that the optimal solution is in range $[k/2,k]$. 

We first construct a graph $G$ in the following way. Every element of the array becomes one vertex of $G$ and we put an edge between two vertices if their corresponding elements in $a$ make an inversion. This way, finding distance to monotonicity of array $a$ is equivalent to finding the smallest vertex cover of $G$. 

Let set $S$ be all the elements that are removed from the array in our $2$-approximate solution (thus we have  $|S| = k$).  We refer to the vertices of $G$ corresponding to set $S$ by $v_1$, $v_2$, $\ldots$, $v_k$. The key observation is that every edge of the graph is incident to at least one vertex $v_i$ otherwise $\cup_{i \in [k]} \{v_i\}$ would not be a valid vertex cover.

We call a vertex of the graph \textit{low-degree} if its degree is upper bounded by $k$ and \textit{high-degree} otherwise. Based on this, we divide the vertices corresponding to set $S$ into two disjoint sets $L$ and $H$ containing the low-degree and high-degree vertices. All vertices of $H$ have to be included in the optimal solution otherwise all their neighbors should be included and their number is more than $k$. Thus, we can include those vertices in our solution and remove them from the graph (this includes their incident edges too).

For each remaining edge of the graph, one end point is in $L$. Moreover, the degrees of the vertices in $L$ are bounded by $k$. Thus, the total number of remaining edges in the graph is bounded by $k^2$. Therefore, apart from at most $2k^2$ vertices, all the other vertices are isolated and definitely do not contribute to the vertex cover. Thus, we need to solve the problem for $O(k^2)$ many elements. This is equivalent to finding \textsf{DTM} for $O(k^2)$ many vertices which can be solved in time $O(k^2 \log n)$. There is an additional $O(\log n)$ overhead involved if access to each element requires time $O(\log n)$.

\begin{lemma}\label{lemma:dtm}
	Let $a$ be an array of length $n$ and $S$ be a set of $k$ elements whose removal from $a$ makes $a$ increasing. Provided oracle access to the elements of $a$, one can compute the distance to monotonicity of $a$ in time $O(k^2 \log n)$.
\end{lemma}
\begin{proof}
	The correctness of the algorithm is already discussed. Here we just show a bound on the runtime. Since set $S$ is given, we just need to compute the degree of each vertex. There are $k$ elements in set $S$ so detecting the edges between them can be done in time $O(k^2)$. Moreover, for every vertex $v_i$ corresponding to the elements of $S$, detecting its edges to the rest of the elements can also be done in time $O(\log n)$ since a binary search suffices for that purpose. Therefore, the total runtime is $O(k^2 \log n)$.
\end{proof}

\subsection{Exact Solution in Quasi-linear Time}
We show that the runtime of the algorithm of Section \ref{sec:exact_slow} can be improved to quasi-linear. Let us for simplicity divide the element of the array into two sets $S$ and $N$. Set $S$ corresponds to the elements of the approximate solution (whose removal makes the array increasing) and set $N$ contains the rest of the elements. Obviously, set $N$ is increasing otherwise $S$ would not be a valid solution to distance to monotonicity. The key to our improvement is the following observation:

\begin{observation}\label{observation:interval}
	Let $a_i < a_j < a_k$ be three elements of set $N$ and $a_l$ be an element of set $S$. If both $(a_i,a_l)$ and $(a_k,a_l)$ are inversions, then $(a_j,a_l)$ is also an inversion.
\end{observation}
\begin{proof}
	Notice that since we have $a_i < a_j < a_k$ and all three are in $N$, then we can infer that $i < j < k$. On the other hand, since both $(a_i,a_l)$ and $(a_k,a_l)$ are inversions, then either both $l < i$ and $a_l > a_k$ hold or both $l > k$ and $a_l < a_i$ hold. In either case, $(a_j,a_l)$ makes an inversion.
\end{proof}

Observation \ref{observation:interval} shows that any element of $S$ makes an inversion with an interval of the elements in $N$. This implies an important consequence: Label each element $a_i \in N$ with a set of elements $I(a_i) \subseteq S$ such that for each element $a_j \in I(a_i)$, pair $(a_i,a_j)$ makes an inversion. Based on Observation \ref{observation:interval} we can prove that the total number of distinct labels in $N$ is bounded by $2|S|+1$.

\begin{lemma}\label{lemma:key}
	For each element $a_i \in N$, define its label by $I(a_i) \subseteq S$ where $I(a_i)$ contains all elements of $S$ that make an inversion with $a_i$. Then we have:
	$$|\bigcup_{a_i \in N} \{I(a_i)\}| \leq 2|S|+1.$$
\end{lemma}
\begin{proof}
	For each element $a_j \in S$ that makes an inversion with an element of $N$ define two thresholds $\alpha$ and $\beta$ where $\alpha$ is the smallest element of $N$ that makes an inversion with $a_j$ and $\beta$ is the smallest element of $N$ larger than $\alpha$ that does not make an inversion with $a_j$. Due to Observation \ref{observation:interval}, an element of $N$ makes an inversion with $a_j$ if and only if its value is at least $\alpha$ and smaller  than $\beta$. The total number of thresholds for all elements of $S$ is at most $2|S|$. Moreover if two elements of $N$ are not separated by any threshold, then their labels are the same. Thus, the total number of distinct labels is bounded by $2|S|+1$.
\end{proof}
\input{figs/elements}

Lemma \ref{lemma:key} is important from an algorithmic point of view because of the following: if two elements have the same label (and thus the same set of conflicting elements in $S$), either they both contribute to the optimal solution (removed from the array) or they both remain in the array.  Thus, one can merge these elements and make a larger element in the array that represents both of them. More generally, for each label, if we merge all the element attributed to that label and make a single element out of them (with a larger size), the size of the optimal solution remains unchanged. Since the total number of labels is bounded by $2|S|+1$ then this transformation leaves us with $2|S|+1$ elements and we only need to solve the problem for them. This can be done in time $O(k \log n)$ as shown in Lemma \ref{lemma:linear}. We note that in the above, we assume random access to the elements of the array is provided in time $O(1)$. 

\begin{lemma}\label{lemma:linear}
	Given query access to an array $a$ of length $n$ and a  $2$-approximate solution for distance to monotonicity of $a$ with size $k$, one can find an exact solution for distance to monotonicity in time $O(k \log n)$. 
\end{lemma}
\begin{proof}
	The algorithm and its correctness is outlined above. Here we just explain the runtime. Since the elements of $N$ are stored in a balanced tree data structure, for each element of $S$ we can find in time $O(\log n)$ which interval of the elements of $N$ makes an inversion with it. This takes a total runtime $O(k \log n)$. Next, we sort all the intervals based on them and merge elements of $N$ whose labels are the same. We refer to these merged elements as super elements. The value of a super element can be equal to the value of an arbitrary element which is used in its construction.
	
	After constructing the super elements in time $O(k \log n)$ the size of the problem reduces to $O(k)$. However, patience sorting does not solve this problem since super elements have weights. In other words, removing a super element incurs a cost equal to the number of elements used to make it. Nonetheless, it is known~\cite{jacobson1992heaviest} that even if the elements are weighted, for an array of size $k$, one can solve both \textsf{LIS} and \textsf{DTM} in time $O(k \log k)$.
\end{proof}

\subsection{$1+\epsilon$ Approximation for \textsf{DTM} with Update Time $O(\log ^2n)$ }
The last step is to obtain a $1+\epsilon$ approximate solution using the above techniques. In parallel, we always run the algorithm of Section \ref{sec:2-approx} to maintain a 2-approximate solution. In order to obtain a $1+\epsilon$ approximation algorithm, we use the framework of Section \ref{sec:amortized}. To design our block-based algorithm we start with an array of length $n$. We set the preprocessing time to $O(k \log n)$ where $k$ is the size of the 2-approximate solution (this we know by the parallel algorithm that we run). In the preprocessing phase, we find an exact solution (say $k' \geq k/2$) in time $O(k \log^2 n)$\footnote{The additional $O(\log n)$ factor is due to the data structure used for random access to the elements of the array.} for the array and report it as the value of distance to monotonicity. We set $g(a) = \epsilon k/2$ and for the next $\epsilon k/2$ steps, we report $d+i$ for the $i$'th operation where $d$ is the solution for the initial array. This is clearly an upper bound on the size of the solution as well as a $1+\epsilon$ approximation of it. Thus, the worst-case update time is $O(1)$.  

By Lemma \ref{lemma:trivial}, our block-based algorithm can be turned into a dynamic algorithm with worst-case update time $O(\log^2 n)$.

\begin{theorem}\label{theorem:dtm}
	For any $\epsilon > 0$, there exists a dynamic algorithm for distance to monotonicity with approximation factor $1+\epsilon$ and worst-case update time $O(\log^2 n)$.
\end{theorem}
\begin{proof}
	For an array $a$, our block-based algorithm has preprocessing time $f(a) = O(k \log^2 n)$ where $k = \mathsf{DTM}(a)$. Also, $g(a) = \epsilon k/2$ and $h(a) = O(1)$. Thus, by Lemma \ref{lemma:reduction}, the worst-case update time of the equivalent dynamic algorithm is $O(\log^2 n)$. Notice that since after $g(a)$ steps the size of the $\mathsf{DTM(a)}$ changes by a small factor, all functions $f,g,$ and $h$ remain asymptotically the same which implies relativity.
\end{proof}

\newpage

\bibliographystyle{abbrv}	
\bibliography{draft}

\newpage
\appendix
\newpage
\section{Streaming Algorithm for \textsf{LIS}}\label{sec:lisstreaming}
We outlined an improved streaming algorithm for \textsf{LIS} in Section \ref{sec:results}. Here we give a proof for its correctness and a bound on its memory. In this setting, we assume that the input is available to us in any desired order. We call this setting, streaming with advisory help. Previous work solves the problem with memory $O(\sqrt{n})$ within a factor $1+\epsilon$ in a single round~\cite{DBLP:conf/soda/GopalanJKK07}.

\begin{observation}
    For any $0 < \kappa < 1$, there exists a streaming algorithm that uses advisory help and approximates the \textsf{LIS} of an array of length $n$ with memory $\tilde O(n^{2/5+\kappa})$ within a factor $O(1/\kappa)$ in three rounds. This algorithm is randomized and gives an approximate solution with probability at least $1-1/n$.
\end{observation}
\begin{proof}
Let $m = n^{1/5}$ be the size of the grid. We use an $(O(m^\kappa \log m),O(1/\kappa))$ approximate solution of grid packing to cover the grid cells with segments. Similar to what we did before, we think of each element $i$ of the array as a point $(i,a_i)$ of the 2D plane. As mentioned earlier, in the first round, we sample $m-1$ elements from the array. We use these elements to draw the horizontal lines of the grid. It follows from standard Chernoff bound that since the elements are chosen uniformly at random, the number of elements in every row of the grid is bounded by $10 n/m \log n$ with probability at least $1-1/n$. Also, we draw the vertical lines evenly so that they divide the elements into chunks of size $n/m$.

In the next two rounds, we ask for the elements of the array in the row-order and column-order. More precisely, in the second round, we first ask for the elements falling in the first row of the grid (in the column order). Next, we ask for the elements of the second row and so on. For each row, we approximate the value of the \textsf{LIS} for each segment. Since the total number of elements in every row is bounded by 
$10 n/m \log n$, we only need memory $\tilde O(\sqrt{n/m \log n}) = \tilde O(n^{2/5})$ to approximate the value of \textsf{LIS} for each segment. However, we need this much memory for multiple segments. This adds an overhead of $\tilde O(m^\kappa)$ to the memory since each grid cell may be covered by at most $\tilde O(m^\kappa)$ segments. Similarly, in the third round, we ask for the elements in the column-order.

Finally, we use a DP to find a subset of non-conflicting segments with the largest total sum of \textsf{LIS}. This can be done with memory $\tilde O(m^{2+\kappa}) = \tilde O(n^{2/5+\kappa})$ as the number of segments is bounded by $\tilde O(m^{2+\kappa})$.

The correctness of the algorithm is similar to the one given for dynamic \textsf{LIS}. We fix an arbitrary \textsf{LIS} of the array and assume that the adversary puts the contribution of each cell of the grid to the fixed \textsf{LIS}. The \textsf{LIS} of each segment is a clear upper bound on the score of that segment on the grid, however, if we use those values instead of their scores, we still obtain a valid solution. Finally, since our solution for the grid packing problem is $(O(m^\kappa \log m),O(1/\kappa))$ approximate, then the score we obtain using non-conflicting segments is at least an $\Omega(\kappa)$ fraction of the score of the grid which is equal to the size of the \textsf{LIS}.
\end{proof}
\newpage
\section{Improved Algorithm for $\mathsf{LIS}^+$}\label{sec:lisplus}
We provide the intuition behind the algorithm and then the formal proof. For simplicity, we assume here that we have random access to all elements of the array in time $O(1)$. Our first goal is to design a block-based algorithm for an array $a$  of length $n$. We set $f(n) = O(n \log n)$ and in the preprocessing phase we compute the \textsf{LIS} of $a$. Let this value be $x$. Since we only add elements in the $\mathsf{LIS}^+$ problem, from here on, $x$ is a lower bound for the solution value. For the next $g(n) = \sqrt{n}$ operations, every new element is added to a separate set, and after each operation, the \textsf{LIS} of the separate set is computed in time $O(\sqrt{n} \log n)$. At an arbitrary step, let this value be $y$. The key observation is that the overall \textsf{LIS} of the array is in range $[\max\{x,y\},x+y] \leq 2\max\{x,y\}$. We therefore have a 2-approximate solution by just reporting $\max\{x,y\}$.  This block-based algorithm for $\lisplus$ with $f(n) = O(n \log n)$, $g(n) = \sqrt{n}$, and $h(n) = O(\sqrt{n} \log n)$ yields a dynamic algorithm with worst-case update time $\tilde O(\sqrt{n})$.

We recurse to improve the runtime down to $O(n^\epsilon)$ for any $\epsilon > 0$. Instead of using the naive algorithm for the $g(n)$ operations after the initialization, we can use the more advanced algorithm explained above (its update time is $\tilde \Omega(\sqrt{n})$ better than computing the \textsf{LIS} from the scratch every time).
%\Michael{What more advanced algorithm?  Cite a theorem or lemma!}\Saeed{Added more details.}
Similar to what we did in \ref{sec:constant}, in each level of recursion, we use the previous algorithm for the operations after initialization. The advantage of this approach in this setting over Section \ref{sec:constant} is that the approximation factor of the algorithm depends linearly on $1/\epsilon$ (and not exponentially). To see this, assume that at some point, $x$ is the solution for the initial array and $y$ is an $\alpha$-approximation for the solution of the second set of operations. The optimal solution is upper bounded by $x+\alpha y$ and lower bounded by $\max\{x,y\}$. Therefore, by reporting $\max\{x,y\}$ we can be sure that our approximation factor is bounded by $\alpha + 1$. 
%\Michael{The following sentence doesn't make sense to me.  If we recurse on what algorithm $\alpha$ times?  We just got an algorithm of $\alpha +1$ in the previous sentence doing this one time.  What does it even mean to recurse on the algorithm $\alpha$ times?  Is $\alpha$ restricted to be a positive intereger, this was never stated?}\Saeed{Added more details.}
In other words, if we recurse on this algorithm $\alpha$ times, then the approximation factor is bounded by $\alpha + 1$. 
%\Michael{Why does this next sentence follow from the followign sentence?  What is the relationship of $\alpha$ and $\epsilon$.  Where did this running time come from?}
This results in an algorithm with worst-case update time $\tilde O(n^{\epsilon})$ and approximation factor $O(1/\epsilon)$.

\begin{observation}\label{obs:append}
For any constant $\epsilon > 0$, there exists an algorithm for dynamic $\lisplus$ whose worst-case update time is $\tilde O(n^{\epsilon})$ and whose approximation factor is $O(1/\epsilon)$.
\end{observation}

If one could show the statement of Observation \ref{obs:append} for any (possibly sub-constant) $\epsilon > 0$, then by setting $\epsilon = 1/\log n$, we could obtain a dynamic algorithm for $\lisplus$ with polylogarithmic update time and logarithmic approximation factor. However, since there is a constant factor overhead in every recursion, Observation \ref{obs:append} only works when $\epsilon$ is constant. This overhead is incurred in the reduction from dynamic algorithms to block-based algorithms. In the following, we provide a variation of the same algorithm that does not use this reduction and achieves polylogarithmic update with and logarithmic approximation.

\input{figs/lisplus}

\begin{theorem}
There exists an algorithm for dynamic $\lisplus$ whose worst-case update time is $O(\log^3 n)$ and whose approximation factor is $O(\log n)$.
\end{theorem}
\begin{proof}
In our algorithm, we put the elements in buckets and for every bucket, we compute the value of the \textsf{LIS}. As more buckets are made, we combine them to construct larger buckets. The sizes of the buckets are always powers of 2.
%\Michael{The above sentence suggests we only ever put 1 element in a bucket.  Do you mean we initially put each element in a separate bucket, combine buckets as we  go, and  keep the value of the LIS for every bucket?}\Saeed{Revised the sentence.}

In the beginning, the array is empty and there are no buckets.
When the first element is inserted, we construct the first bucket that contains only that element. After the construction of each bucket, we compute the \textsf{LIS} of that bucket over the next steps. 
%\Michael{Below -- you haven't described how buckets  are combined at all, so how do we even get a bucket of size $k$?}\Saeed{Fixed.}
More precisely, when a bucket of size $k$ is constructed, we divide the task of computing the \textsf{LIS} of that bucket into $k$ pieces and execute these pieces in the next $k$ operations. Thus, when a bucket of size $1$ is constructed, its \textsf{LIS} is computed immediately. We say a bucket is \textit{finalized}, when our algorithm has already computed its \textsf{LIS}. In our algorithm, we only merge finalized buckets to make larger ones and thus, we maintain the property that at each point in time, each element appears in exactly one finalized bucket.
%\Michael{Can buckets be interrupted -- that is merged before we comput  their LIS?}\Saeed{Added more details.}

Every time a new element is inserted, we make a bucket containing that element alone. However, when there are two finalized buckets of the same size (say $k$), we merge them to obtain a bucket of size $2k$. After this, it takes $2k-1$ more steps to finalize the new bucket but once the new bucket is finalized, we remove the two smaller buckets. This way, each element appears in exactly one finalized bucket at a time throughout the process.

At any point in time, we approximate the \textsf{LIS} of the array by the maximum solution for any of the finalized buckets. We prove that with this construction, there are at most $O(\log n)$ finalized buckets at every point in time and moreover, the number of buckets that are not finalized is also bounded by $O(\log n)$. 

This immediately implies that the approximation factor of our algorithm is bounded by $O(\log n)$. The reason is that each element is always included in exactly one finalized bucket at a time and therefore the total sum of \textsf{LIS}'s for all finalized buckets is an upper bound on the size of the solution. 
%\Michael{I don't get this.  What about unfinalized buckets?  Couldn't elements from  unfinalized buckets take part?  Why is total sum just from finalized buckets an  upper bound?  Maybe you're saying unfinalized buckets only contain REPEATS of elements already in finalized buckets.  If  that's the case, please state that clearly and explicitly earlier!}
Moreover, the maximum solution size for each bucket is a lower bound on the \textsf{LIS} of the entire array. Since the number of finalized buckets is bounded by $O(\log n)$ this implies that the approximation factor of our algorithm is bounded by $O(\log n)$.

We also bound the runtime by $O(\log^3 n)$. In our algorithm, at every point in time, there are $O(\log n)$ different buckets that are not finalized yet. The total runtime needed to compute the \textsf{LIS} of a bucket of size $k$ is $O(k \log k)$ which is divided over $k$ steps. Thus, each bucket which is not finalized yet requires time at most $O(\log n)$ for each step. Thus the overall runtime is $O(\log^2 n)$ for each step. One thing to keep in mind is that we use a balanced tree data structure to access the elements of the array which adds an overhead of $O(\log n)$. Therefore, the worst-case update time is $O(\log^3 n)$.

It follows from the construction of the buckets that at each step, there is at most one bucket of each size which is not finalized. The reason is that after a bucket of size $k$ is made, it takes $k$ more steps to make another bucket of the same size. However, before the new bucket is made, the first one will be finalized. Since the sizes of the buckets are powers of $2$, this implies that there are most $\lfloor \log n\rfloor$ such buckets. Moreover, this also shows that the number of finalized buckets of each size is bounded by $3$, otherwise this makes two buckets of larger size that are not finalized yet.
\end{proof}
\newpage
\section{Sequential Algorithm for \textsf{DTM}}\label{sec:dtmclassic}
We present a simple comparison-based algorithm for \textsf{DTM} with approximation factor $2$ that runs in time $O(n)$. Using Lemma \ref{lemma:key}, the approximation factor improves to $1+\epsilon$ in the following way: If the solution size is bounded by $\sqrt{n}$, then Lemma \ref{lemma:key} gives an exact solution in time $\tilde O(\sqrt{n})$. Otherwise, one can find a $1+\epsilon$ approximate solution in time $\tilde O(\sqrt{n})$ using the solution of~\cite{DBLP:conf/soda/NaumovitzSS17}. Thus, the main bottleneck of the runtime is for the computation of an approximate solution, which we show can be done in time $O(n)$.

By Observation~\ref{observation:obvious}, a 2-approximate solution can be obtained by computing a maximal set of inversion pairs. Our algorithm finds such a set in linear time.

We begin by a empty stack. We iterate over the elements of the array and each time we compare the new element to the element at the top of the stack (if any). If this pair makes an inversion, we put this pair in a set $S$ and remove the last element of the stack. Otherwise, we put the new element on top of the stack and continue on. Obviously, the runtime is $O(n)$, since each element $a_i$ is processed in time $O(1)$. The correctness of the algorithm follows from the fact that the numbers of the stack are always increasing and therefore there is no inversion between them. Thus, set $S$ is a maximal set of inversion pairs.

\begin{observation}
For any constant $\epsilon > 0$, \textsf{DTM} can be approximated within a factor $1+\epsilon$ in time $O(n)$.
\end{observation}

We remark that in the above observation, factors that depend on $1/\epsilon$ are hidden in the $O$ notation.

\newpage
\section{The Algorithm of Chen \textit{et al.}~\cite{chen2013dynamic}}\label{sec:chen-ap}
For formal proofs, we refer the reader to~\cite{chen2013dynamic}. Chen \textit{et al.}~\cite{chen2013dynamic} propose the following algorithm to maintain a solution for dynamic \textsf{LIS}.

For each element $i$ of the array, define $l(i)$ to be the size of the longest increasing subsequence ending at element $a_i$ of the array. Chen \textit{et al.}~\cite{chen2013dynamic} refer to this quantity as the \textit{level} of element $i$. Notice that $l(i)$ can be computed in time $\tilde O(n)$ for all elements of the array using the patience sorting algorithm.

Define $L_k$ to be the set of elements whose levels are equal to $k$. The algorithm of Chen \textit{et al.}~\cite{chen2013dynamic} maintains a balanced binary tree for each $L_k$ that contains the corresponding elements. One key observation is that for each $k$, all the elements  of $L_k$ are decreasing, otherwise their levels would not be the same.

When a new element is added to the array, $L_k$'s may change. More precisely, after an element addition, the levels of some elements may change (but only by 1). Similarly, element removal may change the levels of the elements of the array but again the change is bounded by $1$. Chen \textit{et al.}~\cite{chen2013dynamic}, show that after an insertion, for each $L_k$, the levels of only one interval of the elements may increase. In other words, for each $L_k$, there are two numbers $\alpha$ and $\beta$ such that all the elements whose values are within $[\alpha,\beta]$ increase their levels and the rest remain in $L_k$.

Thus, they use a special balanced tree structure that allows for interval deletion and interval addition in logarithmic time. Therefore, all that remains is to detect which interval of each $L_k$ changes after each operation. They show that this can be computed in time $O(\log n)$ for all $L_k$'s via binary search. Since the number of different levels is equal to the size of the \textsf{LIS}, their update time depends on the size of the solution.

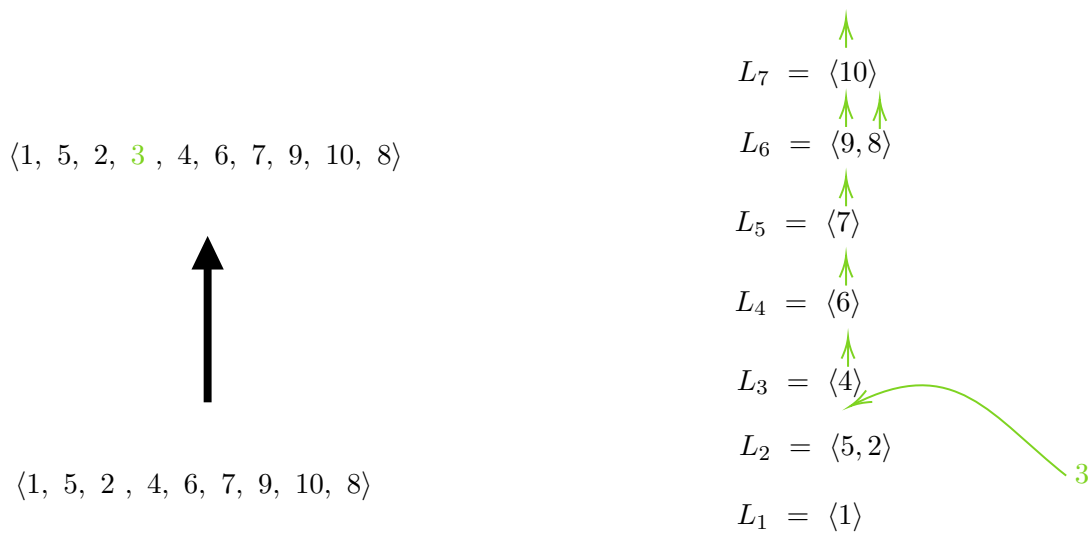
\begin{figure}[ht]

\centering

\tikzset{every picture/.style={line width=0.75pt}} %set default line width to 0.75pt        

\begin{tikzpicture}[x=0.75pt,y=0.75pt,yscale=-1,xscale=1]

%Straight Lines [id:da6222791956949012] 
\draw [line width=3]    (163,205) -- (163,126) ;
\draw [shift={(163,121)}, rotate = 450] [fill={rgb, 255:red, 0; green, 0; blue, 0 }  ][line width=3]  [draw opacity=0] (16.97,-8.15) -- (0,0) -- (16.97,8.15) -- cycle    ;

%Curve Lines [id:da6417752694113028] 
\draw [color={rgb, 255:red, 126; green, 211; blue, 33 }  ,draw opacity=1 ]   (596,242) .. controls (558.38,212.3) and (540.36,179.66) .. (488.58,206.18) ;
\draw [shift={(487,207)}, rotate = 332.15] [color={rgb, 255:red, 126; green, 211; blue, 33 }  ,draw opacity=1 ][line width=0.75]    (10.93,-3.29) .. controls (6.95,-1.4) and (3.31,-0.3) .. (0,0) .. controls (3.31,0.3) and (6.95,1.4) .. (10.93,3.29)   ;

%Straight Lines [id:da4868882358205895] 
\draw [color={rgb, 255:red, 126; green, 211; blue, 33 }  ,draw opacity=1 ][fill={rgb, 255:red, 126; green, 211; blue, 33 }  ,fill opacity=1 ]   (486,187) -- (486,174) ;
\draw [shift={(486,172)}, rotate = 450] [color={rgb, 255:red, 126; green, 211; blue, 33 }  ,draw opacity=1 ][line width=0.75]    (10.93,-3.29) .. controls (6.95,-1.4) and (3.31,-0.3) .. (0,0) .. controls (3.31,0.3) and (6.95,1.4) .. (10.93,3.29)   ;

%Straight Lines [id:da8338957340901085] 
\draw [color={rgb, 255:red, 126; green, 211; blue, 33 }  ,draw opacity=1 ][fill={rgb, 255:red, 126; green, 211; blue, 33 }  ,fill opacity=1 ]   (485,146) -- (485,133) ;
\draw [shift={(485,131)}, rotate = 450] [color={rgb, 255:red, 126; green, 211; blue, 33 }  ,draw opacity=1 ][line width=0.75]    (10.93,-3.29) .. controls (6.95,-1.4) and (3.31,-0.3) .. (0,0) .. controls (3.31,0.3) and (6.95,1.4) .. (10.93,3.29)   ;

%Straight Lines [id:da2688320646501736] 
\draw [color={rgb, 255:red, 126; green, 211; blue, 33 }  ,draw opacity=1 ][fill={rgb, 255:red, 126; green, 211; blue, 33 }  ,fill opacity=1 ]   (485,106) -- (485,93) ;
\draw [shift={(485,91)}, rotate = 450] [color={rgb, 255:red, 126; green, 211; blue, 33 }  ,draw opacity=1 ][line width=0.75]    (10.93,-3.29) .. controls (6.95,-1.4) and (3.31,-0.3) .. (0,0) .. controls (3.31,0.3) and (6.95,1.4) .. (10.93,3.29)   ;

%Straight Lines [id:da7076928865139065] 
\draw [color={rgb, 255:red, 126; green, 211; blue, 33 }  ,draw opacity=1 ][fill={rgb, 255:red, 126; green, 211; blue, 33 }  ,fill opacity=1 ]   (485,66) -- (485,53) ;
\draw [shift={(485,51)}, rotate = 450] [color={rgb, 255:red, 126; green, 211; blue, 33 }  ,draw opacity=1 ][line width=0.75]    (10.93,-3.29) .. controls (6.95,-1.4) and (3.31,-0.3) .. (0,0) .. controls (3.31,0.3) and (6.95,1.4) .. (10.93,3.29)   ;

%Straight Lines [id:da8205708394730371] 
\draw [color={rgb, 255:red, 126; green, 211; blue, 33 }  ,draw opacity=1 ][fill={rgb, 255:red, 126; green, 211; blue, 33 }  ,fill opacity=1 ]   (502,67) -- (502,54) ;
\draw [shift={(502,52)}, rotate = 450] [color={rgb, 255:red, 126; green, 211; blue, 33 }  ,draw opacity=1 ][line width=0.75]    (10.93,-3.29) .. controls (6.95,-1.4) and (3.31,-0.3) .. (0,0) .. controls (3.31,0.3) and (6.95,1.4) .. (10.93,3.29)   ;

%Straight Lines [id:da1772881215319777] 
\draw [color={rgb, 255:red, 126; green, 211; blue, 33 }  ,draw opacity=1 ][fill={rgb, 255:red, 126; green, 211; blue, 33 }  ,fill opacity=1 ]   (485,26) -- (485,13) ;
\draw [shift={(485,11)}, rotate = 450] [color={rgb, 255:red, 126; green, 211; blue, 33 }  ,draw opacity=1 ][line width=0.75]    (10.93,-3.29) .. controls (6.95,-1.4) and (3.31,-0.3) .. (0,0) .. controls (3.31,0.3) and (6.95,1.4) .. (10.93,3.29)   ;

% Text Node
\draw (462,263) node   {$L_{1} \ =\ \langle 1\rangle$};
% Text Node
\draw (470,228) node   {$L_{2} \ =\ \langle 5,2\rangle$};
% Text Node
\draw (462,195) node   {$L_{3} \ =\ \langle 4\rangle$};
% Text Node
\draw (162,81) node   {$\langle 1,\ 5,\ 2,\ \color{rgb, 255:red, 126; green, 211; blue, 33 }3\color{black}\ ,\ 4,\ 6,\ 7,\ 9,\ 10,\ 8\rangle$};
% Text Node
\draw (156,247) node   {$\langle 1,\ 5,\ 2\ ,\ 4,\ 6,\ 7,\ 9,\ 10,\ 8\rangle$};
% Text Node
\draw (461,155) node   {$L_{4} \ =\ \langle 6\rangle$};
% Text Node
\draw (461,115) node   {$L_{5} \ =\ \langle 7\rangle$};
% Text Node
\draw (470,75) node   {$L_{6} \ =\ \langle 9,8\rangle$};
% Text Node
\draw (466,39) node   {$L_{7} \ =\ \langle 10\rangle$};
% Text Node
\draw (604,241) node [color={rgb, 255:red, 126; green, 211; blue, 33 }  ,opacity=1 ]  {$3$};

\end{tikzpicture}

\caption{This examle shows how adding element $3$ to the array changes the levels of the elements. Upward arrows show that the level of the corresponding element increases after we add $3$ to the array.} \label{fig:chan}
\end{figure}

When $n$ elements are given, their runtime for constructing the data structure is $\tilde O(n)$ since patience sorting gives us all the levels in time $\tilde O(n)$ and the balanced trees can be constructed in time $\tilde O(n)$ for all $L_k$.

\end{document}